\documentclass[english]{article}

\usepackage[T1]{fontenc}
\usepackage[latin9]{inputenc}

\usepackage[a4paper,
            bindingoffset=0.2in,
            left=1.0in,
            right=1.0in,
            top=1.0in,
            bottom=1.0in,
            footskip=.25in]{geometry}

\usepackage{environ}

\usepackage{url}

\makeatletter
\newsavebox{\measure@tikzpicture}
\NewEnviron{scaletikzpicturetowidth}[1]{%
  \def\tikz@width{#1}%
  \begin{lrbox}{\measure@tikzpicture}%
  \BODY
  \end{lrbox}%
  \pgfmathparse{#1/\wd\measure@tikzpicture}%
  \BODY
}
\makeatother

\usepackage{csquotes}
\usepackage{nicefrac}
\usepackage{algorithm,algpseudocode}
\algnewcommand{\LineComment}[1]{\State \(\triangleright\) #1}
\usepackage[nocompress]{cite}
\usepackage{mathtools}
\usepackage{xspace}
\usepackage{paralist}
\usepackage[usenames,dvipsnames,svgnames,table,x11names]{xcolor}
\usepackage{scalerel}

\usepackage{amssymb,mathrsfs}
\usepackage{pifont}
\newcommand{\cmark}{\ding{51}}%

\usepackage{tikz}
\usetikzlibrary{shapes.geometric}
\usetikzlibrary{circuits.logic.US,circuits.logic.IEC}

\usetikzlibrary{arrows,calc,positioning,arrows.meta}
\definecolor{xdxdff}{rgb}{0.49019607843137253,0.49019607843137253,1.}
\definecolor{qqqqcc}{rgb}{0.,0.,0.8}

\colorlet{LightOrange}{orange!60!white}
\colorlet{DarkOrange}{orange!40!black!60}
\colorlet{LightViolet}{DarkViolet!60!white}
\colorlet{DarkViolet}{DarkViolet!40!black!60}
\colorlet{LightBlue}{blue!60!white}
\colorlet{DarkBlue}{blue!40!black!60}
\colorlet{LightGreen}{green!60!white}
\colorlet{DarkGreen}{green!40!black!60}

\newenvironment{absolutelynopagebreak}
  {\par\nobreak\vfil\penalty0\vfilneg
   \vtop\bgroup}
  {\par\xdef\tpd{\the\prevdepth}\egroup
   \prevdepth=\tpd}

\usetikzlibrary{calc}
\usetikzlibrary{arrows}
\usetikzlibrary{positioning}

\definecolor{lightergray}{rgb}{0.8,0.8,0.8}
\definecolor{verylightgray}{rgb}{0.95,0.95,0.95}
\definecolor{verylightblue}{rgb}{0.925,0.95,1.0}
\definecolor{lightred}{rgb}{1.0,0.85,0.85}

\newcommand{\floornlogn}{\left(\frac{n}{\log n}\right)}

\newcommand{\lmm}{\mathsf{Y}_{d}(n,p)}

\newcommand{\BCC}{\mathcal{B}}

\newcommand{\ECC}{E}
\newcommand{\FCC}{F}

\newcommand{\HCC}{H}
\newcommand{\ICC}{\mathcal{I}}

\newcommand{\KCC}{{K}}
\newcommand{\LCC}{{L}}
\newcommand{\MCC}{\mathcal{M}}

\newcommand{\PCC}{\mathcal{P}}

\newcommand{\SCC}{\mathcal{S}}

\newcommand{\UCC}{\mathcal{U}}
\newcommand{\VCC}{\mathcal{V}}

\newcommand{\HKCC}{\HCC_\KCC}

\newcommand{\VSS}{\mathcal{V}}

\newcommand{\bbz}{\mathbb{Z}}
\newcommand{\bbn}{\mathbb{N}}

\newcommand{\bbr}{\mathbb{R}}
\newcommand{\bba}{\mathbb{A}}

\newcommand{\cfm}{\mathsf{X}(n,\textbf{p})}
\newcommand{\boldp}{\textbf{p}}

\newcommand{\bfj}{\mathbf{j}}

\newcommand{\kc}{K(C)}
\newcommand{\kprime}{K'(C)}

\DeclareMathOperator{\conv}{conv}

\newcommand{\satgates}{\mathscr{S}}
\newcommand{\unsatgates}{\mathscr{\overline{S}}}

\newcommand{\tildev}{\tilde{\VCC}}
\newcommand{\tildef}{\tilde{F}}

\newcommand{\weirderfactor}{2^{\log^{(1-\epsilon)}n}}

\newcommand{\erk}{\textnormal{er}(\KCC)}

\newcommand{\opa}{\textnormal{OPT}_\textnormal{A}}
\newcommand{\opb}{\textnormal{OPT}_\textnormal{B}}
\newcommand{\opt}[1]{\textnormal{OPT}_{#1}}

\newcommand{\jin}{j \in [1,n]}

\usepackage{amsmath,amsthm}

\usepackage{xr}
\usepackage{hyperref}
\usepackage{cleveref}

\theoremstyle{plain}
\newtheorem{theorem}{Theorem}[section]

\newtheorem*{theorem*}{Theorem}

\newtheorem{claim}{Claim}[section]
\newtheorem{lemma}[theorem]{Lemma}
\newtheorem{proposition}[theorem]{Proposition}
\newtheorem{corollary}{Corollary}

\theoremstyle{definition}
\newtheorem{definition}{Definition}[section]

\theoremstyle{remark}
\newtheorem{remark}{Remark}

\newtheorem{notation}{Notation}

\newcommand{\fpt}{{\bf FPT}\xspace}
\newcommand{\NP}{{\bf NP}\xspace}

\newcommand{\QP}{{\bf QP}\xspace}
\newcommand{\WP}{{\bf W{[P]}}\xspace}

\newcommand{\khat}{\hat{K}}

\newcommand{\khatmax}{\khat_{\max}}
\newcommand{\brat}{(2)}
\newcommand{\kone}{K^1}

\newcommand{\dcomplex}[2]{{#1}^{(#2)}}
\newcommand{\twocomplex}[1]{\dcomplex{#1}{2}}

\newcommand{\alphastrings}{\Sigma^{\ast}}

\newcommand{\OPTA}{\opt\mincircuitsat(C)}
\newcommand{\OPTB}{\opt\minrmm(\kc)}

\newcommand{\madk}{\widehat{K}}
\newcommand{\madv}{\widehat{\VCC}}

\newcommand{\OPTC}{\opt\minrmm(K)}
\newcommand{\OPTD}{\opt\minmm(K)}
\newcommand{\OPTDD}{\opt\minmm(\widehat{K})}

\newcommand{\ALGA}{m_{\mincircuitsat} (C,\ICC(C,\VCC))}
\newcommand{\ALGAA}{m_{\mincircuitsat} (C,\ICC(C,\tilde{\VCC}))}
\newcommand{\ALGB}{m_{\minrmm} (\kc,\VCC)}
\newcommand{\ALGBB}{m_{\minrmm} (\kc,\tilde{\VCC})}

\DeclareMathOperator*{\im}{im}

\newcommand{\wcs}{\textsc{Weighted Circuit Satisfiability}\xspace}
\newcommand{\minwcs}{\textsc{Min-Weighted Circuit Satisfiability}\xspace}

\newcommand{\paraexptwo}{\textsc{Erasability Expansion Height}\xspace}

\newcommand{\paraexpert}{\textsc{Expansion Height}\xspace}

\newcommand{\minmonotone}{\textsc{Min-Monotone Circuit Sat}\xspace}

\newcommand{\minmorse}{\textsc{Min-Morse Matching}\xspace}

\newcommand{\erasability}{\textsc{Erasability}\xspace}

\newcommand{\minrmorse}{\textsc{Min-Reduced Morse Matching}\xspace}

\newcommand{\maxmorse}{\textsc{Max-Morse Matching}\xspace}

\newcommand{\minmm}{\mathsf{MinMM}\xspace}

\newcommand{\minrmm}{\mathsf{MinrMM}\xspace}

\newcommand{\maxmm}{\mathsf{MaxMM}\xspace}

\newcommand{\wcsabbr}{\mathsf{WCS}\xspace}
\newcommand{\minwcsabbr}{\mathsf{MinWCS}\xspace}
\newcommand{\mincircuitsat}{\minwcsabbr[\mathcal{C}^{+}]}
\renewcommand{\mincircuitsat}{\mathsf{MinMCS}\xspace}

\newcommand{\ijpair}{(i,j)}
\newcommand{\qjpair}{(q,j)}
\newcommand{\pjpair}{(p,j)}
\newcommand{\kppair}{(k,p)}

\newcommand{\lppair}{(\ell,p)}
\newcommand{\jqpair}{(j,q)}

\newcommand{\ionepair}{(i,1)}

\newcounter{problemcounter}
\newcommand{\gadget}{\mathbf{D}}

\newcommand{\maingadgetx}{\gadget_{m,\ell}^{\zeta_1}}
\newcommand{\altgadgetx}{\gadget_{m,\ell}^{\zeta_2}}

        \newtheorem{case}{Case}

        \theoremstyle{nonumberplain}

\usepackage{todonotes}   

\begin{document}

\title{Parameterized inapproximability of Morse matching} 




\author{Ulrich Bauer \\   \texttt{ulrich.bauer@tum.de}  \and \and Abhishek Rathod \\   \texttt{arathod@purdue.edu} }

\maketitle

\begin{abstract}

We study the problem of minimizing the number of critical simplices from the point of view of inapproximability and parameterized complexity.
We first show inapproximability of \minmorse within a factor of $\weirderfactor$.
Our second result shows that \minmorse is $\WP$-hard with respect to the standard parameter. Next, we show that \minmorse with standard parameterization has no FPT approximation algorithm for \emph{any} approximation factor $\rho$.
The above hardness results are applicable to complexes of dimension $\ge 2$.

On the positive side, we provide a factor $O(\frac{n}{\log n})$ approximation algorithm for \minmorse on $2$-complexes, noting that no such algorithm is known for higher dimensional complexes.
Finally, we devise discrete gradients with very few critical simplices for typical instances drawn from a fairly wide range of parameter values of the Costa--Farber model of random complexes.


\end{abstract}

\section{Introduction}

Classical Morse theory~\cite{Milnor} is an analytical tool for studying topology of smooth manifolds. Forman's discrete Morse theory is a combinatorial analogue of Morse theory that is applicable to simplicial complexes, and more generally regular cell complexes~\cite{Fo98}. 
In Forman's theory, discrete Morse functions play the role of smooth Morse functions, whereas discrete gradient vector fields are the analogues of gradient-like vector fields.
The principal objects of study are, therefore, the so-called discrete gradient vector fields (or discrete gradients) on simplicial complexes. Discrete gradients are  partial
face-coface matchings that satisfy certain acyclicity conditions.  Forman's theory also has an elegant  graph theoretic formulation~\cite{Ch00}, in which the  acyclic matchings (or \emph{Morse matchings}) in the Hasse diagram of a simplicial complex are in one-to-one correspondence with the discrete gradients on the simplicial complex.  For this reason, we  use the terms \emph{gradient vector fields} and \emph{Morse matchings} interchangeably.

 Discrete Morse theory has become a popular tool in computational topology, image processing and visualization~\cite{Bauer2017Morse,bauerjact,Caz03,shiva2012,wangchen,robins,stratified,mukherjee}, and is actively studied in algebraic and toplogical combinatorics~\cite{Jollenbeck2009,Ko08,Mil07,scoville,kozlovorganized}. Over the period of last decade, it has emerged a powerful computational tool for several problems in topological data analysis~\cite{kkm,holmgren,DeyWangWang,reininghaus}.   Because of the wide array of applications there is a lot of practical interest in computing gradient vector fields on simplicial complexes with a (near-)optimal number of critical simplices~\cite{Lan16,ripser,Brendel,MR3472422,HMN,HMMNWJD10,LewinerLT04}. 
The idea of using discrete Morse theory to speed up the computation of (co)homology~\cite{HMMNWJD10,lampret,nanda}, persistent homology~\cite{MN13,ripser}, zigzag persistence~\cite{escolar,maria}, and multiparameter persistent homology~\cite{scaramuccia} relies on the fact that discrete Morse theory can be employed to reduce the problem of computing the homology of an input simplicial complex to that of a much smaller chain complex. 

The effectiveness of certain heuristics for Morse matching  raises an important question: to what extent is it feasible to obtain near-optimal solutions for Morse matching in polynomial time?  To this end,  inapproximability results for \minmorse for simplicial complexes of dimension $d \geq 3$ were established in~\cite{BR19}. To this date, however, we are unaware of  any hardness results for \minmorse on $2$-complexes from the perspective of inapproximability or parameterized complexity (although the related  \erasability problem was shown to be \WP-hard by Burton et al.~\cite{MR3472422}).  With this paper, we seek to close the knowledge gap.  By establishing various hardness results, we demonstrate the limitations of polynomial time methods for computing near-optimal Morse matchings. On the other hand, by devising an approximation algorithm, we make it evident that \minmorse on $2$-complexes is not entirely inapproximable. 
We also observe that the typical Morse matching instances drawn from a wide range of parameter values of the Costa--Farber complexes are a lot easier in contrast to the discouraging worst case inapproximability bounds.

\subsection{Related work}

Joswig and Pfetsch~\cite{JP06} showed that finding an optimal gradient vector field is an \NP-hard problem  based on the relationship between erasability and Morse Matching observed by Lewiner~\cite{Le02,Le03a}. The \emph{erasability problem} was first studied by E\v gecio\v glu and Gonzalez~\cite{EG96}. Joswig and Pfetsch also posed the question of approximability of optimal Morse matching as an open problem. On the positive side, Rathod et al.~\cite{RBN17} devised the first approximation algorithms for \maxmorse on simplicial complexes that provide constant factor approximation bounds for fixed dimension. Complementing these results, Bauer and Rathod~\cite{BR19} showed that for simplicial complexes of dimension $d \geq 3$ with $n$ simplices, it is \NP-hard to approximate \minmorse within a factor of $O(n^{1-\epsilon})$, for any $\epsilon>0$.  However, the question of approximability of \minmorse for $2$-complexes is left unanswered in~\cite{BR19}. 

Next, Burton et al.~\cite{MR3472422}  showed that the \erasability problem (that is, finding the number of $2$-simplices that need to be removed to make a $2$-complex erasable) is \WP-complete. 
We note the \WP-hardness of erasability can be inferred from our methods as well, and therefore our result can be seen as a strengthening of the hardness result from~\cite{MR3472422}.
Moreover, our parameterized inapproximability results rely on the machinery developed by Eickmeyer et al.~\cite{Eickmeyer} and Marx~\cite{Marx}.

Our reduction techniques have a flavor that is similar to the techniques used by Malgouryes and Franc\'es~\cite{MF} and Tancer~\cite{MR3439259} for proving \NP-hardness of certain collapsibility problems.
In particular, Tancer~\cite{MR3439259} also describes a procedure for \enquote{filling $1$-cycles with disks} to make the complex contractible (and even collapsible for satisfiable inputs). Our technique of filling $1$-cycles is however entirely different (and arguably  simpler) than Tancer's procedure.
Our work is also related to \cite{BR19,BRS19} in that we use the so-called modified dunce hats for constructing the gadget used in the reduction. Recently, modified dunce hats were used to provide a simpler proof of \NP-completeness of the shellability decision problem~\cite{woodroofe}.

\subsection{The Morse Matching Problems\label{sec:mmprob}}

The \maxmorse problem ($\maxmm$) can be described
as follows: Given a simplicial complex $\mathcal{K}$, compute a gradient
vector field that maximizes the cardinality of matched (regular) simplices,
over all possible gradient vectors fields on $\mathcal{K}$. Equivalently,
the goal is to maximize the number of gradient pairs. For the complementary
problem \minmorse ($\minmm$), the goal is to compute
a gradient vector field that minimizes the number of unmatched (critical)
simplices, over all possible gradient vector fields on $\KCC$. While the problem of finding the exact
optimum are equivalent for $\minmm$ and $\maxmm$, the approximation variants behave quite differently.

Additionally, we define another variant of the minimization problem for $2$-dimensional complexes, namely  \minrmorse ($\minrmm$). For this problem, we seek to minimize the total number of critical simplices minus one.
This variant is natural, since any discrete gradient necessarily has at least one critical $0$-simplex.
It corresponds to a variant definition of simplicial complexes commonly used in combinatorics, which also consider the empty set as a simplex of dimension $-1$.

\subsection{Our contributions}

\begin{table}[h]
\begin{tabular}{|c|c|c|}
\hline 
 Approx. algorithm  & {Bauer, R. (2019)} & {This paper} \tabularnewline
\hline 
\hline 
 \minmorse  (dim. $ = 2$) & -- & $\ensuremath{O(\frac{n}{\log n})}$\tabularnewline
\hline 
\end{tabular}
\end{table}

\begin{table}[h]
\begin{tabular}{|c|c|c|}
\hline 
 Inapproximability   & {Bauer, R. (2019)} & {This paper} \tabularnewline
\hline 
\hline 
\maxmorse (dim. $\ge2$) & $\left(1-\frac{1}{4914}\right)+\epsilon$ & --\tabularnewline
\hline 
  \minmorse (dim. $\ge3$) & $O(n^{1-\epsilon})$ & --\tabularnewline
\hline 
 \minmorse (dim. $=2$) & -- & $\weirderfactor$\tabularnewline
\hline 

\end{tabular}
\caption{ (In)approximability of Morse matching}
\end{table}

\begin{table}[h]
\begin{tabular}{|c|c|c|}
\hline 
 & {Burton et al. (2013)} &  {This paper} \tabularnewline
\hline 
\hline 
 \WP-hardness of $\textbf{er}(K)$ (SP) & \cmark & \cmark\tabularnewline
\hline 
\WP-hardness of \minmorse (SP) & -- & \cmark\tabularnewline
\hline 
FPT-inapproximability of \minmorse (SP) & -- & \cmark\tabularnewline
\hline 
FPT-algorithm for \minmorse (TW) & \cmark & --\tabularnewline
\hline 
\end{tabular}

\caption{Parameterized complexity of Morse matching.  In the above table, SP denotes standard parameterization, whereas TW denotes treewidth parameterization.}
\end{table}

In \Cref{sec:mainminrmorse}, we establish several hardness results for  $\minrmm$, using a reduction from $\mincircuitsat$. In particular, we show the following:
 \begin{compactitem} 
 \item $\minrmm$ has no approximation within a factor of  $2^{\log^{(1-\epsilon)}n}  $, for any $\epsilon>0$, unless \mbox{$\NP \subseteq \QP$} (throughout this paper, $\log$ denotes the logarithm by base $2$),
 \item the standard parameterization of $\minrmm$ is \WP-hard, and 
 \item $\minrmm$ with standard parameterization has no FPT approximation algorithm for \emph{any} approximation ratio function $\rho$, unless $\fpt = \WP$.
 \end{compactitem}

In \Cref{sec:mainminmorse}, we first show that the \WP-hardness result and FPT-inapproximability results easily carry over from $\minrmm$ to $\minmm$. 
To the best of our knowledge, this constitutes the first  FPT-inapproximability result in computational topology.
Using the amplified complex construction introduced in~\cite{BR19}, we  observe that the inapproximability result also carries over from $\minrmm$ to 
$\minmm$. In particular, we show that even for $2$-complexes $\minmm$ cannot be approximated within a factor of $2^{\log^{(1-\epsilon)}n}  $, for any $\epsilon>0$, unless $\NP \subseteq \QP$, where $n$ denotes the number of simplices in the complex. 

 \Cref{sec:approxalgo,sec:costafarber} are concerned with some positive results.
First, in \Cref{sec:approxalgo}, we design an $O(\frac{n}{\log n})$-factor algorithm for $\minmm$ on $2$-complexes.
Then, in \Cref{sec:costafarber}, we make the  observation that  Kahle's techniques~\cite{kahle2} for designing  discrete gradients on random clique complexes generalize to  Costa--Farber random complexes. Specifically, we show that for a wide range of parameter  values, there exist discrete gradients for which the ratio of expected number of critical $r$-simplices to the expected number of $r$-simplices (for any fixed dimension $r$) tends to zero. Although these methods do not lead to  approximation algorithms, they fall under the general paradigm of beyond worst-case  analysis~\cite{roughgarden}.

 Note that we do not distinguish between abstract and geometric simplicial complexes since every abstract simplicial complex can be embedded in a Euclidean space of appropriate dimension.
As a final remark, we believe that with this paper we tie all the loose ends regarding complexity questions in discrete Morse theory.


\section{Topological preliminaries}

\subsection{Simplicial complexes}  \label{sub:Simplicial_complexes}

A $k$\emph{-simplex} $\sigma=\conv V$ is the convex hull
of a set $V$ of $(k+1)$ affinely independent points in $\mathbb{R}^{d}$. We call $k$ the dimension of $\sigma$.
We say that $\sigma$ is \emph{spanned} by the points $V$.
Any nonempty subset of $V$
also spans a simplex, a \emph{face} of $\sigma$. A simplex $\sigma$ is said to be a \emph{coface}
of a simplex  $\tau$ if and only if $\tau$ is face of $\sigma$. 
We say that $\sigma$ is a \emph{facet} of $\tau$ and $\tau$ is a \emph{cofacet} of $\sigma$ if $\sigma$ is a face of $\tau$ with $\dim\sigma=\dim\tau-1$. 
A \emph{simplicial complex} $\KCC$  is a collection of simplices that satisfies the following conditions:
\begin{compactitem}
\item	any face of a simplex in $\KCC$ also belongs to $\KCC$, and
\item	the intersection of two simplices $\sigma_1,\sigma_2\in \KCC$ is either empty or a face of both $\sigma_1$ and $\sigma_2$. 
\end{compactitem}

For a complex $\KCC$, we denote the set of $d$-simplices of $\KCC$ by $\dcomplex{\KCC}{d}$. The $n$-\emph{skeleton} of a simplicial complex $\KCC$ is the simplicial complex $ \bigcup_{m=0}^{n}\dcomplex{\KCC}{m}$. 
A simplex $\sigma$ is called a \emph{maximal face} of a simplicial complex $\KCC$ if it is not a strict subset of any other simplex $\tau \in \KCC$.
The \emph{underlying space} of $\KCC$ is the union of its simplices, denoted by $|\KCC|$. The underlying space is implicitly used whenever we refer to $\KCC$ as a topological space.

An \emph{abstract simplicial complex} $\SCC$ 
is a collection of finite nonempty sets $A \in \SCC$ such that every nonempty subset of $A$ is also contained in $\SCC$.
The sets in $\SCC$ are called its \emph{simplices}.
A subcomplex of $\KCC$ is an abstract simplicial complex $L$ such that every face of $L$ belongs to $K$; denoted as $L\subset K$.
For example, the vertex sets of the simplices in a geometric complex form an abstract simplicial complex, called its \emph{vertex scheme}.
Given an abstract simplicial complex $\KCC$ with $n$ simplices, we can associate a \emph{pointed simplicial complex}  to it by choosing an arbitrary vertex and regarding it as the distinguished basepoint of $\KCC$. The \emph{$m^{\text{th}}$ wedge sum of $\KCC$} is then the quotient space of a disjoint union of $m$ copies of $\KCC$ with the distinguished basepoints of each of the copies of $\KCC$ identified.

\subsection{Discrete Morse theory and Erasability}
\label{sub:Discrete-Morse-Theory}

We assume that the reader is familiar with simplicial complexes. \Cref{sub:Simplicial_complexes} summarizes the key definitions.
In this section, we provide a brief description of Forman's discrete Morse theory on simplicial
complexes. For a comprehensive expository introduction, we refer the reader to~\cite{Fo02}.

%
%

A real-valued function $f$ on a simplicial complex $\KCC$ is called a \emph{discrete Morse function} if
\begin{compactitem}
\item $f$ is \emph{monotonic}, i.e., $\sigma \subseteq \tau$ implies $f(\sigma) \leq f(\tau)$, and
\item for all $t \in \im(f) $, the preimage $f^{-1}(t)$ is either a singleton $\left\{ \sigma\right\}$ (in which case $\sigma$ is a \emph{critical simplex}) or a pair $\left\{\sigma,\tau\right\}$, where $\sigma$ is a facet of $\tau$ (in which case $\left(\sigma,\tau\right)$ form a \emph{gradient pair} and $\sigma$ and $\tau$ are \emph{regular simplices}). 
\end{compactitem}
Given a discrete Morse function $f$ defined on complex $\KCC$, the \emph{discrete gradient vector field} $\VSS$ of $f$ is the  collection  of pairs of simplices  $\left(\sigma,\tau\right)$, where $\left(\sigma,\tau\right)$ is in $\VSS$ if and only if $\sigma$ is a facet of $\tau$ and $f(\sigma)=f(\tau)$.


Discrete gradient vector fields have a useful interpretation in terms of acyclic graphs obtained from matchings on Hasse diagrams, due to Chari~\cite{Ch00}.
%
%
%
Let $\KCC$ be a simplicial complex, let $\HKCC$ be its Hasse diagram, and let $M$ be a matching in the underlying undirected graph~$\HKCC$. Let $\HKCC(M)$ be the directed graph obtained from $\HKCC$ by reversing the direction of each edge of the matching $M$. Then $M$ is a \emph{Morse matching} if and only if $\HKCC(M)$ is a directed acyclic graph. Every Morse matching $M$ on the Hasse diagram $\HKCC$ corresponds to a unique gradient vector field $\VCC_M$ on complex $\KCC$ and vice versa. For a Morse matching~$M$, the unmatched vertices correspond to critical simplices of $\VCC_M$, and the matched vertices correspond to the regular simplices of $\VSS_M$.



A non-maximal face $\sigma\in\KCC$ is said to be a \emph{free face} if it
is contained in a unique maximal simplex $\tau\in\KCC$. 
If $d = \dim\tau = \dim\sigma+1$,
we say that $\KCC^{\prime} = \KCC \setminus \{\sigma,\tau\}$ arises from $\KCC$ by an \emph{elementary
collapse}, or an \emph{elementary $d$-collapse} denoted by $\KCC\searrow^{e}\KCC^{\prime}$. Furthermore, we say that
$\KCC$ \emph{collapses} to $\LCC$, denoted by $\KCC\searrow\LCC$,
if there exists a sequence $\KCC=\KCC_{1},\KCC_{2},\dots\KCC_{n}=\LCC$
such that $\KCC_{i}\searrow^{e}\KCC_{i+1}$ for all $i$. 
 If $\KCC$ collapses to a point, one says that $\KCC$ is collapsible.

A simplicial collapse can be encoded by a discrete gradient.
\begin{theorem}[Forman~\cite{Fo98}, Theorem 3.3]
  \label{Gradient Collapsing Theorem}
  Let $\KCC$ be a simplicial complex with a vector field $\VSS$,
  and let $\LCC \subseteq \KCC$ be a subcomplex.
  If $\KCC \setminus \LCC$ is a union of pairs in $\VSS$, then $\KCC \searrow \LCC$.
\end{theorem}
In this case, we say that the collapse $\KCC \searrow \LCC$ is \emph{induced by} the gradient $\VSS$. As a consequence of this theorem, we obtain:

\begin{theorem}[Forman~\cite{Fo98}, Corollary 3.5]
Let $\KCC$ be a simplicial complex with a discrete gradient vector field $\VSS$ and let $m_d$ denote the number of critical simplices of $\VSS$ of dimension~$d$. Then $\KCC$ is homotopy equivalent to a CW complex with exactly $m_d$ cells of dimension $d$.
\end{theorem}
In particular, a discrete gradient vector field on $\KCC$ with $m_d$ critical simplices of dimension~$d$ gives rise to a chain complex having dimension $m_d$ in each degree $d$, whose homology is isomorphic to that of $\KCC$. This condensed representation motivates the algorithmic search for (near-)optimal Morse matchings.



Following the terminology used in~\cite{EG96,BR19}, we make the following definitions:
A maximal face $\tau$ in a simplicial complex $\KCC$ is called an \emph{internal simplex} if it has no free face. 
If a $2$-complex $\KCC$ collapses to a $1$-complex, we say that $\KCC$ is \emph{erasable}. Moreover, for a $2$-complex~$\KCC$, the quantity $\erk$ is the minimum number
of internal $2$-simplices that need to be removed so that the resulting complex collapses to a $1$-complex. Equivalently, it is the minimum number of critical $2$-simplices of any discrete gradient on $\KCC$.
Furthermore, we say that a subcomplex $\LCC \subseteq \KCC$ is an \emph{erasable subcomplex of $\KCC$ (through the gradient $\VSS$)}
if there exists another subcomplex $\MCC \subseteq \KCC$ with $\KCC \searrow \MCC$ (induced by the gradient $\VSS$) such that the set of $2$-dimensional  simplices of these complexes satisfy the following relation:
$ \twocomplex\LCC \subseteq \twocomplex\KCC \setminus \twocomplex\MCC $. We call such a gradient $\VSS$ an \emph{erasing gradient}.
 Finally, we say that a simplex $\sigma$ in a complex $\KCC$ is \emph{eventually free (through the gradient $\VSS$)} if there exists a subcomplex $\LCC$ of $\KCC$ such that $\KCC \searrow \LCC$ (induced by $\VSS$) and $\sigma$ is free in $\LCC$.
Equivalently, $\KCC$ collapses further to a subcomplex not containing~$\sigma$.

We recall the following results from \cite{BR19}. 

\begin{lemma}[Bauer, Rathod \cite{BR19}, Lemma 2.1]
\label{lem:uniqueCriticalVertex}
Let $\KCC$ be a connected simplicial complex, let $p$ be a vertex of $\KCC$, and let $\VSS_1$ be a discrete gradient on $\KCC$ with $m_0>1$ critical simplices of dimension $0$ and $m$ critical simplices in total. Then there exists a polynomial time algorithm to compute another gradient vector field $\widetilde\VSS$ on $\KCC$ with $p$ as the only critical simplex of dimension $0$ and $m - 2(m_0-1)$ critical simplices  in total. 
\end{lemma} 


\begin{lemma}[Bauer, Rathod \cite{BR19}, Lemma 2.3]
\label{lem:subErasable}
If $\KCC$ is an erasable complex, then any subcomplex $\LCC\subset \KCC$ is also erasable. 
\end{lemma}

\begin{lemma} \label{lem:easy} Suppose that we are given a complex $K$ and a  set $M $ of simplices in $K$ with the property that simplices in $M$ have no cofaces in $K$. Then  $L  = K \setminus M$ is a subcomplex of $K$ with the property that the
  gradient vector fields on  $K$ with all simplices in $M$ critical  are in one-to-one correspondence with gradient vector fields on $L$.
\end{lemma}
\begin{proof} Given a gradient vector field  on $L$, we extend it to a gradient vector field on $K$ by making all simplices in $M$ critical. Given a vector field  $\VCC$ on $K$, the restriction $\VCC |_{L}$ is a gradient vector field on $L$.
\end{proof}

\begin{notation}For the remainder of the paper, we use $[m]$ to denote the set $\{1,2,\dots,m\}$ for any $m\in \mathbb{N}$, and $[i,j]$  to denote the set $\{i, i+1, \dots , j\}$ for any $i,j\in \mathbb{N}$.
\end{notation}

\section{Algorithmic preliminaries}

\subsection{Approximation algorithms} \label{sec:apxintro}

An $\alpha$-\emph{approxi\-ma\-tion algorithm} for an optimization problem is a polynomial-time algorithm that, for all instances of the problem, produces a solution whose objective value is within a factor $\alpha$ of the objective value of an optimal solution. The factor $\alpha$
is called the \emph{approximation ratio} (or \emph{approximation factor}) of the algorithm. 

An \emph{approximation preserving reduction} is a polynomial time procedure for transforming an optimization problem $A$ to an optimization problem $B$, such that an $\alpha$-approximation algorithm for $B$ implies an $f(\alpha)$-approximation algorithm for $A$, for some function $f$. Then, if $A$ is hard to approximate within factor $f(\alpha)$, the reduction implies that $B$ is hard to approximate within factor $\alpha$.
A particularly well-studied class of approximation preserving reductions is given by the \emph{L-reductions}, which provide an effective tool in proving hardness of approximability results~\cite{PY91, WS10}. 

Now, consider a minimization problem $A$  with a non-negative integer valued objective function $m_{A}$.
Given an instance $x$ of $A$, the goal is to find a solution~$y$  minimizing the objective function $m_{A}(x,y)$. Define $\opa(x)$ as the minimum value of the objective function on input $x$.   
An \emph{L-reduction} (with parameters $\mu$ and $\nu$)
from a minimization problem $A$ to another minimization problem
$B$ is a pair of polynomial time computable functions $f$ and $g$, and fixed constants $\mu,\nu>0$, satisfying the following conditions:
\begin{enumerate}
\item The function $f$ maps instances of $A$ to instances of $B$.
\item For any instance $x$ of $A$, we have \[\opb(f(x))\leq \mu \, \opa(x) . \]
\item The function $g$ maps an instance $x$ of $A$ and a solution of the corresponding instance $f(x)$ of $B$ to a solution of $x$.
\item For any instance $x$ of $A$, and any solution $y$ of $f(x)$, we have
 \[m_A(x,g(x,y)) - \opa(x)  \leq \nu \left(m_B(f(x),y) -  \opb(f(x)) \right) . \]
\end{enumerate}
If $\mu=\nu=1$, the reduction is \emph{strict}.


We will use the following straightforward fact about L-reductions, which appears as Theorem~16.6 in a book by Williamson and Shmoys~\cite{WS10}.
\begin{theorem}\label{thm:newfactorappend} If there is an L-reduction with parameters $\mu$ and $\nu$ from a minimization problem $A$ to another minimization problem $B$, and there is a $(1+ \delta )$-approximation algorithm for $B$,
then there is a $(1+ \mu\nu\delta)$-approximation algorithm for $A$.
\end{theorem}



\subsection{Parameterized complexity} \label{sec:paramintro}

Parameterized complexity, as introduced by Downey and Fellows in \cite{MR1656112}, is a refinement of classical complexity theory. The theory revolves around the general idea of developing complexity bounds for instances of a problem not just based on their size, but also involving an additional \emph{parameter}, which might be significantly smaller than the size.
Specifically, we have the following definition.

\begin{definition}[Parameter, parameterized problem~\cite{FG06}]
  Let $\Sigma$ be a finite alphabet.
  \begin{enumerate}
    \item A \emph{parameter} of $\alphastrings$, the set of strings over $\Sigma$, is a function $\rho: \alphastrings \to \bbn$, attaching to every input $w \in \alphastrings$ a natural number $\rho(w)$.
    \item A \emph{Parameterized problem} over $\Sigma$ is a pair $(P,\rho)$ consisting of a set $P\subseteq \alphastrings$ and a  (polynomial time computable) parametrization $\rho : \alphastrings \to \bbn$.
    \item  A parameterized problem $(P,\rho)$ is said to be {\em fixed-parameter tractable} or {\em FPT in the parameter $\rho$} if the question 
            $$( x , p ) \in \{ (y, \rho(y)) \mid y \in P \}$$ 
         can be decided in running time $O(g(p)) \cdot  | x |^{O(1)}$, where $g \colon \bbn \to \bbn$ is an arbitrary computable function depending only on the parameter $p$.
  \end{enumerate}
\end{definition}

FPT reductions provide a principal tool to establish hardness results in the parameterized complexity landscape. 

\begin{definition}[FPT reduction~\cite{FG06}]
Given two parameterized problems $(P,k)$ and $(Q,k')$,
we say that there is an \emph{FPT reduction} from $(P,k)$ and $(Q,k')$,  if there exists a functions $\varphi$ that transforms parameterized instances of $P$ to parameterized instances of $Q$ while satisfying the following properties:
\begin{enumerate}
\item $\varphi$ is computable by an FPT algorithm,
\item $\varphi(x)$ is a yes-instance of $(Q,k')$ if and only if $x$ is a yes-instance of $(P,k)$.
\item There exists a computable function $g \colon \bbn \to \bbn$ such that $k'(\varphi(x)) \le g(k(x))$.
\end{enumerate}
\end{definition}


The natural way of turning a minimization problem into a decision problem is to add a value $k$ to the input instance, and seek a solution with cost at most $k$. Taking this value $k$ appearing in the input as the parameter is called \emph{the standard parameterization} of the minimization problem (sometimes also referred to as \emph{the natural parameterization}).
In general, the parameter can be any function of the input instance, for example, the treewidth of the input graph, or the maximum degree of the input graph.

Parameterized approximability is an extension of the notion of classical approximability. Informally, an FPT approximation algorithm is an algorithm whose running time is fixed parameter tractable for the
parameter \emph{cost of the solution} and whose approximation factor~$\rho$ is a function  of the parameter (and independent of the input size).
For instance, every polynomial time approximation algorithm with constant approximation factor is automatically an FPT approximation algorithm, but an approximation algorithm with approximation factor $\Theta(\sqrt{n})$,
where $n$ denotes the input size, is not an FPT approximation algorithm. 
Next, following \cite{Marx}, for standard parameterization of minimization problems, we provide definitions for  FPT approximation algorithms and FPT cost approximation algorithms. Analogous definitions for maximization problems are also considered in~\cite{Marx}.

\begin{definition}[FPT approximation algorithm~\cite{Marx}] 
Let $P$ be an \NP minimization problem, and let $\rho:\mathbb{N}\rightarrow\mathbb{R}_{\geq1}$
be a computable function such that $k \mapsto k\cdot\rho(k)$ is nondecreasing. An \emph{FPT approximation algorithm} for $P$
(over some alphabet $\Sigma$) with approximation ratio $\rho$ is
an algorithm $\mathbb{A}$ with the following properties: 
\begin{enumerate}
\item For every input $(x,k)$ 
whose optimal solution has cost at most $k$, $\mathbb{A}$
computes a solution for $x$ of cost at most $k\cdot\rho(k)$. For
inputs $(x,k)$
 without a solution of cost at most $k$, the output can be arbitrary. 
\item The runtime of $\mathbb{A}$ on input $(x,k)$ is 
$O(g(k)\cdot|x|^{O(1)})$ for some computable function $g$.
\end{enumerate}
\end{definition}




It is often convenient to work with a weaker notion of approximability
where an algorithm is only required to compute the cost of an optimal
solution rather than an actual optimal solution, and to work with
decision rather than optimization problems. With that in mind, the
notion of \emph{FPT cost approximability} was introduced in~\cite{ChenFlumGrohe}.

\begin{definition}[FPT cost approximation algorithm~\cite{Marx}]
Let $P$ be an \NP minimization problem (over the alphabet $\Sigma$), and
 $\rho:\mathbb{N}\to\mathbb{R}_{\geq1}$  a computable function.
 For an instance $x$ of $P$, let $\min(x)$ denote its optimal value.
Then, a decision algorithm $\mathbb{A}$ is an \emph{FPT cost approximation
algorithm} for $P$ with approximation ratio $\rho$ if 
\begin{enumerate}
\item For feasible instances $x$ of $P$ and parameterized instances $(x,k)$, $\mathbb{A}$ satisfies: 

\begin{enumerate}
\item If $k\ge\min(x)\cdot\rho(\min(x))$, then $\mathbb{A}$ accepts $(x,k)$.
\item If $k<\min(x)$, then $\mathbb{A}$ rejects $(x,k)$.
\end{enumerate}
\item $\mathbb{A}$ is an FPT algorithm. That is, there exists a computable function $f$ with the property that for an input $(x,k)$, the running time of $\mathbb{A}$ is bounded by $f(k)\cdot|x|^{O(1)}$ .
\end{enumerate}
\end{definition}

It can be readily checked that FPT-approximability implies FPT cost approximability with the same approximation factor. Please refer to Section~3.1~of~\cite{ChenFlumGrohe} for more details.

\begin{theorem}[Chen et al.~\cite{ChenFlumGrohe}] \label{thm:cfg}
Let $P$ be an \NP minimization problem over the alphabet $\Sigma$,
and let $\rho:\mathbb{N}\rightarrow\mathbb{R}_{\ge1}$ be a computable
function such that $k\cdot\rho(k)$ is nondecreasing and unbounded.
Suppose that $P$ is FPT approximable with approximation ratio $\rho$.
Then $P$ is FPT cost approximable with approximation ratio $\rho$.
\end{theorem} 

An immediate consequence of the theorem above is that if $P$ is not FPT cost approximable with approximation ratio $\rho$ (under certain complexity theory assumptions), then $P$ is not FPT approximable with approximation ratio $\rho$ 
(under the same assumptions).

Gap problems and gap-preserving reductions were originally introduced in the context of proving the PCP theorem~\cite{Arora} -- a cornerstone in the theory of approximation algorithms.
These notions have natural analogues in the parameterized approximability setting. Below, we follow the definitions as provided by Eickmeyer et al.~\cite{Eickmeyer}.

\begin{definition}[gap instance of a parameterized problem~\cite{Eickmeyer}] 
Let $\delta:\mathbb{N}\to\mathbb{R}_{\geq1}$ be a function, $P$ a minimization problem, and $P'$ its standard parameterization. An instance $(x,k)$ is a $\delta$-gap instance of $P'$
if either $\min(x)\leq k$ or $\min(x)\geq k \cdot\delta(k)$.
\end{definition}

\begin{definition}[gap-preserving FPT reduction~\cite{Eickmeyer}] 
Let $\alpha,\beta \colon \mathbb{N} \to \mathbb{R}_{\geq1}$ be two computable functions, and let $P$ and $Q$ be two minimization problems.
Let $P'$ and $Q'$ be the natural parameterizations of $P$ and $Q$, respectively.
We say that a reduction $R$ from $P'$ to $Q'$ is a \emph{$(\alpha,\beta)$-gap-preserving
FPT reduction} if
\begin{enumerate}
\item $R$ is an FPT reduction from $P'$ to $Q'$,
\item for every $\alpha$-gap instance $(x,k)$ of $P'$, the instance $R(x,k)$ is a $\beta$-gap instance of $Q'$.
\end{enumerate}
\end{definition}

We use gap-preserving FPT reductions to establish FPT-inapproximability.

\subsection{Circuits}


First, we recall some elementary notions from Boolean circuits. 
In particular, by an \emph{and-node}, we mean the digital logic gate that implements logical conjuction ($\wedge$), by an \emph{or-node}, we mean the digital logic gate that implements logical disjunction ($\vee$), 
and by a \emph{not-node}, we mean the digital logic gate that implements negation ($\lnot$).

\begin{definition}[Boolean circuit]
A \emph{Boolean circuit} $C$ is a directed acyclic graph, where each
node is labeled in the following way:
\begin{enumerate}
\item every node with in-degree greater than $1$ is either an \emph{and-node} or an \emph{or-node}, 
\item each node of in-degree $1$ is labeled as a \emph{negation node}, 
\item and each node of in-degree $0$ is an \emph{input node}. 
\end{enumerate}
Moreover, exactly one of the nodes with out-degree 0 is labeled as the \emph{output node}.
\end{definition}

Below, we  recall some essential parameterized complexity results concerning  circuits.

We use the terms \emph{gates} and \emph{nodes} interchangeably. We say that a gate has \emph{fan-in} $k$ if its in-degree is at most $k$.
We say that a gate is an \emph{ordinary gate} if it is neither an input gate nor an output gate. 
We denote the nodes and edges in $C$ by  $V(C)$ and $E(C)$ respectively. The size of a circuit $C$, denoted by $|C|$,  is the total number of nodes and edges in $C$. 
That is, $|C| = |V(C)| + |E(C)|$. The \emph{Hamming weight} of an assignment is the number of input gates receiving value 1. An assignment on the input nodes induces an assignment on all nodes. So given an assignment  from the input nodes of circuit $C$ to $\{0, 1\}$, we say that the assignment satisfies $C$ if the value of the output node is $1$ for that assignment. Let $\mathscr{G}_{I}$ denote the set of input gates of~$C$.
Then, an assignment $A$ can be viewed as a binary vector of size $|\mathscr{G}_{I}|$.
In the \wcs ($\wcsabbr$) problem, we are given a circuit $C$ and an integer $k$, and the task is to decide if $C$ has a satisfying assignment of Hamming weight at most $k$.
Accordingly, in the \minwcs ($\minwcsabbr$) problem, we are given a circuit~$C$, and the task is to find a satisfying assignment with minimum Hamming weight.


\begin{definition}[\WP]
A parameterized problem $W$ belongs to the class \WP if it can be reduced to the standard parameterization of $\wcsabbr$.
\end{definition}


A Boolean circuit is monotone if it does not contain any negation nodes. Let $\mathcal{C}^+$ be the class of all monotone Boolean circuits.  
Then, \minmonotone ($\mincircuitsat$) is the restriction of  the problem $\minwcsabbr$ to input circuits belonging to $\mathcal{C}^+$.
%

The following result seems to be folklore and appears in the standard literature \cite{MR2238686,Downey}.
\begin{theorem}[Theorem 3.14~\cite{FG06}] \label{thm:wpcircuit}
The standard parameterization of $\mincircuitsat$ is \WP-complete.
\end{theorem}
Furthermore, Eickmeyer et al.~\cite{Eickmeyer} showed that unless \WP $=$ \fpt, $\mincircuitsat$ does not have an FPT approximation algorithm with polylogarithmic approximation factor $\rho$.
The FPT-inapproximability result was subsequently improved by Marx~\cite{Marx} as follows.
\begin{theorem}[Marx~\cite{Marx}]  \label{thm:bigguy}
$\mincircuitsat$ is not FPT cost approximable, unless $\fpt = \WP$.
\end{theorem}
Combined with~\Cref{thm:cfg}, the above theorem implies that $\mincircuitsat$ is not FPT-approximable for any function $\rho$, unless $\fpt = \WP$.

 \begin{remark}[Fan-in $2$ circuits] \label{rem:fan}
 We note that it is possible to transform a monotone circuit~$C$ to another monotone circuit $C'$ such that both circuits are satisfied on the same inputs, and  every gate of $C'$ has fan-in $2$.
This is achieved as follows: Each or-gate of in-degree~$k$ in~$C$ is replaced by a \emph{tree} of or-gates with in-degree-$2$ in $C'$, and each and-gate of in-degree $k$ in $C$ is replaced by a \emph{tree} of and-gates with in-degree-$2$ in $C'$.
In each case, we  transform a single gate having fan-in $k$ to a sub-circuit of $\Theta(k)$ gates having depth $\Theta(\log k)$ and fan-in $2$. In fact, it is easy to check that $|C'|$ is a polynomial function of $|C|$, and $C'$ can be computed from $C$ in time polynomial in $C$. Since the number of input gates for $C$ and $C'$ is the same, for the rest of the paper we will assume without loss of generality that an input circuit instance has fan-in $2$.
\end{remark}

\section{Reducing $\mincircuitsat$ to $\minrmm$} \label{sec:reduceit}

In this section, we describe how to construct a $2$-complex $\kc$ that corresponds to a monotone circuit $C(V,E)$. By~\Cref{rem:fan}, we assume without loss of generality that $C$ has fan-in~$2$.
For the rest of the paper, we denote the number of gates in $C$ by $n$. 
Also, throughout, we use the notation $j \in [a,b]$ to mean that $j$ takes integer values in the interval $[a,b]$.

Following the notation from~\Cref{sec:apxintro},  given a monotone circuit $C=(\VCC,\ECC)$ and the associated complex $\kc$, let $\OPTA$ denote the optimal value of the $\mincircuitsat$ problem on $C$, and let $\OPTB$ denote the optimal value of the $\minrmm$ problem on $\kc$.
The value of the objective function $\ALGB$ is the number of critical simplices in $\VCC$ minus one; the value of the objective function $\ALGA$  is the Hamming weight of the input assignment.
In~\Cref{sec:kcmain}, we describe the map $K$ that transforms instances of
$\mincircuitsat$ (monotone circuits $C$) to instances of $\minrmm$ (simplicial complexes $\kc$), and the map $\ICC$ that transforms solutions of $\minrmm$ (discrete gradients $\VCC$ on $\kc$) to solutions of $\mincircuitsat$ (satisfying input assignments $\ICC(C,\VCC)$ of circuit $C$).

\subsection{The building block for the gadget}

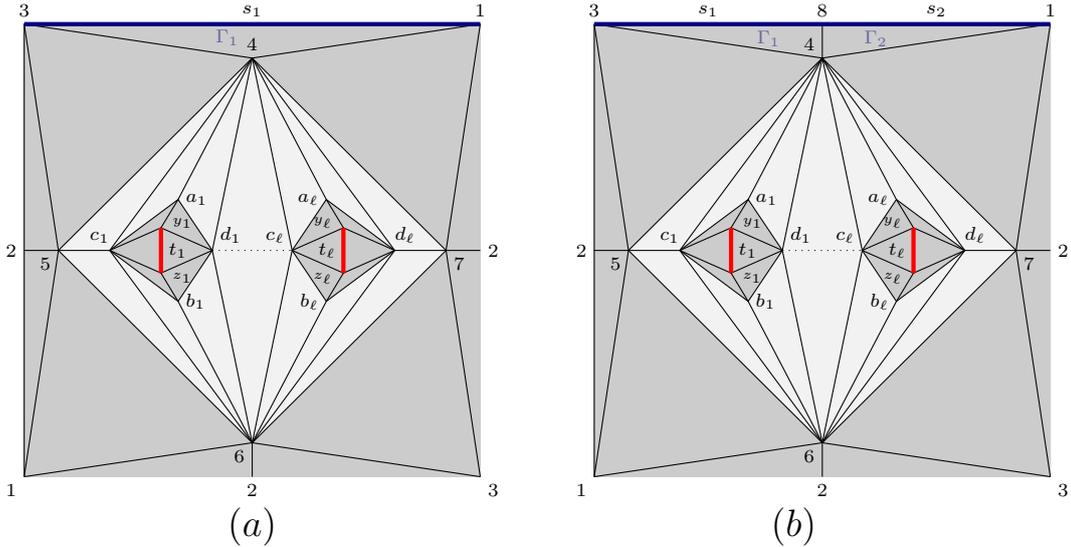
\begin{figure}[hbt]
   \centering{
    \scriptsize
  \begin{tikzpicture}[scale=1.5,baseline]%
{ [every path/.style = {line cap=round,line join=round}]

     \coordinate (lab) at (0,-2.2);

             \coordinate (r) at (0,-1.7);
    \coordinate (w) at (-1.7,0);
    \coordinate (t) at (0,1.7);
    \coordinate (v) at (1.7,0);
    \coordinate (s) at (2,-2);
    \coordinate (u) at (2,2);
    \coordinate (sb) at (-2,2);
    \coordinate (q) at (-2,0);
    \coordinate (ur) at (-2,-2);
    \coordinate (ql) at (2,0);
    \coordinate (qr) at (0,-2);
  
      \coordinate (c0) at (-1.5+.25,0);
    \coordinate (d0) at (-0.6+.25,0);
    \coordinate (a0) at (-0.9+.25,0.45);
    \coordinate (b0) at (-0.9+.25,-0.45);
    \coordinate (y0) at (-1.05+.25,0.2);
    \coordinate (z0) at (-1.05+.25,-0.2);

    \coordinate (cl) at (0.6-.25,0);
    \coordinate (dl) at (1.5-.25,0);
    \coordinate (al) at (0.9-.25,0.45);
    \coordinate (bl) at (0.9-.25,-0.45);
    \coordinate (yl) at (1.05-.25,0.2);
    \coordinate (zl) at (1.05-.25,-0.2);

    \fill[verylightgray] (t)--(w)--(r)--(v)--cycle;
    \fill[lightergray] 
    (u)--(s)--(ur)--(sb)--cycle 
    (t)--(w)--(r)--(v)--cycle
    (a0)--(d0)--(b0)--(c0)--cycle
    (al)--(dl)--(bl)--(cl)--cycle
    ;
    
    \draw[lightergray] (u)--(s)--(ur)--(sb);


    \draw (ur)--(sb);
    \draw (r)--(v);
    \draw (v)--(u);
    \draw (t)--(sb);
    \draw (t)--(v);
    \draw (t)--(w);
    \draw (w)--(q);
    \draw (w)--(ur);
    \draw (w)--(sb);
    \draw (r)--(ur);
    \draw (r)--(qr);
    \draw (s)--(r);
    \draw (s)--(v);
    \draw (ql)--(v);
    \draw (u)--(t);

    \draw (r)--(w);

    
    \draw (w)--(c0);

    \draw (a0)--(d0)--(b0)--(c0)--cycle;
    \draw (y0)--(d0)--(z0)--(c0)--cycle;
    \draw (t)--(a0)--(y0)--(z0)--(b0)--(r);
    \draw (t)--(c0)--(r);
    \draw (t)--(d0)--(r);
    
    \draw[dotted] (d0)--(cl);
    
    \draw (al)--(dl)--(bl)--(cl)--cycle;
    \draw (yl)--(dl)--(zl)--(cl)--cycle;
    \draw (t)--(al)--(yl)--(zl)--(bl)--(r);
    \draw (t)--(cl)--(r);
    \draw (t)--(dl)--(r);

    \draw (dl)--(v);
      \draw (u)--(sb);
      
      \draw[ultra thick, red] (y0)--(z0);

\draw[ultra thick, red] (yl)--(zl);

\draw[ultra thick, NavyBlue] (u)--(sb);
     \node[above left=0.1cm,DarkBlue] at (t) {$\Gamma_1$};

    \node[below left] at (r) {$6$};
    \node[below left] at (w) {$5$};
    \node[above ] at (t) {$4$};
    \node[below right] at (v) {$7$};
    \node[below right] at (s) {$3$};
    \node[above] at (u) {$1$};
    \node[above] at (sb) {$3$};
    \node[above] at ($(u)!1/2!(sb)$) {$s_1$};
    \node[left] at (q) {$2$};
    \node[below left] at (ur) {$1$};
    \node[right] at (ql) {$2$};
    \node[below] at (qr) {$2$};


    \node[above left=0mm and -1mm] at (c0) {$c_1$};
    \node[above right] at (d0) {$d_1$};
    \node[right] at (a0) {$a_1$};
    \node[right] at (b0) {$b_1$};
    \node[above right=-1mm and .5mm] at (y0) {\tiny $y_1$};
    \node[below right=-1mm and .5mm] at (z0) {\tiny $z_1$};    
    \node[right] at($(y0)!1/2!(z0)$) {$t_1$};

    \node[above left] at (cl) {$c_\ell$};
    \node[above right=0mm and -1mm] at (dl) {$d_\ell$};
    \node[left] at (al) {$a_\ell$};
    \node[left] at (bl) {$b_\ell$};
    \node[above left=-1mm and .5mm] at (yl) {\tiny $y_\ell$};
    \node[below left=-1mm and .5mm] at (zl) {\tiny $z_\ell$};    
    \node[left] at($(yl)!1/2!(zl)$) {$t_\ell$};

    \node[below] at (lab) {\Large $(a)$};
    %

             \coordinate (rnew) at (0+5,-1.7);
    \coordinate (wnew) at (-1.7+5,0);
    \coordinate (tnew) at (0+5,1.7);
    \coordinate (vnew) at (1.7+5,0);
    \coordinate (snew) at (2+5,-2);
    \coordinate (unew) at (2+5,2);
    \coordinate (sbnew) at (-2+5,2);
    \coordinate (qnew) at (-2+5,0);
    \coordinate (urnew) at (-2+5,-2);
    \coordinate (qlnew) at (2+5,0);
    \coordinate (qrnew) at (0+5,-2);
    \coordinate (x1new) at (0+5,2);
  
      \coordinate (c0new) at (-1.5+5.25,0);
    \coordinate (d0new) at (-0.6+5.25,0);
    \coordinate (a0new) at (-0.9+5.25,0.45);
    \coordinate (b0new) at (-0.9+5.25,-0.45);
    \coordinate (y0new) at (-1.05+5.25,0.2);
    \coordinate (z0new) at (-1.05+5.25,-0.2);

    \coordinate (clnew) at (0.6+4.75,0);
    \coordinate (dlnew) at (1.5+4.75,0);
    \coordinate (alnew) at (0.9+4.75,0.45);
    \coordinate (blnew) at (0.9+4.75,-0.45);
    \coordinate (ylnew) at (1.05+4.75,0.2);
    \coordinate (zlnew) at (1.05+4.75,-0.2);
     \coordinate (labnew) at (0+4.75,-2.2);
     
    \fill[verylightgray] (tnew)--(wnew)--(rnew)--(vnew)--cycle;
    \fill[lightergray] 
    (unew)--(snew)--(urnew)--(sbnew)--cycle 
    (tnew)--(wnew)--(rnew)--(vnew)--cycle
    (a0new)--(d0new)--(b0new)--(c0new)--cycle
    (alnew)--(dlnew)--(blnew)--(clnew)--cycle
    ;
    
    \draw[lightergray] (unew)--(snew)--(urnew)--(sbnew);


    \draw (urnew)--(sbnew);
    \draw (rnew)--(vnew);
    \draw (vnew)--(unew);
    \draw (tnew)--(sbnew);
    \draw (tnew)--(vnew);
    \draw (tnew)--(wnew);
    \draw (wnew)--(qnew);
    \draw (wnew)--(urnew);
    \draw (wnew)--(sbnew);
    \draw (rnew)--(urnew);
    \draw (rnew)--(qrnew);
    \draw (snew)--(rnew);
    \draw (snew)--(vnew);
    \draw (qlnew)--(vnew);
    \draw (unew)--(tnew);

    \draw (rnew)--(wnew);

    \draw (tnew)--(x1new);
    
    \draw (wnew)--(c0new);

    \draw (a0new)--(d0new)--(b0new)--(c0new)--cycle;
    \draw (y0new)--(d0new)--(z0new)--(c0new)--cycle;
    \draw (tnew)--(a0new)--(y0new)--(z0new)--(b0new)--(rnew);
    \draw (tnew)--(c0new)--(rnew);
    \draw (tnew)--(d0new)--(rnew);
    
    \draw[dotted] (d0new)--(clnew);
    
    \draw (alnew)--(dlnew)--(blnew)--(clnew)--cycle;
    \draw (ylnew)--(dlnew)--(zlnew)--(clnew)--cycle;
    \draw (tnew)--(alnew)--(ylnew)--(zlnew)--(blnew)--(rnew);
    \draw (tnew)--(clnew)--(rnew);
    \draw (tnew)--(dlnew)--(rnew);

    \draw (dlnew)--(vnew);
    
          \draw (unew)--(sbnew);
          
  \draw[ultra thick, red] (y0new)--(z0new);

\draw[ultra thick, red] (ylnew)--(zlnew);

\draw[ultra thick, NavyBlue] (unew)--(sbnew);

         \node[DarkBlue] at ([shift={(160:0.5)}]tnew) {$\Gamma_1$};
     \node[DarkBlue] at ([shift={(20:0.5)}]tnew) {$\Gamma_2$};
    \node[below left] at (rnew) {$6$};
    \node[below left] at (wnew) {$5$};
    \node[above left] at (tnew) {$4$};
    \node[below right] at (vnew) {$7$};
    \node[below right] at (snew) {$3$};
    \node[above] at (unew) {$1$};
    \node[above] at (sbnew) {$3$};
    \node[left] at (qnew) {$2$};
    \node[below left] at (urnew) {$1$};
    \node[right] at (qlnew) {$2$};
    \node[below] at (qrnew) {$2$};

    \node[above] at (x1new) {$8$};

    \node[above left=0mm and -1mm] at (c0new) {$c_1$};
    \node[above right] at (d0new) {$d_1$};
    \node[right] at (a0new) {$a_1$};
    \node[right] at (b0new) {$b_1$};
    \node[above right=-1mm and .5mm] at (y0new) {\tiny $y_1$};
    \node[below right=-1mm and .5mm] at (z0new) {\tiny $z_1$};    
    \node[right] at($(y0new)!1/2!(z0new)$) {$t_1$};

    \node[above left] at (clnew) {$c_\ell$};
    \node[above right=0mm and -1mm] at (dlnew) {$d_\ell$};
    \node[left] at (alnew) {$a_\ell$};
    \node[left] at (blnew) {$b_\ell$};
    \node[above left=-1mm and .5mm] at (ylnew) {\tiny $y_\ell$};
    \node[below left=-1mm and .5mm] at (zlnew) {\tiny $z_\ell$};    
    \node[left] at($(ylnew)!1/2!(zlnew)$) {$t_\ell$};

    \node[above] at($(sbnew)!1/2!(x1new)$) {$s_1$};
        \node[above] at($(unew)!1/2!(x1new)$) {$s_2$};
    \node[below] at (labnew) {\Large $(b)$};

}
\end{tikzpicture}
  }
  \caption{The  figure $(a)$ on the left depicts $\gadget_{1,\ell}$ that is collapsible through one free face, namely $s_1 = \{3,1\}$. The figure $(b)$ on the right depicts  $\gadget_{2,\ell}$ that is collapsible through two free faces, namely $s_1  = \{3,8\}$ and $s_2=\{8,1\}$. The edges $\{1,2\}$ and $\{2,3\}$ on the right and at the bottom of both subfigures are shown in light grey to indicate that they are identified to $\{1,2\}$ and $\{2,3\}$ on the left.   \label{fig:gadget}}
\end{figure}

We shall first describe a complex that serves as the principal building block for the gadget in our reduction. 
The building block is based on a modification of Zeeman's \emph{dunce hat}~\cite{MR0156351}. The dunce hat is a simplicial complex that is contractible (i.e. has the homotopy type of a point) but has no free faces and is therefore not collapsible.
In contrast, we work with \emph{modified dunce hats}~\cite{Hachimori} that are collapsible through either one or two free edges. The modified dunce hat has been previously used to show  hardness of approximation of \maxmorse~\cite{BR19} and \WP-hardness of \paraexptwo~\cite{BRS19}, and is discussed extensively in these papers.

\Cref{fig:gadget} depicts two triangulations of modified dunce hats, which we denote by $\gadget_{m,\ell}$. 
We use the subscript $m,\ell$ to designate the numbers of distinguished edges of $\gadget_{m,\ell}$, which come in two types:
the free edges of $\gadget_{m,\ell}$ denoted by $s_i $, $1 \leq i \leq m$, and the edges of type $t_j = \{ y_j,z_j \}$, $1 \leq j \leq \ell$. 
These distinguished edges are precisely the edges that are identified to edges from other building blocks.
In \Cref{fig:gadget}, we depict $\gadget_{1,\ell}$, and $\gadget_{2,\ell}$, with distinguished edges $s_i $, and  $t_j$ highlighted.
 
Note that, in this paper, we only consider modified dunce hats with either one or two free edges. That is, for the purpose of this paper, $m \in \{1, 2\}$.
 Also, abusing terminology, we often refer to \enquote{modified dunce hats}  as simply \enquote{dunce hats}.

\begin{remark} \label{rem:gradient}
It is easy to check that after executing a series of elementary $2$-collapses, $\gadget_{m,\ell}$ collapses to a complex induced by edges
 \[\left\{\{1,2\}, \{2,3\},\{2,6\}, \{6,5\}, \{6,7\}, \{6,b_i\},\{6,c_i\},\{6,d_i\},\{b_i,z_i\} ,\{z_i,y_i\} ,\{y_i,a_i\} \right\} \cup \FCC\]
for $i \in [1,\ell]$, and,
 \begin{compactitem}
 \item $\FCC =  \{ \{v,4\} \}$ if $m = 1$, where $v = 6$, if $\ell$ is even, and $v = \nicefrac{(\ell - 1)}{2} + 1$  if $\ell$ is odd, 
 \item $\FCC = \{s_1, \{4,8\}\}$ if $m=2$ and  the collapse starts with a gradient pair involving $s_2$,    
 \item $\FCC = \{s_2, \{4,8\}\}$ if $m=2$ and  the collapse starts with a gradient pair involving $s_1$.
 \end{compactitem}
 The edges that are left behind after executing all the $2$-collapses are highlighted using examples in \Cref{fig:gadgetpaths}.
$(a)$ depicts the case where $\ell$ is odd and $m = 1$.
$(b)$ depicts the case where  $\ell$ is even, $m = 2$ and   the collapse starts with a gradient pair involving $s_1$.
\end{remark}


\begin{figure}[hbt]
   \centering{
    \scriptsize
  \begin{tikzpicture}[scale=1.5,baseline]%
{ [every path/.style = {line cap=round,line join=round}]

     \coordinate (lab) at (0,-2.2);

    \coordinate (r) at (0,-1.7);
    \coordinate (w) at (-1.7,0);
    \coordinate (t) at (0,1.7);
    \coordinate (v) at (1.7,0);
    \coordinate (s) at (2,-2);
    \coordinate (u) at (2,2);
    \coordinate (sb) at (-2,2);
    \coordinate (q) at (-2,0);
    \coordinate (ur) at (-2,-2);
    \coordinate (ql) at (2,0);
    \coordinate (qr) at (0,-2);
    \coordinate (x1) at (-1.2,2);
    \coordinate (x2) at (-0.4,2);
    \coordinate (xm1) at (1.2,2);


    \coordinate (c0) at (-1.5,0);
    \coordinate (d0) at (-0.6,0);
    \coordinate (a0) at (-0.9,0.45);
    \coordinate (b0) at (-0.9,-0.45);
    \coordinate (y0) at (-1.05,0.2);
    \coordinate (z0) at (-1.05,-0.2);

    \coordinate (cq) at (-0.4,0);
    \coordinate (dq) at (0.4,0);
    \coordinate (aq) at (0,0.45);
    \coordinate (bq) at (0,-0.45);
    \coordinate (yq) at (0,0.2);
    \coordinate (zq) at (0,-0.2);
    
    \coordinate (cl) at (0.6,0);
    \coordinate (dl) at (1.5,0);
    \coordinate (al) at (0.9,0.45);
    \coordinate (bl) at (0.9,-0.45);
    \coordinate (yl) at (1.05,0.2);
    \coordinate (zl) at (1.05,-0.2);

    \fill[verylightgray] (t)--(w)--(r)--(v)--cycle;
    \fill[lightergray] 
    (u)--(s)--(ur)--(sb)--cycle 
    (t)--(w)--(r)--(v)--cycle
    (a0)--(d0)--(b0)--(c0)--cycle
    (al)--(dl)--(bl)--(cl)--cycle
     (aq)--(dq)--(bq)--(cq)--cycle
    ;
    
         \fill[verylightblue] (u)--(sb)--(t);

    \draw[lightergray] (u)--(s)--(ur)--(sb);


    \draw (ur)--(sb);
    \draw (r)--(v);
    \draw (v)--(u);
    \draw (t)--(sb);
    \draw (t)--(v);
    \draw (t)--(w);
    \draw (w)--(q);
    \draw (w)--(ur);
    \draw (w)--(sb);
    \draw (r)--(ur);
    \draw (r)--(qr);
    \draw (s)--(r);
    \draw (s)--(v);
    \draw (ql)--(v);
    \draw (u)--(t);

    \draw (r)--(w);

    
    \draw (w)--(c0);

    \draw (a0)--(d0)--(b0)--(c0)--cycle;
    \draw (y0)--(d0)--(z0)--(c0)--cycle;
    \draw (t)--(a0)--(y0)--(z0)--(b0)--(r);
    \draw (t)--(c0)--(r);
    \draw (t)--(d0)--(r);
    
  \draw (d0)--(cq);
  \draw (dq)--(cl);
    
    \draw (al)--(dl)--(bl)--(cl)--cycle;
    \draw (yl)--(dl)--(zl)--(cl)--cycle;
    \draw (t)--(al)--(yl)--(zl)--(bl)--(r);
    \draw (t)--(cl)--(r);
    \draw (t)--(dl)--(r);
    
        \draw (dl)--(v);
    
    \draw (aq)--(dq)--(bq)--(cq)--cycle;
    \draw (yq)--(dq)--(zq)--(cq)--cycle;
        \draw (t)--(aq)--(yq)--(zq)--(bq)--(r);
  \draw (t)--(cq)--(r);
    \draw (t)--(dq)--(r);
    \draw (u)--(sb);

\draw[ultra thick, Maroon3] (sb)--(ur);
\draw[ultra thick, Maroon3] (qr)--(r);
\draw[ultra thick, Maroon3] (b0)--(r);
\draw[ultra thick, Maroon3] (bq)--(r);
\draw[ultra thick, Maroon3] (bl)--(r);
\draw[very thick, SeaGreen4] (c0)--(r);
\draw[very thick, SeaGreen4] (cq)--(r);
\draw[very thick, SeaGreen4] (cl)--(r);
\draw[very thick, SeaGreen4] (d0)--(r);
\draw[very thick, SeaGreen4] (dq)--(r);
\draw[very thick, SeaGreen4] (dl)--(r);
\draw[very thick, SeaGreen4] (w)--(r);
\draw[very thick, SeaGreen4] (v)--(r);

\draw[very thick, SeaGreen4] (bq)--(t);
\draw[ultra thick, Maroon3] (b0)--(z0);
\draw[very thick, RoyalBlue] (z0)--(y0);
\draw[very thick, SeaGreen4] (y0)--(a0);

\draw[ultra thick, Maroon3] (bq)--(zq);
\draw[very thick, RoyalBlue] (zq)--(yq);
\draw[very thick, SeaGreen4] (yq)--(aq);

\draw[ultra thick, Maroon3] (bl)--(zl);
\draw[very thick, RoyalBlue] (zl)--(yl);
\draw[very thick, SeaGreen4] (yl)--(al);

 \draw[ ->] ($(u)!1/2!(sb)$)--($(u)!1/2!(sb)!1/2!(t)$);
 
    \draw[->] ($(u)!1/2!(t)$)--($(u)!1/2!(t)!1/3!(ql)$);
    \draw[->] ($(u)!1/2!(v)$)--($(u)!1/2!(v)!1/3!(ql)$);
    \draw[->] ($(s)!1/2!(v)$)--($(s)!1/2!(v)!1/3!(r)$);
    \draw[->] ($(v)!1/2!(ql)$)--($(v)!1/2!(ql)!1/3!(s)$);
    \draw[->] ($(r)!1/2!(s)$)--($(r)!1/2!(s)!1/3!(qr)$);
  
   \draw[->] ($(r)!1/2!(ur)$)--($(r)!1/2!(ur)!1/3!(qr)$);
    \draw[->] ($(w)!1/2!(ur)$)--($(w)!1/2!(ur)!1/3!(r)$);
    \draw[->] ($(q)!1/2!(w)$)--($(q)!1/2!(w)!1/3!(ur)$);
    \draw[->] ($(sb)!1/2!(w)$)--($(sb)!1/2!(w)!1/3!(q)$);
    \draw[->] ($(sb)!1/2!(t)$)--($(sb)!1/2!(t)!1/3!(w)$);
    
    \draw[->] ($(w)!1/3!(t)$)--($(w)!1/3!(t)!1/3!(c0)$);
           \draw[->] ($(w)!1/2!(c0)$)--($(c0)!1/2!(w)!1/7!(r)$);

      \draw[->] ($(c0)!1/2!(t)$)--($(c0)!1/2!(t)!1/3!(a0)$);
        \draw[->] ($(a0)!1/2!(c0)$)--($(a0)!1/2!(c0)!2/3!(y0)$);
         \draw[->] ($(c0)!1/2!(y0)$)--($(c0)!1/2!(y0)!1/2!(z0)$); 
         \draw[->] ($(c0)!2/3!(z0)$)--($(c0)!2/3!(z0)!1/3!(b0)$); 
 \draw[->] ($(c0)!1/2!(b0)$)--($(c0)!1/2!(b0)!1/6!(r)$); 
 
       \draw[->] ($(cq)!1/2!(t)$)--($(cq)!1/2!(t)!1/3!(aq)$);
        \draw[->] ($(aq)!1/2!(cq)$)--($(aq)!1/2!(cq)!2/3!(yq)$);
         \draw[->] ($(cq)!1/2!(yq)$)--($(cq)!1/2!(yq)!1/2!(zq)$); 
         \draw[->] ($(cq)!2/3!(zq)$)--($(cq)!2/3!(zq)!1/3!(bq)$); 
 \draw[->] ($(cq)!1/2!(bq)$)--($(cq)!1/2!(bq)!1/6!(r)$);

        \draw[->] ($(al)!1/2!(t)$)--($(al)!1/2!(t)!1/3!(cl)$);
        \draw[->] ($(al)!1/2!(cl)$)--($(al)!1/2!(cl)!2/3!(yl)$);
         \draw[->] ($(cl)!1/2!(yl)$)--($(cl)!1/2!(yl)!1/2!(zl)$); 
         \draw[->] ($(cl)!2/3!(zl)$)--($(cl)!2/3!(zl)!1/3!(bl)$); 
 \draw[->] ($(cl)!1/2!(bl)$)--($(cl)!1/2!(bl)!1/6!(r)$);

     \draw[->] ($(v)!1/3!(t)$)--($(v)!1/3!(t)!1/3!(dl)$);
          \draw[->] ($(v)!1/2!(dl)$)--($(v)!1/2!(dl)!1/7!(r)$);
     
  \draw[->] ($(dl)!1/2!(t)$)--($(dl)!1/2!(t)!1/3!(al)$);
   \draw[->] ($(al)!1/2!(dl)$)--($(al)!1/2!(dl)!2/3!(yl)$);
      \draw[->] ($(yl)!1/2!(dl)$)--($(yl)!1/2!(dl)!1/2!(zl)$);
      \draw[->] ($(zl)!1/3!(dl)$)--($(zl)!1/3!(dl)!1/3!(bl)$);
  \draw[->] ($(bl)!1/2!(dl)$)--($(bl)!1/2!(dl)!1/6!(r)$);      
  
    \draw[->] ($(dq)!1/2!(t)$)--($(dq)!1/2!(t)!1/3!(aq)$);
   \draw[->] ($(aq)!1/2!(dq)$)--($(aq)!1/2!(dq)!2/3!(yq)$);
      \draw[->] ($(yq)!1/2!(dq)$)--($(yq)!1/2!(dq)!1/2!(zq)$);
      \draw[->] ($(zq)!1/3!(dq)$)--($(zq)!1/3!(dq)!1/3!(bq)$);
  \draw[->] ($(bq)!1/2!(dq)$)--($(bq)!1/2!(dq)!1/6!(r)$);

      \draw[->] ($(a0)!1/2!(t)$)--($(a0)!1/2!(t)!1/3!(d0)$);
   \draw[->] ($(a0)!1/2!(d0)$)--($(a0)!1/2!(d0)!2/3!(y0)$);
      \draw[->] ($(y0)!1/2!(d0)$)--($(y0)!1/2!(d0)!1/2!(z0)$);
      \draw[->] ($(z0)!1/3!(d0)$)--($(z0)!1/3!(d0)!1/3!(b0)$);
  \draw[->] ($(b0)!1/2!(d0)$)--($(b0)!1/2!(d0)!1/6!(r)$);      
  
    \draw[->] ($(t)!2/3!(d0)$)--($(t)!2/3!(d0)!1/3!(cq)$); 
      \draw[->] ($(cq)!1/2!(d0)$)--($(cq)!1/2!(d0)!1/6!(r)$); 
      
          \draw[->] ($(t)!2/3!(cl)$)--($(t)!2/3!(cl)!1/3!(dq)$); 
      \draw[->] ($(dq)!1/2!(cl)$)--($(dq)!1/2!(cl)!1/6!(r)$); 
      
        \draw[SeaGreen4,thick,->] (w)--($(w)!1/3!(r)$); 
        \draw[SeaGreen4,thick,->] (v)--($(v)!1/3!(r)$); 
           \draw[SeaGreen4,thick,->] (c0)--($(c0)!5/12!(r)$); 
               \draw[SeaGreen4,thick,->] (dl)--($(dl)!5/12!(r)$); 
               
                          \draw[SeaGreen4,thick,->] (dq)--($(dq)!1/2!(r)$); 
               \draw[SeaGreen4,thick,->] (cq)--($(cq)!1/2!(r)$);
                          \draw[SeaGreen4,thick,->] (d0)--($(d0)!1/2!(r)$); 
               \draw[SeaGreen4,thick,->] (cl)--($(cl)!1/2!(r)$);
      
       \draw[SeaGreen4,thick,->] (t)--($(t)!2/3!(aq)$);

       \draw[SeaGreen4,thick,->] (aq)--($(aq)!2/3!(yq)$); 
        \draw[RoyalBlue,very thick,->] (yq)--($(yq)!1/2!(zq)$); 
        
               \draw[SeaGreen4,thick,->] (a0)--($(a0)!2/3!(y0)$); 
        \draw[RoyalBlue,very thick,->] (y0)--($(y0)!1/2!(z0)$); 
        
               \draw[SeaGreen4,thick,->] (al)--($(al)!2/3!(yl)$); 
        \draw[RoyalBlue,very thick,->] (yl)--($(yl)!1/2!(zl)$); 
      
%
    

    \node[below left] at (r) {$6$};
     \node[above left=0.1cm,DarkBlue] at (t) {$\Gamma_1$};
    \node[above] at (u) {$1$};
    \node[above] at (sb) {$3$};
   \node[left] at (q) {$2$};
    \node[below left] at (ur) {$1$};
    \node[right] at (ql) {$2$};
    \node[below] at (qr) {$2$};
     \node[below right] at (s) {$3$};

    


    \node[above] at($(sb)!1/2!(u)$) {$s_1$};
    

    \node[below] at (lab) {\Large $(a)$};
    
             \coordinate (rnew) at (0+5,-1.7);
    \coordinate (wnew) at (-1.7+5,0);
    \coordinate (tnew) at (0+5,1.7);
    \coordinate (vnew) at (1.7+5,0);
    \coordinate (snew) at (2+5,-2);
    \coordinate (unew) at (2+5,2);
    \coordinate (sbnew) at (-2+5,2);
    \coordinate (qnew) at (-2+5,0);
    \coordinate (urnew) at (-2+5,-2);
    \coordinate (qlnew) at (2+5,0);
    \coordinate (qrnew) at (0+5,-2);
    \coordinate (x1new) at (0+5,2);
  
      \coordinate (c0new) at (-1.5+5.25,0);
    \coordinate (d0new) at (-0.6+5.25,0);
    \coordinate (a0new) at (-0.9+5.25,0.45);
    \coordinate (b0new) at (-0.9+5.25,-0.45);
    \coordinate (y0new) at (-1.05+5.25,0.2);
    \coordinate (z0new) at (-1.05+5.25,-0.2);

    \coordinate (clnew) at (0.6+4.75,0);
    \coordinate (dlnew) at (1.5+4.75,0);
    \coordinate (alnew) at (0.9+4.75,0.45);
    \coordinate (blnew) at (0.9+4.75,-0.45);
    \coordinate (ylnew) at (1.05+4.75,0.2);
    \coordinate (zlnew) at (1.05+4.75,-0.2);
     \coordinate (labnew) at (0+4.75,-2.2);
     
    \fill[verylightgray] (tnew)--(wnew)--(rnew)--(vnew)--cycle;
    \fill[lightergray] 
    (unew)--(snew)--(urnew)--(sbnew)--cycle 
    (tnew)--(wnew)--(rnew)--(vnew)--cycle
    (a0new)--(d0new)--(b0new)--(c0new)--cycle
    (alnew)--(dlnew)--(blnew)--(clnew)--cycle
    ;

               \fill[verylightblue] (unew)--(sbnew)--(tnew);
 
    \draw[lightergray] (unew)--(snew)--(urnew)--(sbnew);


    \draw (urnew)--(sbnew);
    \draw (rnew)--(vnew);
    \draw (vnew)--(unew);
    \draw (tnew)--(sbnew);
    \draw (tnew)--(vnew);
    \draw (tnew)--(wnew);
    \draw (wnew)--(qnew);
    \draw (wnew)--(urnew);
    \draw (wnew)--(sbnew);
    \draw (rnew)--(urnew);
    \draw (rnew)--(qrnew);
    \draw (snew)--(rnew);
    \draw (snew)--(vnew);
    \draw (qlnew)--(vnew);
    \draw (unew)--(tnew);

    \draw (rnew)--(wnew);

    \draw (tnew)--(x1new);
    
    \draw (wnew)--(c0new);

    \draw (a0new)--(d0new)--(b0new)--(c0new)--cycle;
    \draw (y0new)--(d0new)--(z0new)--(c0new)--cycle;
    \draw (tnew)--(a0new)--(y0new)--(z0new)--(b0new)--(rnew);
    \draw (tnew)--(c0new)--(rnew);
    \draw (tnew)--(d0new)--(rnew);
    
  \draw (rnew)--(tnew);
    
    \draw (alnew)--(dlnew)--(blnew)--(clnew)--cycle;
    \draw (ylnew)--(dlnew)--(zlnew)--(clnew)--cycle;
    \draw (tnew)--(alnew)--(ylnew)--(zlnew)--(blnew)--(rnew);
    \draw (tnew)--(clnew)--(rnew);
    \draw (tnew)--(dlnew)--(rnew);

    \draw (dlnew)--(vnew);
    
          \draw (unew)--(sbnew);
          
  \draw[->] ($(unew)!1/2!(tnew)$)--($(unew)!1/2!(tnew)!1/2!(x1new)$);
   \draw[->] ($(x1new)!1/2!(sbnew)$)--($(x1new)!1/2!(sbnew)!1/2!(tnew)$);

    \draw[->] ($(vnew)!1/2!(tnew)$)--($(vnew)!1/2!(tnew)!1/3!(unew)$);
    \draw[->] ($(unew)!1/2!(vnew)$)--($(unew)!1/2!(vnew)!1/3!(qlnew)$);
    \draw[->] ($(snew)!1/2!(vnew)$)--($(snew)!1/2!(vnew)!1/3!(rnew)$);
    \draw[->] ($(vnew)!1/2!(qlnew)$)--($(vnew)!1/2!(qlnew)!1/3!(snew)$);
    \draw[->] ($(rnew)!1/2!(snew)$)--($(rnew)!1/2!(snew)!1/3!(qrnew)$);
  
   \draw[->] ($(rnew)!1/2!(urnew)$)--($(rnew)!1/2!(urnew)!1/3!(qrnew)$);
    \draw[->] ($(wnew)!1/2!(urnew)$)--($(wnew)!1/2!(urnew)!1/3!(rnew)$);
    \draw[->] ($(qnew)!1/2!(wnew)$)--($(qnew)!1/2!(wnew)!1/3!(urnew)$);
    \draw[->] ($(sbnew)!1/2!(wnew)$)--($(sbnew)!1/2!(wnew)!1/3!(qnew)$);
    \draw[->] ($(sbnew)!1/2!(tnew)$)--($(sbnew)!1/2!(tnew)!1/3!(wnew)$);
    
    \draw[->] ($(wnew)!1/3!(tnew)$)--($(wnew)!1/3!(tnew)!1/3!(c0new)$);
           \draw[->] ($(wnew)!1/2!(c0new)$)--($(c0new)!1/2!(wnew)!1/7!(rnew)$);

       \draw[->] ($(c0new)!1/2!(tnew)$)--($(c0new)!1/2!(tnew)!1/3!(a0new)$);
        \draw[->] ($(a0new)!1/2!(c0new)$)--($(a0new)!1/2!(c0new)!2/3!(y0new)$);
         \draw[->] ($(c0new)!1/2!(y0new)$)--($(c0new)!1/2!(y0new)!1/2!(z0new)$); 
         \draw[->] ($(c0new)!2/3!(z0new)$)--($(c0new)!2/3!(z0new)!1/3!(b0new)$); 
 \draw[->] ($(c0new)!1/2!(b0new)$)--($(c0new)!1/2!(b0new)!1/6!(rnew)$);

        \draw[->] ($(clnew)!1/2!(tnew)$)--($(clnew)!1/2!(tnew)!1/3!(alnew)$);
        \draw[->] ($(alnew)!1/2!(clnew)$)--($(alnew)!1/2!(clnew)!2/3!(ylnew)$);
         \draw[->] ($(clnew)!1/2!(ylnew)$)--($(clnew)!1/2!(ylnew)!1/2!(zlnew)$); 
         \draw[->] ($(clnew)!2/3!(zlnew)$)--($(clnew)!2/3!(zlnew)!1/3!(blnew)$); 
 \draw[->] ($(clnew)!1/2!(blnew)$)--($(clnew)!1/2!(blnew)!1/6!(rnew)$);

   \draw[->] ($(alnew)!1/2!(tnew)$)--($(alnew)!1/2!(tnew)!1/3!(dlnew)$);
   \draw[->] ($(alnew)!1/2!(dlnew)$)--($(alnew)!1/2!(dlnew)!2/3!(ylnew)$);
      \draw[->] ($(ylnew)!1/2!(dlnew)$)--($(ylnew)!1/2!(dlnew)!1/2!(zlnew)$);
      \draw[->] ($(zlnew)!1/3!(dlnew)$)--($(zlnew)!1/3!(dlnew)!1/3!(blnew)$);
  \draw[->] ($(blnew)!1/2!(dlnew)$)--($(blnew)!1/2!(dlnew)!1/6!(rnew)$);

      \draw[->] ($(a0new)!1/2!(tnew)$)--($(a0new)!1/2!(tnew)!1/3!(d0new)$);
   \draw[->] ($(a0new)!1/2!(d0new)$)--($(a0new)!1/2!(d0new)!2/3!(y0new)$);
      \draw[->] ($(y0new)!1/2!(d0new)$)--($(y0new)!1/2!(d0new)!1/2!(z0new)$);
      \draw[->] ($(z0new)!1/3!(d0new)$)--($(z0new)!1/3!(d0new)!1/3!(b0new)$);
  \draw[->] ($(b0new)!1/2!(d0new)$)--($(b0new)!1/2!(d0new)!1/6!(rnew)$);      
  
     \draw[->] ($(d0new)!1/2!(tnew)$)--($(d0new)!1/2!(tnew)!1/3!(rnew)$);
     
        \draw[->] ($(rnew)!1/2!(tnew)$)--($(rnew)!1/2!(tnew)!1/2!(clnew)$);

                 \draw[->] ($(dlnew)!1/3!(tnew)$)--($(dlnew)!1/3!(tnew)!1/3!(vnew)$);
          \draw[->] ($(vnew)!1/3!(dlnew)$)--($(vnew)!1/3!(dlnew)!1/7!(rnew)$);
          
                  \draw[SeaGreen4,thick,->] (wnew)--($(wnew)!1/3!(rnew)$); 
        \draw[SeaGreen4,thick,->] (vnew)--($(vnew)!1/3!(rnew)$); 
           \draw[SeaGreen4,thick,->] (c0new)--($(c0new)!5/12!(rnew)$); 
               \draw[SeaGreen4,thick,->] (dlnew)--($(dlnew)!5/12!(rnew)$);

                          \draw[SeaGreen4,thick,->] (d0new)--($(d0new)!1/2!(rnew)$); 
               \draw[SeaGreen4,thick,->] (clnew)--($(clnew)!1/2!(rnew)$);

\draw[ultra thick, Maroon3] (sbnew)--(urnew);
\draw[ultra thick, Maroon3] (qrnew)--(rnew);
\draw[ultra thick, Maroon3] (b0new)--(rnew);
\draw[ultra thick, Maroon3] (blnew)--(rnew);
\draw[very thick, SeaGreen4] (c0new)--(rnew);
\draw[very thick, SeaGreen4] (clnew)--(rnew);
\draw[very thick, SeaGreen4] (d0new)--(rnew);
\draw[very thick, SeaGreen4] (dlnew)--(rnew);
\draw[very thick, SeaGreen4] (wnew)--(rnew);
\draw[very thick, SeaGreen4] (vnew)--(rnew);

\draw[ultra thick, Maroon3] (b0new)--(z0new);
\draw[very thick, RoyalBlue] (z0new)--(y0new);
\draw[very thick, SeaGreen4] (y0new)--(a0new);

\draw[ultra thick, Maroon3] (blnew)--(zlnew);
\draw[very thick, RoyalBlue] (zlnew)--(ylnew);
\draw[very thick, SeaGreen4] (ylnew)--(alnew);

\draw[very thick, SeaGreen4] (tnew)--(x1new);
\draw[very thick, RoyalBlue] (x1new)--(unew);

                \draw[SeaGreen4,thick,->] (alnew)--($(alnew)!2/3!(ylnew)$); 
        \draw[RoyalBlue,very thick,->] (ylnew)--($(ylnew)!1/2!(zlnew)$); 
        
         \draw[SeaGreen4,thick,->] (a0new)--($(a0new)!2/3!(y0new)$); 
        \draw[RoyalBlue,very thick,->] (y0new)--($(y0new)!1/2!(z0new)$); 
     
       \draw[SeaGreen4,thick,->] (tnew)--($(tnew)!2/3!(x1new)$);                             
      \draw[RoyalBlue,very thick,->] (x1new)--($(x1new)!1/2!(unew)$); 

    \node[below left] at (rnew) {$6$};
    \node[below right] at (snew) {$3$};
         \node[right, DarkBlue] at ([shift={(160:0.6)}]tnew) {$\Gamma_1$};
     \node[left, DarkBlue] at ([shift={(20:0.6)}]tnew) {$\Gamma_2$};

    \node[above] at (unew) {$1$};
    \node[above] at (sbnew) {$3$};
    \node[left] at (qnew) {$2$};
    \node[below left] at (urnew) {$1$};
    \node[right] at (qlnew) {$2$};
    \node[below] at (qrnew) {$2$};




    \node[above] at($(sbnew)!1/2!(x1new)$) {$s_1$};
        \node[above] at($(unew)!1/2!(x1new)$) {$s_2$};

    \node[below] at (labnew) {\Large $(b)$};

}
\end{tikzpicture}
  }
  \caption{The  figure $(a)$ on the left depicts $\gadget_{1,3}$ that is collapsible through a unique free face, namely $s_1 = \{3,1\}$. 
  The complex $\gadget_{1,3}$ collapses to the subcomplex induced by the highlighted edges.
 When $s_1$ cannot be made free (because of edge identifications from other dunce hats), 
  then $s_1$ and $\Gamma_1$ are made critical. In any case, the remaining $2$-collapses are executed as shown in  figure $(a)$.
  The figure $(b)$ on the right depicts  $\gadget_{2,2}$ that is collapsible through two free faces, namely $s_1  = \{3,8\}$ and $s_2=\{8,1\}$. 
  There exists a collapsing sequence for $\gadget_{2,2}$ starting from the gradient pair $(s_1,\Gamma_1)$ such that $\gadget_{2,2}$ collapses to the subcomplex induced by the highlighted edges. A symmetric statement can be made for a collapse starting from  $(s_2,\Gamma_2)$. 
   \label{fig:gadgetpaths}}
  
\end{figure}
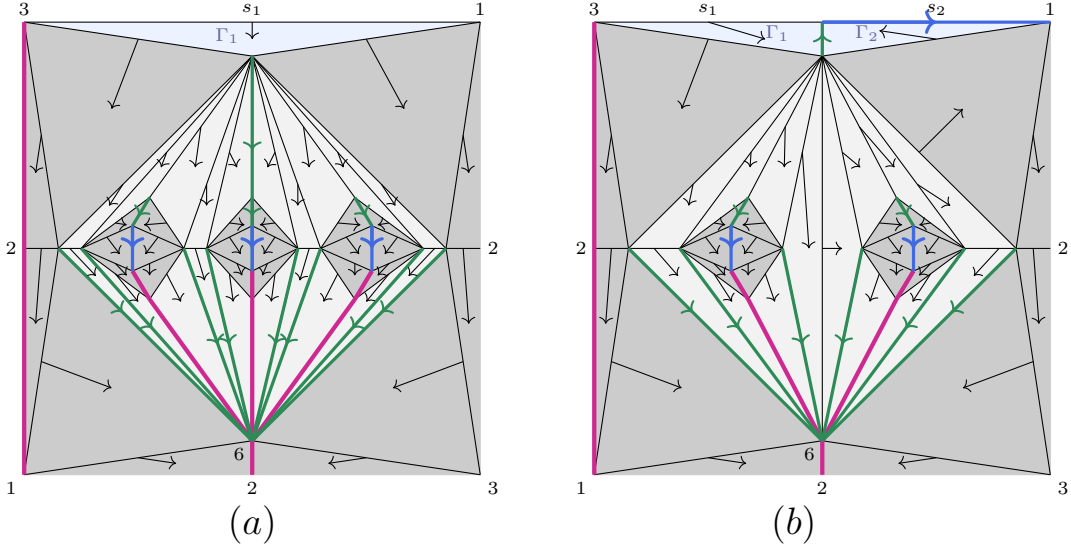

 \subsection{Construction of the complex $K(C)$ and the map $\ICC(C,\VCC)$} \label{sec:kcmain}

Given a circuit $C$, we first explain the construction of an intermediate complex $K'(C)$.  
We use the notation $\gadget_{m,\ell}^{\ijpair}$ to refer to the $j$-th copy of the the dunce hat associated to gate~$i$, having $m$ $s$-edges and $\ell$ $t$-edges.
Sometimes, we suppress the subscript, and use the notation $\gadget^{\ijpair}$ in place of $\gadget_{m,\ell}^{\ijpair}$.
As illustrated in~\Cref{fig:opgates}, to each input gate $G_i$ we associate a dunce hat $\gadget^{(i,1)}$. To the  output gate  $G_o$, we associate $n$ copies of dunce hats $\{ \gadget^{(o,j)} \}_{j = 1}^{n} $.
Moreover,~\Cref{fig:nonoutputgates} depicts how we associate to each ordinary gate $G_i$  $n$ blocks $\{^{1}\gadget^{(i,j)},\, ^{2}\gadget^{(i,j)}, \, ^{3}\gadget^{(i,j)} \}_{j = 1}^{n} $.
The superscript to the left indexes dunce hats internal to the block.
We call $^{3}\gadget^{(i,j)}$ the \emph{output component} of block $j$ associated to $G_i$. Likewise, we call $^{1}\gadget^{(i,j)}$  and $^{2}\gadget^{(i,j)}$ the \emph{input components} of block $j$ associated to $G_i$.
If $G_p$ serves as one of the two inputs to $G_q$, we say that $G_p$ is a \emph{predecessor} of $G_q$, and $G_q$ is the \emph{successor} of $G_p$.

A simplex labeled $\sigma$ in $\gadget_{m,\ell}$ is correspondingly labeled as $\sigma^{\pjpair}$ in  $\gadget_{m,\ell}^{\pjpair}$, (respectively as $^{k}\sigma^{\pjpair}$ in  $^{k}\gadget_{m,\ell}^{\pjpair}$).
We call the unique $s$-edge of the dunce hat associated to an input gate $G_i$, namely $s^{(i,1)}_1$, its \emph{feedback edge}. 
As depicted in~\Cref{fig:nonoutputgates}, for an ordinary gate $G_p$, for each $j \in [1,n]$, the $s_2$ edges of  $^{1}\gadget^{(p,j)}$ and  $\, ^{2}\gadget^{(p,j)}$, namely $^{1}s^{(p,j)}_2$ and $^{2}s^{(p,j)}_2$ respectively, are called the \emph{feedback edges} of the $j$-th block associated 
to $G_p$. For an ordinary gate $G_p$, for each $j \in [1,n]$, the $s_1$ edges of  $^{1}\gadget^{(p,j)}$ and  $\, ^{2}\gadget^{(p,j)}$,  namely $^{1}s^{(p,j)}_1$ and $^{2}s^{(p,j)}_1$ respectively, are called the \emph{input edges} of the $j$-th block associated to $G_p$.
For the output gate $G_o$, the  $s_1$ and  $s_2$ edges of $\gadget^{(o,j)}$, namely $s^{(o,j)}_1$ and $s^{(o,j)}_2$ respectively, are called the \emph{input edges} of the $j$-th copy associated to $G_o$.

To bring the notation of the edges closer to their function in the gadget, for the rest of the paper, we use the following alternative notation for $s$-edges. We denote the feedback edges  $s^{(i,1)}_1$, $^{1}s^{(p,j)}_2$ and $^{2}s^{(p,j)}_2$ described above as  $s^{(i,1)}_{f}$,  $s^{(p,j)}_{f_1}$ and  $s^{(p,j)}_{f_2}$ respectively. Also, we denote the input edges $^{1}s^{(p,j)}_1$,  $^{2}s^{(p,j)}_1$, $s^{(o,j)}_1$ and $s^{(o,j)}_2$ described above by $s^{(p,j)}_{\iota_1}$, $s^{(p,j)}_{\iota_2}$, $s^{(o,j)}_{\iota_1}$ and $s^{(o,j)}_{\iota_2}$ respectively.  Please see~\Cref{fig:nonoutputgates-glue}  for an example.

 
We start with a disjoint union of dunce hats (or blocks) associated to each gate.  
Then, for an ordinary gate $G_p$ that is a predecessor of $G_q$,  for all $j,k \in [1,n]$ two distinct $t$-edges from the $j$-th copy (output component of the block) associated to $G_q$ are identified to the two feedback edges of the   $k$-th block associated to $G_p$.
Also, for all $j,k \in [1,n]$ a $t$-edge from the output component of the $k$-th  block associated to $G_p$ is identified to an input edge of the $j$-th copy (block) associated to $G_q$.

For an input gate $G_p$ that is a predecessor of $G_q$, for all $j \in [1,n]$, a $t$-edge from the $j$-th copy (output component of the block) associated to $G_q$ is identified to the feedback edge of the unique dunce hat associated to $G_p$.
Also, for all $j \in [1,n]$, a $t$-edge from the  dunce hat associated to $G_p$ is attached to the input edge of  the $j$-th copy (block) associated to $G_q$.

Moreover, these identifications are done to ensure that: a feedback edge of a copy (block) associated to $G_p$ is free only if all the copies (output components of all the blocks) associated to \emph{all} the successors of $G_p$ have been erased, and an input 
edge of a copy (block) associated to $G_q$ is free only if the unique copy (all the output components of all the block) associated to the predecessors of $G_p$ have been erased.
Please refer to~\Cref{fig:nonoutputgates-glue} for an example illustrating the identifications.

 
 \begin{figure}
\centering
\includegraphics[scale=0.72]{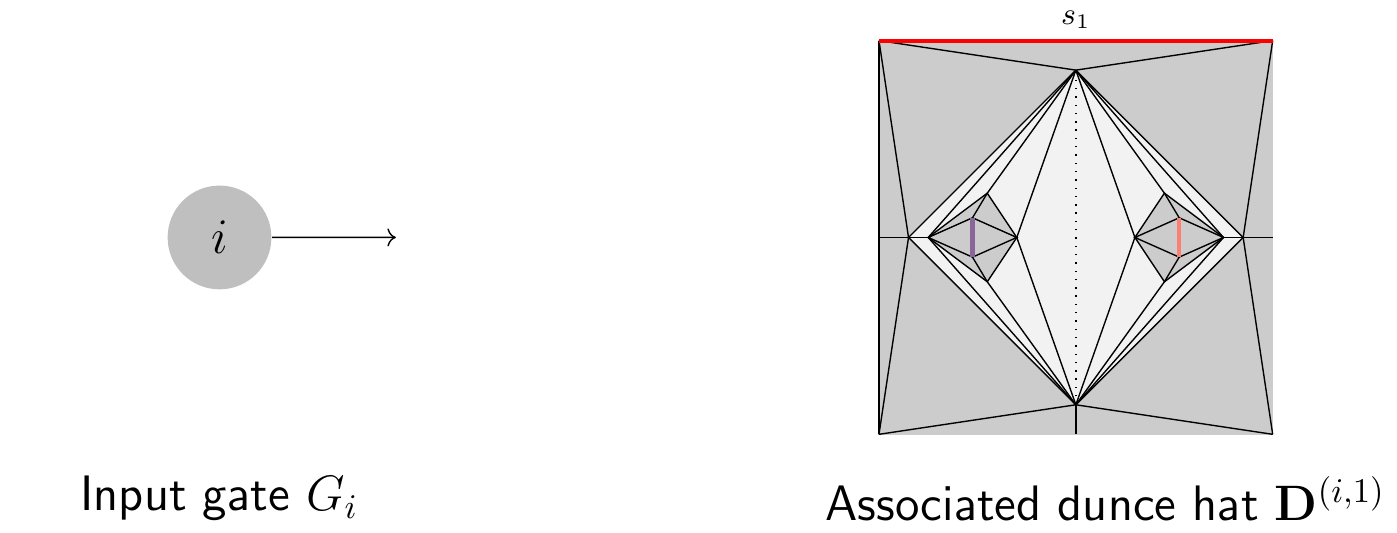} \\
\vspace{1cm}
\includegraphics[scale=0.72]{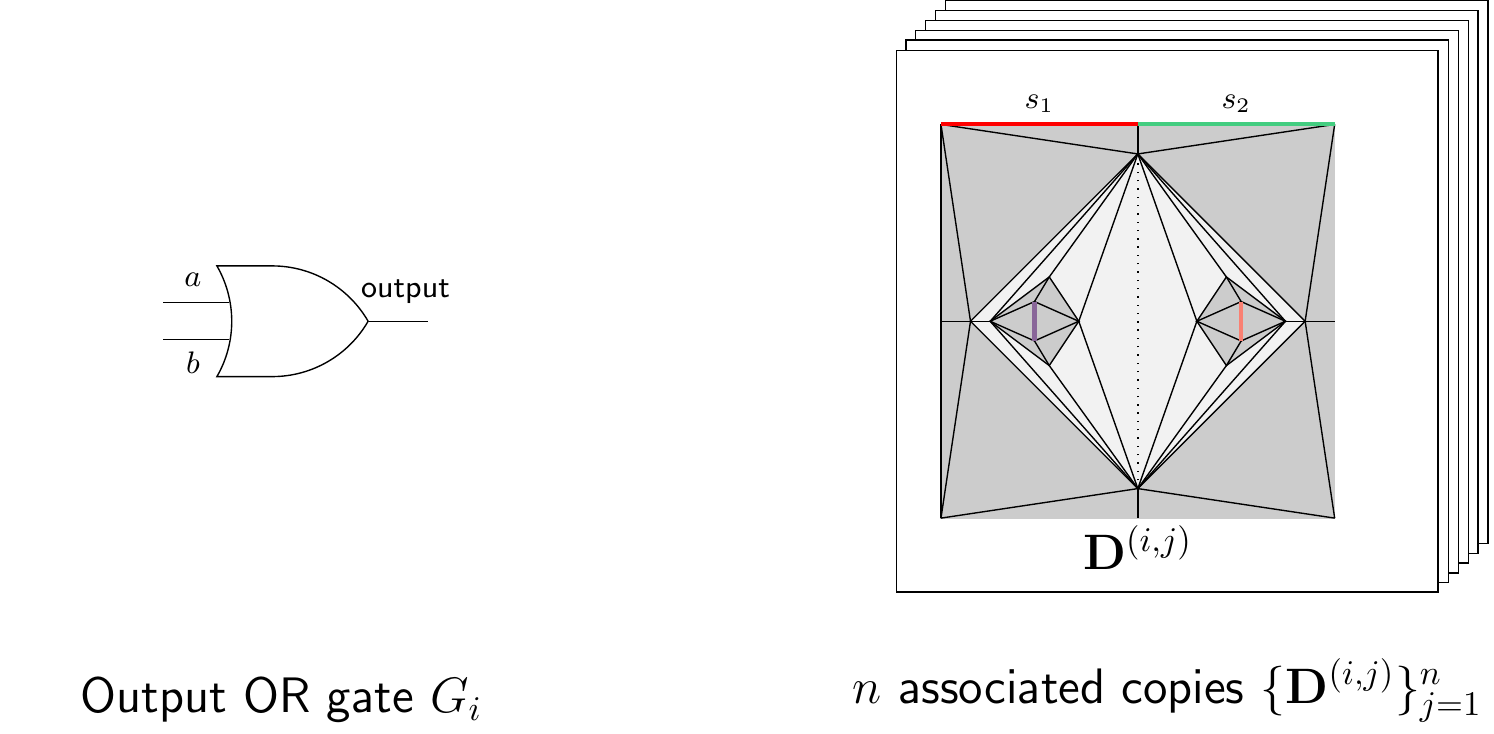} \\
\vspace{1cm}
\includegraphics[scale=0.72]{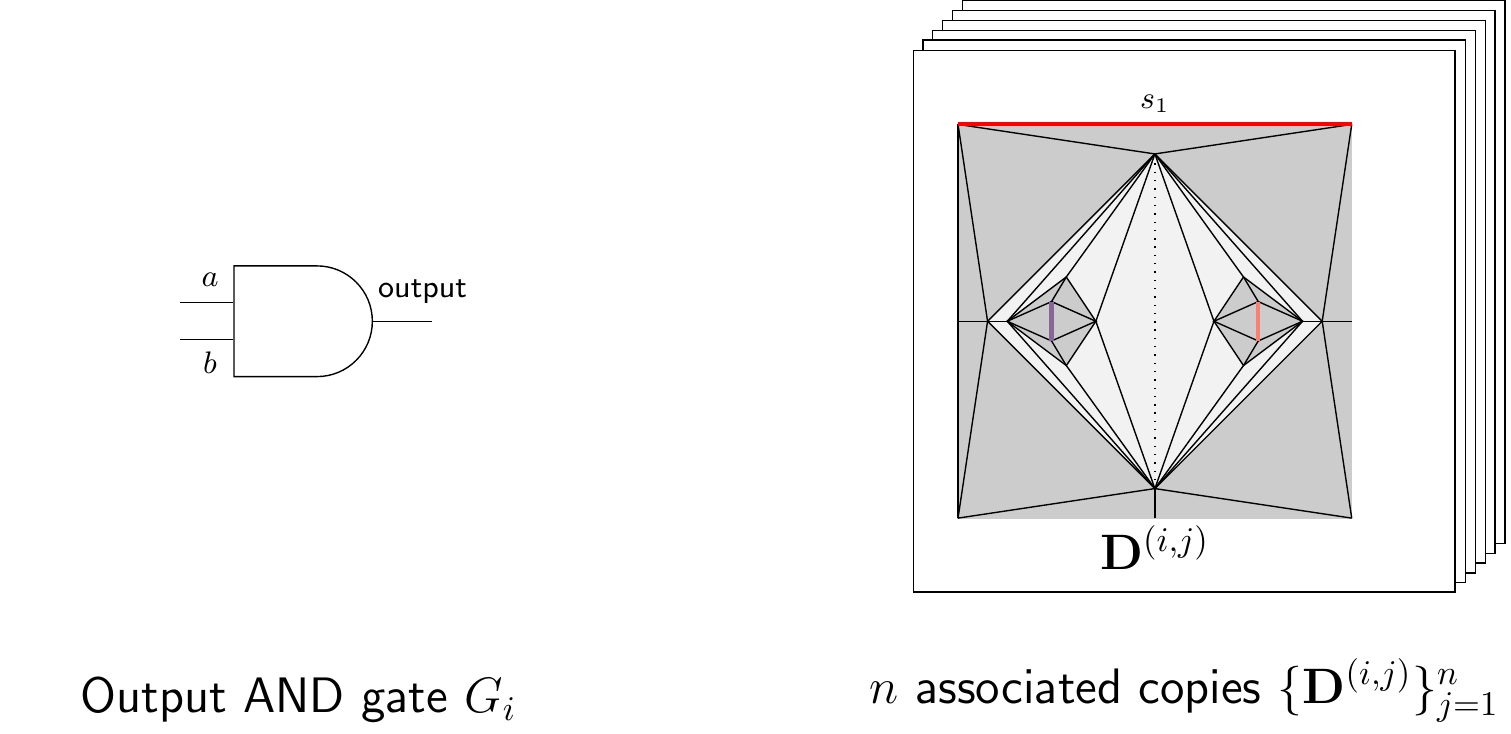} 

 \caption{ In all three figures, the distinguished  edges are highlighted. The top figure shows an input gate and the  dunce hat associated to it. 
 We conceive the input gate as activated when the associated dunce hat has a critical $2$-simplex in it. If the  dunce hat doesn't have critical $2$-simplices, then  $s_1$  must be paired to its coface for the dunce hat to be erased. The edge $s_1$ supports a feedback mechanism. In particular, if all the dunce hats associated to the output gate are erased without activating $G_i$, then we need an \emph{alternative} means to erase $G_i$, which is provided by $s_1$.
 The figure in the middle (resp. bottom) shows an output or-gate (resp. an output and-gate) and the associated $n$ copies of dunce hats. In both cases, the $j$-th copy consists of a single dunce hat $\gadget^{(i,j)}$,
  where $j \in [1,n]$. 
   The idea behind the dunce hat associated to the or-gate is that if either  $s_1$ or  $s_2$  is free, then $ \gadget^{(i,j)}$ can be erased. The idea behind the dunce hat associated to the and-gate is that if $s_1$  is  free, then $ \gadget^{(i,j)}$ can be erased. Finally, we have $n$ copies instead of a single copy per gate to ensure that optimum values of $\mincircuitsat$ and $\minrmm$ are the same. 
   \label{fig:opgates}}
\end{figure}

\begin{figure}
\centering
\includegraphics[scale=0.64]{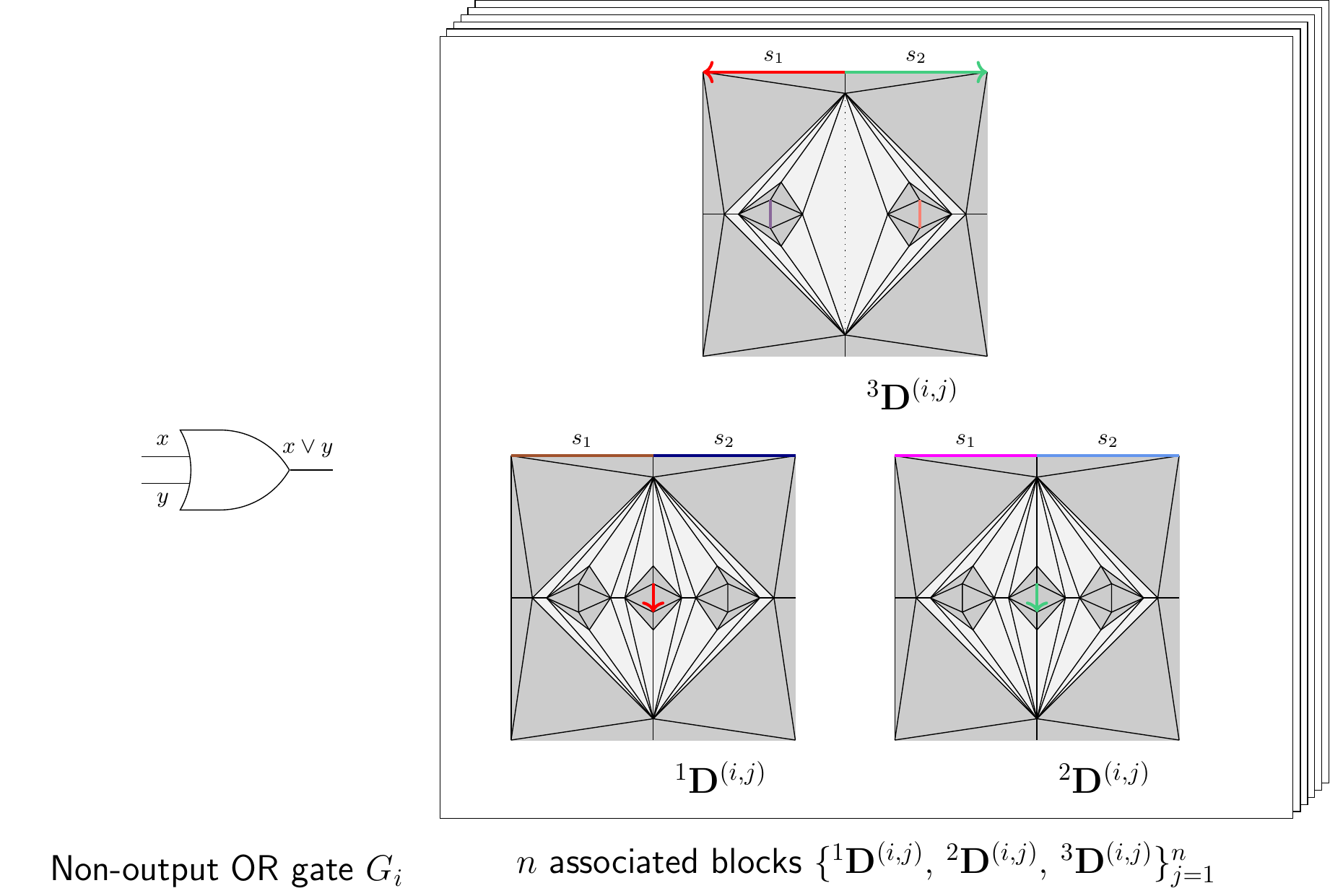} \\
\vspace{5mm}
\includegraphics[scale=0.64]{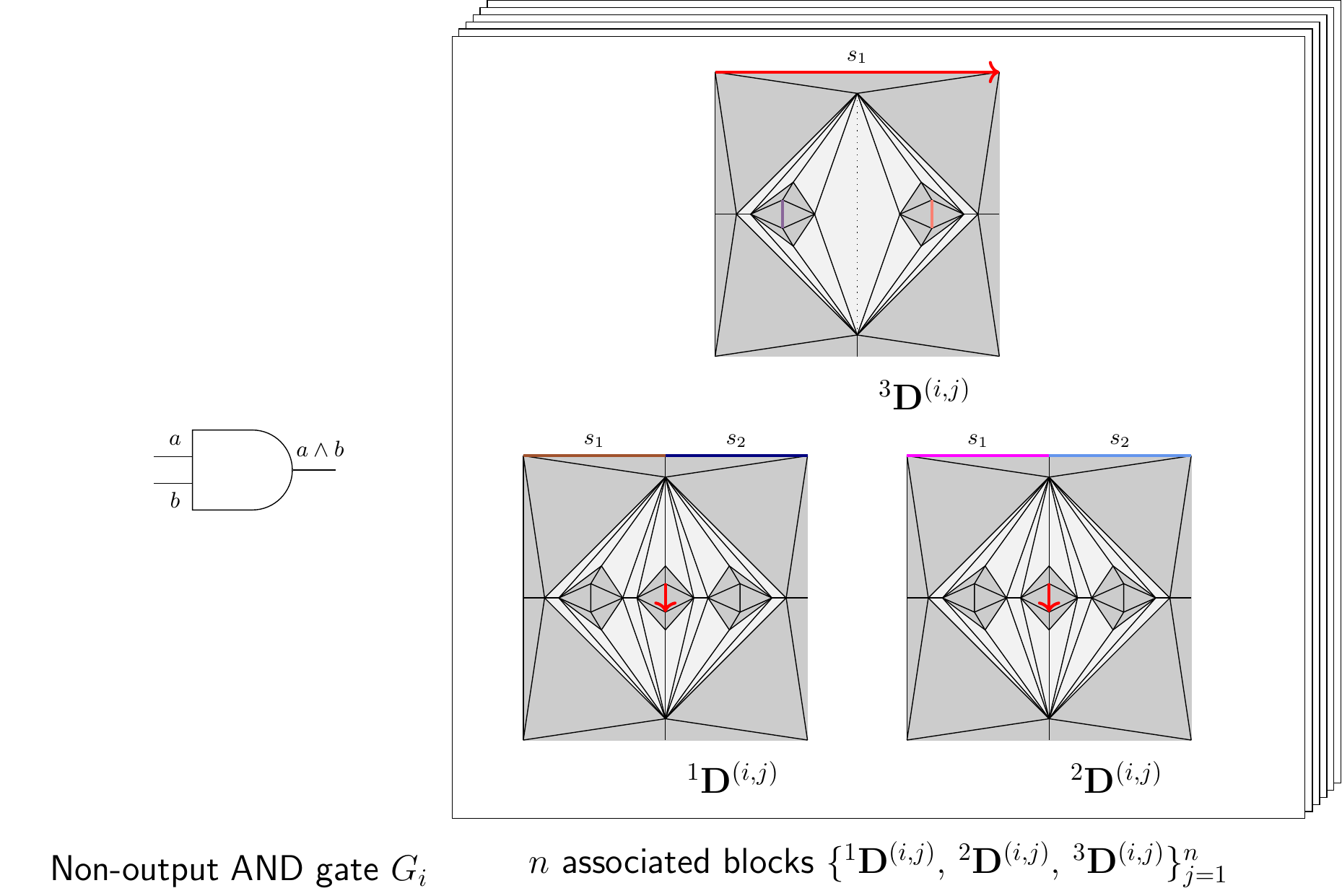} 

 \caption{ The figure on the top (resp. bottom) shows a non-output or-gate (resp. a non-output and-gate) and the associated $n$ blocks of dunce hats. In both cases $j$-th block consists of 3 dunce hats $\{^{1}\gadget^{(i,j)},\, ^{2}\gadget^{(i,j)}, \, ^{3}\gadget^{(i,j)} \}$,
  where $j \in [1,n]$. All distinguished  edges are highlighted, and identical color coding indicates identifications. That is, red edges are glued to red edges and green to green. The arrows on the highlighted edges show the orientations of identifications.
   The idea behind the blocks associated to the or-gate is that if either the $s_1$ edge of $ ^{1}\gadget^{(i,j)}$ or the $s_1$ edge of $ ^{2}\gadget^{(i,j)}$ is free, then all three dunce hats in the $j$-th block can be erased. The idea behind the blocks associated to the and-gate is that if the $s_1$ edge of $^{1}\gadget^{(i,j)}$ and the $s_1$ edge of $ ^{2}\gadget^{(i,j)}$ are free, then all three dunce hats in the $j$-th block can be erased. For each block, the dark and the light blue  $s_2$ edges of $ ^{1}\gadget^{(i,j)}$ and $ ^{2}\gadget^{(i,j)}$ respectively support a feedback mechanism. In particular, if the dunce hats associated to the output gate are erased then, we need an \emph{alternative} means to erase all the dunce hats, since the satisfaction of the output gate is all we really care about. Finally, we have $n$ blocks instead of a single block per gate to ensure that optimum values of $\mincircuitsat$ and $\minrmm$ are the same. 
   \label{fig:nonoutputgates}}
\end{figure}

\begin{figure}
\centering
\includegraphics[scale=0.6]{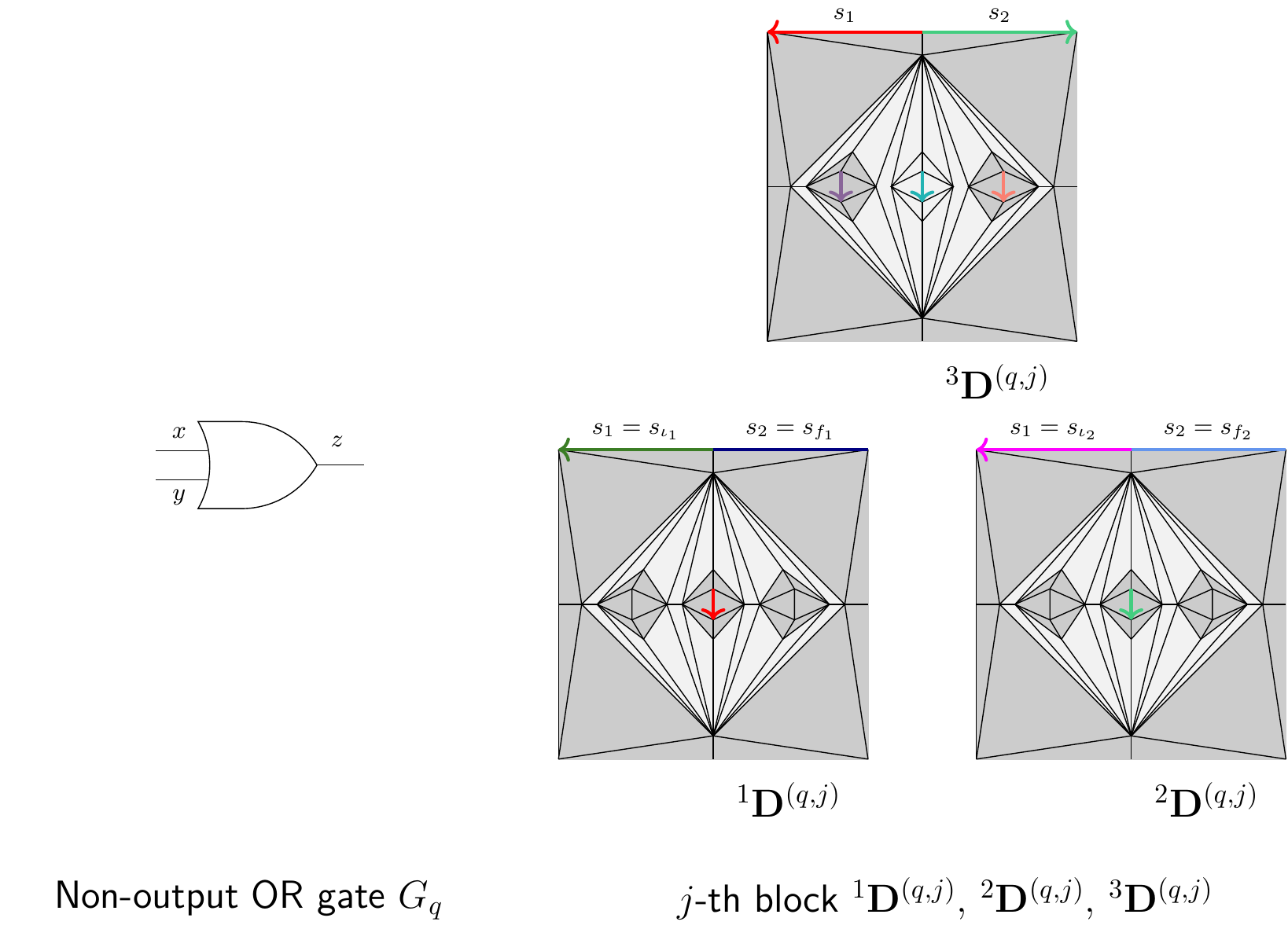} 
\vspace{1cm}
\includegraphics[scale=0.6]{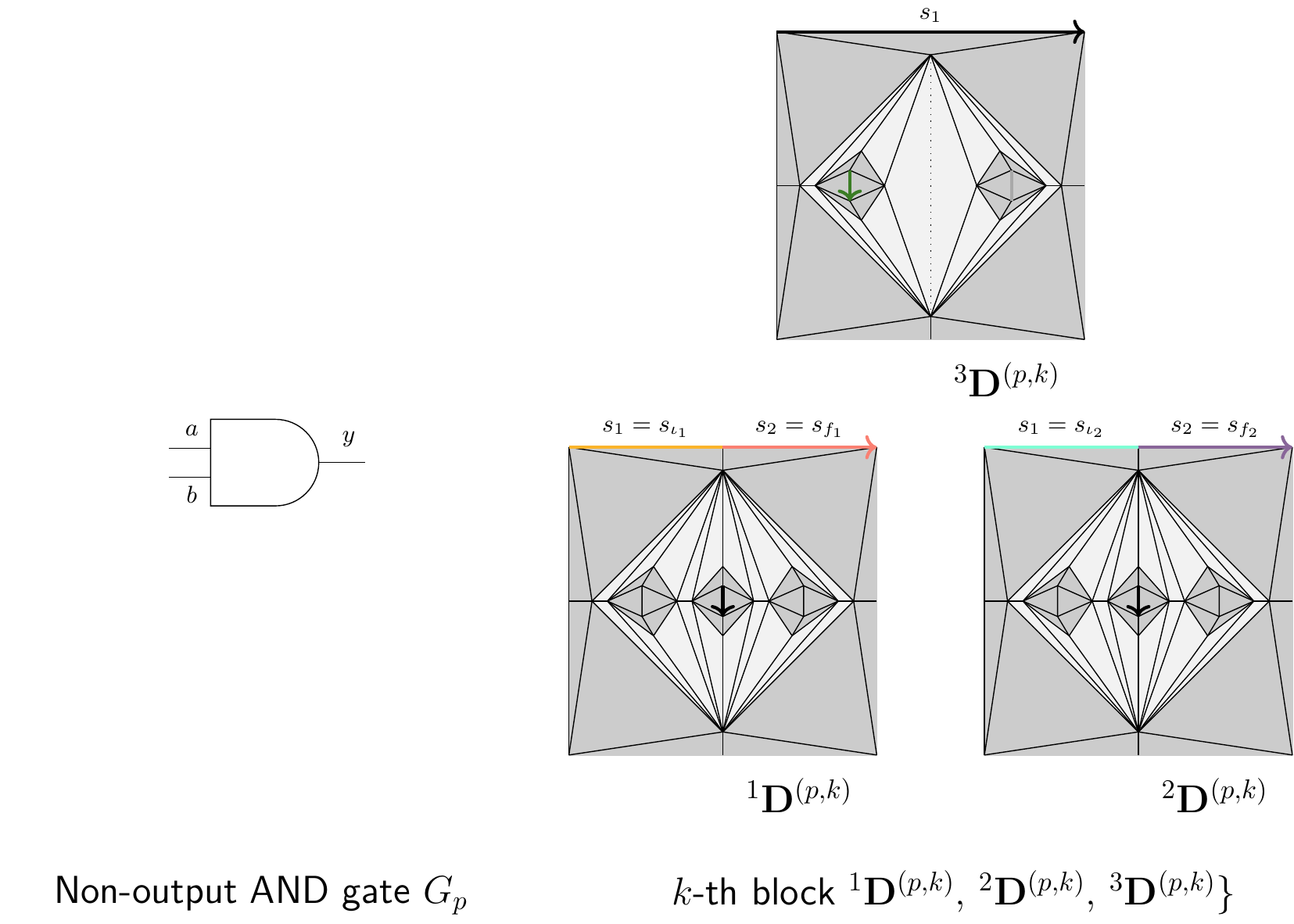} 
\vspace{1cm}
\includegraphics[scale=0.6]{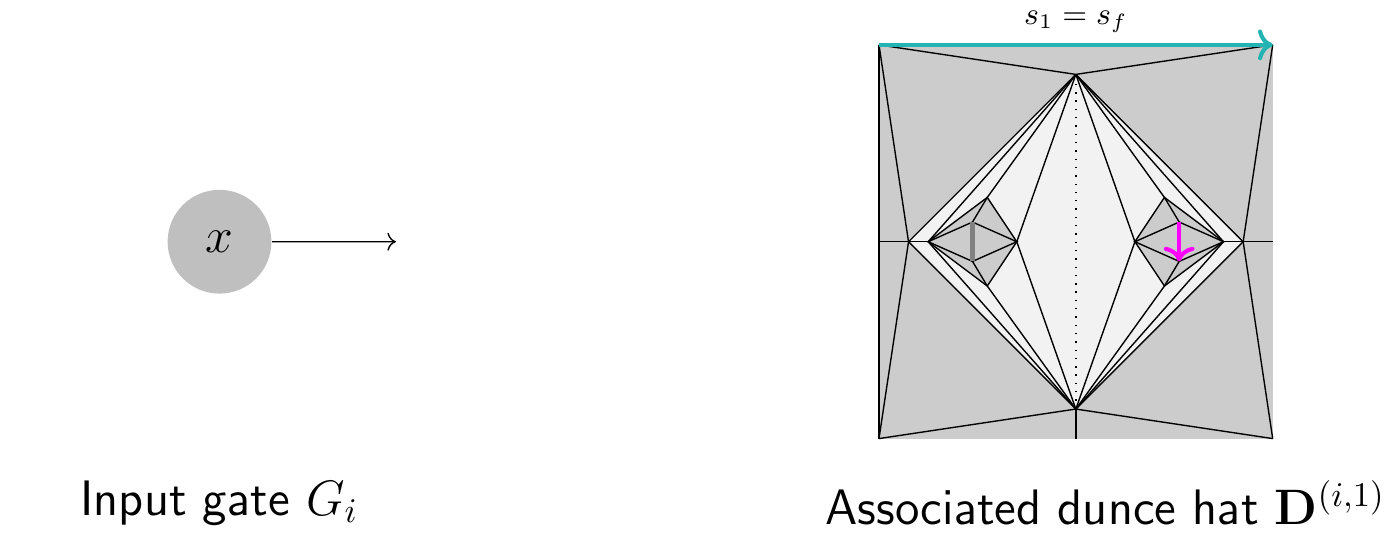} 
 \caption{ In this figure we depict the part of the complex  associated to the (partial) circuit that implements $z = (a \wedge b) \vee x$, where $x$ is an input to the circuit. Identical color coding indicates identifications, and the arrows indicate orientations of idenitifications. Here we only show identifications for $j$-th block of $G_q$ and $k$-th block of $G_p$ for arbitrary $j,k \in [1,n]$. Similar identifications occur across all respective associated blocks.
   \label{fig:nonoutputgates-glue}}
\end{figure}

It is important to note that the gluing is done so that the $s$-edges from two different copies (blocks) associated to the same gate are \emph{never} identified as an outcome of gluing, nor do they intersect in a vertex. In particular if $G_p$ is a gate with $G_{p_1}$ and $G_{p_2}$ as inputs, where,  for instance, if $G_{p_1}$ is an input gate and $G_{p_2}$ is an ordinary gate, then for every $k \in [1,n]$, $s^{(p,k)}_{\iota_1}$ is identified to a unique $t$-edge from the dunce hat associated to $G_{p_1}$, and   $s^{(p,k)}_{\iota_2}$  is identified to $n$ $t$-edges each from a block associated to  $G_{p_2}$.  These are the only identifications for edges $s^{(p,k)}_{\iota_1}$ and $s^{(p,k)}_{\iota_2}$. For every non-output gate $G_p$, let $\theta_p$ denote the number of  successors of $G_p$. Then, for all $k \in [1,n]$,  $s^{(p,k)}_{f_1}$ and $s^{(p,k)}_{f_2}$ each have $\theta_p n$ identifications from $t$-edges coming from each of the blocks associated to each of the successors of $G_p$. These are the only identifications for $s^{(p,k)}_{f_1}$ and $s^{(p,k)}_{f_2}$. 
If $G_p$ is an input gate, then $s^{(p,1)}_{f}$ is identifies to  $\theta_p n$ $t$-edges from blocks associated successors of $G_p$.
Finally, the input $s$-edges of the $k$-th copy associated to the output gate $G_o$ is identified to either one or $n$ $t$-edges coming from dunce hats associated to predecessor gates, depending on whether the predecessor is an input gate or an ordinary gate. 
We refrain from providing indices for the identified $t$-edges as this would needlessly complicate the exposition.

For every non-input gate $G_i$, set $\phi_i = 1$ if the first input to $G_i$ is from an input gate, and set $\phi_i = 2n$ otherwise.
Similarly, for every non-input gate $G_i$, set $\psi_i = 1$, if the second input to $G_i$ is from an input gate, and set $\psi_i = 2n$ otherwise.

Now we can readily check the following: In our construction, for a dunce hat  $^{3}\gadget_{m,\ell}^{\ijpair}$ associated to an ordinary gate $G_i$, we have $m = 1$ or $m = 2$ (depending on whether it is an and-gate or an or-gate), and $\ell = \theta_i n + \phi_i + \psi_i$. 
For dunce hats $^{2}\gadget_{m,\ell}^{\ijpair}$ and $^{1}\gadget_{m,\ell}^{\ijpair}$ associated to an ordinary gate $G_i$, we have $m = 2$, and $\ell = 1$. 
For a dunce hat  $\gadget_{m,\ell}^{\ijpair}$ associated to an output gate $G_i$,  we have  $m = 1$ or $m = 2$, and $\ell =  \phi_i + \psi_i$.  
Finally, for the dunce hat  $\gadget_{m,\ell}^{(i,1)}$ associated to an input gate $G_i$,  we have  $m = 1$ , and $\ell = \theta_i n$.

\begin{remark}\label{rem:reindex}
We reindex the dunce hats described above using the indexing set $\Xi$.  
That is, for every dunce hat  in $K'(C)$ there exists a unique $\zeta \in [1, |\Xi|]$ such that $\gadget^{\zeta}_{m,\ell} $ identifies the dunce hat of interest. Sometimes in our exposition it is more convenient to refer to dunce hats with a single index as opposed to using two or three indices in the superscript. 
\end{remark}

 Let $\zeta$  be the indexing variable, and $\Xi$ the indexing set as described in \Cref{rem:reindex}.
 For a dunce hat $\gadget_{m,\ell}^{\zeta}$, we call the complex induced by the edges $\{\{1^{\zeta},2^{\zeta}\},$ $\{2^{\zeta},3^{\zeta}\},$ $\{2^{\zeta},6^{\zeta}\},$  $\{6^{\zeta},b_k^{\zeta}\},$  $\{b_k^{\zeta},z_k^{\zeta}\} \}$ (i.e., the pink edges of $\gadget_{m,\ell}$ in~\Cref{fig:gadgetpaths}) the \emph{stem} of $\gadget_{m,\ell}^{\zeta}$.
 Then, in  complex $K'(C)$, let $H$ be the $1$-dimensional subcomplex formed by the union of stems of $\gadget_{m,\ell}^{\zeta}$,  for all $\zeta \in [1, |\Xi|]$. We call $H$ the  \emph{stem of the complex $K'(C)$}.
It can be shown that a basis for the  first homology group of the complex $\kprime$ is supported by the edges in the stem of the complex. 
 The complex $K(C)$ is formed as follows: We first assemble the minimal cycle basis of the stem of the complex $K'(C)$ in a matrix $\mathbf{M}$ and then make this matrix upper-triangular. Then, each cycle (column in $\mathbf{M}$) is filled with a triangulated disk, giving us the desired complex $K(C)$. Please refer to~\Cref{sec:kc} for further details.
 
 \begin{remark}[Design choices for ordinary gates]
At this point we would like to remark that, in principle, one could construct a complex for circuits with arbitrary fan-ins wherein the or-gate and and-gate like behaviour can easily be implemented with a single (suitably sub-divided) dunce hat having two or more free edges. 
The problem with this approach is that it is much harder to control where the $1$-cycles in the complex appear, and this makes the cycle filling procedure far more technical. This motivates our approach to first pass to circuits with fan-in  two and then implement or-gates and and-gates with blocks of three instead of single dunce hats. As we shall see later, this leads to a straightforward instance-independent description of the $1$-homology basis of $K'(C)$, which in turn simplifies cycle filling.
\end{remark}

 Given a gradient vector field $\tildev$ on $\KCC$, we construct the map $\ICC(C,\tildev) $ as follows: For every input gate $G_i$ whose associated dunce hat has a critical $2$-simplex in $\tildev$, we set $\ICC(C,\tildev) (G_i) = 1$. 
 Please refer to~\Cref{sec:algalg} for further details.

\subsection{Construction of the complex $K(C)$} \label{sec:kc}

 From the identifications described in \Cref{sec:kcmain}, it is easy to check that $H = (V_H,E_H)$ is, in fact, a connected graph. Please refer to  \Cref{lem:hconnap} for a simple proof.
  The procedure for constructing $K(C)$ is described in \Cref{alg:buildkc}.

\begin{algorithm}
\caption{Procedure for constructing $K(C)$ from $K'(C)$ }\label{alg:buildkc}
\begin{algorithmic}[1]

	\State{$\kc \gets \kprime$.}
	\LineComment{Initially, $\kc$ consists only of simplices from $\kprime$.}
 	\LineComment{Compute a cycle basis $\BCC$ of $H$ with $\bbz_2$ coefficients as follows (Steps 4--13).}
	\State{$i=1$; $\BCC =  \emptyset$.}
	\State{Let $\prec$ be an arbitrary total order on the edges of $H$.}
	\While{$E_H$ is non-empty}
		\State{Successively remove every edge from $H$ that is incident on a vertex of degree $1$.}	
		\State{Choose an  simple cycle in $H$  incident on the highest indexed edge w.r.t. $\prec$. }
		\State{Denote the cycle by $z_i$, and the highest indexed edge by $e_i^1$.}
		\State{ $\BCC = \BCC \bigcup \{ z_i \}$.}
		\State{Remove $e_i^1$ from $H$.}
		\State{$i = i+1$.}
	\EndWhile
	\State{Assemble the basis vectors of $\BCC$ in a matrix $\mathbf{M}$, where $\prec$ is used to index the rows of $\mathbf{M}$, and the iterator variable $i$ from the while loop above is used to index the columns.}
	\LineComment{For every $z_i \in \BCC$, let  $n_i$ be the number of edges in $z_i$, and $\{e_i^j \mid j \in [n_i] \}$ denote the edges in $z_i$.}
	 \For{$i \gets 1,  |\BCC|$}
		 \State{Add a new vertex $v_i$ to  $\kc$.}
	\For{$j \gets 1, n_i$}
		\State{Add to  $\kc$ a $2$-simplex $\sigma_i^j = e_i^j * v_i$ for each edge $e_i^j$ of $z_i$.}
		\State{Add to  $\kc$ all of the faces of simplices $\sigma_i^j$.}
	\EndFor
 \EndFor
 \State{$\mathscr{D} \gets \kc \setminus \kprime$.}
 \State{\textsc{Return} $\kc, \mathscr{D}$.}
\end{algorithmic}
\end{algorithm}

\begin{remark}\label{rem:basicobs}
By construction,  the edges $e_i^1$, for $i \in [ |\BCC| ]$,  do not appear in the cycles $z_j$, for $j  > i$. Hence, $\mathbf{M}$ is upper-triangular.
\end{remark}
 
Note that there exists a polynomial time subroutine to implement Step 8 of \Cref{alg:buildkc}. In particular, because the edges of $H$ that are not incident on any cycles of $H$ are removed in Step 7, every edge of $H$ is incident on some simple cycle contained in a minimum cycle basis  (with unit weights on edges) of $H$.
 
 \begin{remark} \label{rem:topology}
  In the construction described in \Cref{alg:buildkc}, the star of the vertex $v_i$ may be viewed as a \enquote{disk} that fills the cycle $z_i$. See~\Cref{fig:cyclefill} for an illustration.
Furthermore, it can be shown that 
 \begin{itemize}
\item  the second homology groups $H_2(K'(C))$ and $H_2(K(C))$ are trivial,
\item  The cycles in $\BCC$  form a basis for $H_1(K'(C))$, 
 \item  $K(C)$ is contractible.
\end{itemize}
However, our hardness results can be established without proving any of the statements in \Cref{rem:topology}.
Having said that, it is important to bear in mind that the procedure of going from $K'(C)$ to $K(C)$ is, in fact, a $1$-cycle filling procedure. 
\end{remark}

To establish hardness results, we introduce some additional notation.
Given a monotone circuit $C=(\VCC,\ECC)$ let $\kc$ be its associated complex. Now let $\OPTA$ denote the optimal value of the $\mincircuitsat$ problem on $C$, and let $\OPTB$ denote the optimal value of the $\minrmm$ problem on $\kc$.
The value of the objective function $\ALGB$ is the number of critical simplices in $\VCC$ minus one; the value of the objective function $\ALGA$  is the Hamming weight of the input assignment.

 \subsection{Reducing $\mincircuitsat$ to $\minrmm$: Forward direction} \label{sec:relateopt}
 
Given a circuit $C$, suppose that we are given an  input assignment  $A$ that satisfies the circuit $C = (V(C),E(C))$.
Let $\satgates$ be the set of gates that are satisfied by the assignment, and let $I(\satgates)$ be the set of input gates that are assigned $1$. 
Clearly, $I(\satgates) \subset \satgates$, and also the output gate $G_o \in \satgates$.
Let $\unsatgates = V(C) \setminus  \satgates$ denote the set of gates that are not satisfied by the input $A$.
Clearly, the subgraph $C_{\satgates}$ of $C$ induced by the gates in $\satgates$ is a connected graph.
Also, since $C$ is a directed acyclic graph, the induced subgraph $C_{\satgates}$ is also directed acyclic.
Let $\prec_{\satgates}$ be some total order on $\satgates$ consistent with the partial order imposed by $C_{\satgates}$, and let $\prec_{C}$ be some total order on $V(C)$ consistent with the partial order imposed by $C$.

 Next, given an assignment $A$ on $C$, we describe how to obtain a gradient vector field $\VCC$ on $K(C)$.  We denote the complex obtained after $i$-th step by $K^{i}(C)$.
 
  \paragraph*{Step 1: Erase satisfied input gates} 
First, for every input gate $G_i \in \satgates$, we make $\Gamma_1^{(i,1)} $ critical. By \Cref{lem:easy}, this is akin to removing  $\Gamma_1^{(i,1)}$ from $\gadget^{(i,1)}$. 
 Next, we make all $s_f^{(i,1)}$ for all $G_i \in \satgates$ critical.
 We then use  \Cref{lem:remainerase} from  \Cref{sec:spofr} to erase all the dunce hats  $\gadget^{(i,1)}$ associated to  satisfied input gates $G_i$, giving $K^{1}(C)$.

 \paragraph*{Step 2: Forward collapsing}
Assume throughout Step 2 that the gates in  $\satgates$ are indexed from $1$ to $|\satgates|$ so that 
\[\text{ for all }G_i, G_j \in \satgates,\quad i < j \Leftrightarrow G_i \prec_{\satgates} G_j.\]

\begin{lemma} \label{lem:successiveforwardap} 
Let  $G_p \in \satgates \setminus I(\satgates)$.
Suppose that all the gates in $I(\satgates)$ have been erased, and for all gates $G_k \in \satgates \setminus I(\satgates)$ with $k< p$ the associated dunce hats $^{3}\gadget^{(k,r)}$ for all  $r \in [1,n]$  have  been erased. Then,  the dunce hats $^{3}\gadget^{(p,j)}$, for all  $j \in [1,n]$ associated to $G_{p}$ can be erased.
\end{lemma} 
\begin{proof} 
\setcounter{case}{0}

Let $G_{p_1}$ and $ G_{p_2}$  be inputs to $G_p$.
Assume without loss of generality that  $G_{p_1} , G_{p_2}$ are non-input gates. By our assumption on indexing, $p_1 < p$ and $p_2 < p$. 
By construction, the only identifications to  $s_{\iota_1}^{\pjpair} \in ~ ^{1}\gadget^{\pjpair}$ are from $t$-edges that belong to $^{3}\gadget^{(p_1,r)}$  for all  $r \in [1,n]$,
and the only identifications to $s_{\iota_2}^{\pjpair} \in ~ ^{2}\gadget^{\pjpair}$ are from  $t$-edges that belong to $^{3}\gadget^{(p_2,r)}$  for all  $r \in [1,n]$. We have two cases:

\begin{case}
Assume that  $G_{p}$ is a satisfied or-gate. Then, either  $G_{p_1} \in \satgates$ or $G_{p_2} \in \satgates$. Without loss of generality, we assume that $G_{p_1} \in \satgates$.
Then,  for all $j$, $s_{\iota_1}^{\pjpair} $  become free since, by assumption,  the dunce hats $^{3}\gadget^{(p_1,r)}$  associated to $G_{p_1}$ have been erased. So using~\Cref{lem:erasefree} from \Cref{sec:spofr}, for all $j$, $^{1}\gadget^{\pjpair}$ can be erased. 
For each $j$, the unique identification to $^{3}s_1^{\pjpair}$ is from a $t$-edge in  $^{1}\gadget^{\pjpair}$.
Hence, for all $j$, $^{3}s_1^{\pjpair}$ becomes free, making it possible to erase  $^{3}\gadget^{\pjpair}$ for all $j$.
\end{case}

\begin{case}
Now, assume that  $G_{p}$ is a satisfied and-gate. Then, both   $G_{p_1} \in \satgates$ and $G_{p_2} \in \satgates$. 
Thus, for all $j$, $s_{\iota_1}^{\pjpair} $  and $s_{\iota_2}^{\pjpair} $ become free since, by assumption, all the dunce hats $^{3}\gadget^{(p_1,p)}$ for all  $p \in [1,n]$ associated to $G_{p_1}$ and all dunce hats $^{3}\gadget^{(p_2,q)}$, for all  $q \in [1,n]$ associated to $G_{p_2}$ have been erased. So, using~\Cref{lem:erasefree} from \Cref{sec:spofr}, for all $j\in [1,n]$, $^{1}\gadget^{\pjpair}$ 
and $^{2}\gadget^{\pjpair}$  can be erased. 
For all $j \in [1,n]$, the only two edges identified to $^{3}s_1^{\pjpair}$ belong to $^{1}\gadget^{\pjpair}$ 
and $^{2}\gadget^{\pjpair}$ respectively.
Hence, for all $j\in [1,n]$, $^{3}s_1^{\pjpair}$ becomes free, making it possible to erase  $^{3}\gadget^{\pjpair}$ for all $j \in [1,n]$.
Thus, the dunce hats $^{3}\gadget^{\pjpair}$ for  $j\in [1,n]$  associated to $G_{p}$ can be erased. 
\end{case}

The argument is identical for the case when $G_{p_1}$ or $G_{p_2}$ is an input gate.
\end{proof}

\begin{lemma} \label{lem:lasteraseap} All dunce hats associated to the output gate are erased.
\end{lemma}
\begin{proof} Note that a satisfying assignment $A$ that satisfies the circuit, in particular, also satisfies the output gate. 
A simple inductive argument using \Cref{lem:successiveforwardap} proves the lemma.
\end{proof}

After applying Step 1, we apply Step 2, which comprises of executing the collapses described by \Cref{lem:successiveforwardap,lem:lasteraseap}. This immediately gives us the following claim.
 \begin{claim}\label{cl:sateraseap} If there exists an assignment satisfying a circuit $C$ with Hamming weight $m$, then there exists a gradient vector field on $K(C)$ such that after making $m$  $2$-cells critical, all the dunce hats $^{3}\gadget^{\pjpair}$ associated to the satisfied non-output gates $G_p$ and all dunce hats associated to the output gate  can be erased.
 \end{claim}

 The complex obtained after erasing executing Step 2 is denoted by $K^{2}(C)$. 
 We have, $K^{1}(C) \searrow K^{2}(C)$.

 \begin{remark}
 Note that the forward collapses do not erase all the dunce hats associated to satisfied gates. For instance, for a satisfied or-gate $G_p$, if one of the input gates, $G_{p_1}$, is satisfied and the other, $G_{p_2}$, is not, then $^{1}\gadget^{\pjpair}$ and $^{3}\gadget^{\pjpair}$ will be erased, but $^{2}\gadget^{\pjpair}$ will  not be erased. The dunce hats associated to the unsatisfied gates and the unerased dunce hats associated to the satisfied gates are erased in the next step while executing the backward collapses.
 \end{remark}
 
 \paragraph*{Step 3: Backward collapsing}
Assume throughout Step 3 that the gates in  $V(C)$ are indexed from $1$ to $n$ so that 
\[\text{ for all }G_i, G_j \in V(C),\quad i < j \Leftrightarrow G_i \prec_{C} G_j.\]
The idea behind backward collapsing is that the feedback edges become successively free when one starts the collapse from dunce hats associated to highest indexed gate and proceeds in descending order of index.

 \begin{lemma} \label{lem:nexteraseap} If all the dunce hats associated to gates $G_k$, where $k>i$, have been erased, then the dunce hats associated to $G_i$ can be erased.
\end{lemma} 
\begin{proof} 
We have three cases to verify:
\setcounter{case}{0}

\begin{case}
First, assume that $G_i$ is an ordinary gate. The only identifications to edges $s_{f_1}^{\ijpair} \in ~ ^{1}\gadget^{\ijpair}$ and $ s_{f_2}^{\ijpair} \in ~ ^{1}\gadget^{\ijpair}$  respectively are from the $t$-edges in dunce hats associated to successors of $G_i$.
By assumption, all dunce hats $^{3}\gadget^{\kppair}$ associated to ordinary gates $G_k$ where $k>i$ have been erased, and all dunce hats $\gadget^{(o,q)}$ associated to the output gate $G_o$ have been erased. Hence,  $s_{f_1}^{\ijpair}$  and  $s_{f_2}^{\ijpair} $ are free, for every $j$. Therefore, for every $j$,  
dunce hats $^{1}\gadget^{\ijpair}$ and $^{2}\gadget^{\ijpair}$ can be erased. 
\end{case}

\begin{case}
If  $G_i$ is an unsatisfied gate, then for all $j$, the only identifications to $s$-edge(s) of $^{3}\gadget^{\ijpair}$ are from $t$-edges of $^{1}\gadget^{\ijpair}$ and $^{2}\gadget^{\ijpair}$. So the $s$-edge(s) of $^{3}\gadget^{\ijpair}$ become free for all $j$, allowing us to erase $^{3}\gadget^{\ijpair}$,  for all $j$. Thus, all dunce hats associated to $G_i$ can be erased.
\end{case}

\begin{case}
Now assume that $G_i$ is an input gate. Then, the unique $s$-edge of the unique copy associated to $G_i$ is identified to $t$-edges of dunce hats associated to successors of $G_i$.
Since, by assumption, all dunce hats associated to  gates $G_k$, where $k>i$, have been erased, $s_{1}^{\ionepair}$ becomes free, allowing us to erase $\gadget^{(i,1)}$.
\end{case}
 Note that in the proof of this lemma, for ordinary satisfied gates only Case 1 may be relevant, whereas for ordinary unsatisfied gates both Case 1 and Case 2 apply.
\end{proof}

 \begin{claim}\label{cl:unsateraseap} If there exists an assignment satisfying a circuit $C$ with Hamming weight $m$, then there exists a gradient vector field on $K(C)$ with exactly $m$ critical $2$-cells.
 \end{claim}
 \begin{proof} We prove the claim by induction. The base step of the induction is provided by~\Cref{cl:sateraseap}.
Then, we repeatedly apply the steps below until all gates in $\kc$  are erased: 
\begin{enumerate}
\item  Choose the highest indexed gate whose associated dunce hats haven't been erased.
 \item Apply the collapses  described in~\Cref{lem:nexteraseap} to erase  dunce hats associated to $G_k$.
 \end{enumerate}
 \end{proof}
 
  The complex obtained after erasing all dunce hats in $\kc$ is denoted by $K^{3}(C)$. 
 We have, $K^{1}(C) \searrow K^{2}(C) \searrow K^{3}(C)$.
 
 \paragraph*{Step 4: Deleting critical $1$-simplices}
 
 Note that in complex $K^{3}(C)$, the  $s$-edges $s_{f}^{(i,1)}$ from $\gadget^{(i,1)}$,  for all $G_i \in \satgates$ have no cofaces.
 Since they were already made critical in Step 1, by \Cref{lem:easy}, we can delete  $s_1^{(i,1)}$ from $K^{3}(C)$ for all $G_i \in \satgates$, and continue designing the gradient vector field on the subcomplex  $K^{4}(C)$ obtained after the deletion.
 
 \paragraph*{Step 5: Removing dangling edges}
 
 Since the $2$-collapses executed in Steps 1-3 are as described in~\Cref{rem:gradient}~and~\Cref{{fig:gadgetpaths}}, it is easy to check that for each $\maingadgetx \subset \kc$, the edges that remain are of the form:
$ \{ \{1,2\}^{\zeta_1},$   $\{2,3\}^{\zeta_1},$   $\{2,6\}^{\zeta_1},$   $\{5,6\}^{\zeta_1},$  $\{7,6\}^{\zeta_1},$  $\{b_k, 6\}^{\zeta_1},$  $\{c_k,6\}^{\zeta_1},$  $\{d_k,6\}^{\zeta_1},$  $\{b_k,z_k\}^{\zeta_1},$  $\{z_k,y_k\}^{\zeta_1},$  $\{y_k,a_k\}^{\zeta_1} \}$  $\cup \FCC$
for $k \in [1,\ell]$, and,
 \begin{compactitem}
 \item $\FCC =  \{ \{v,4\} \}$ if $m = 1$, where $v = 6$, if $\ell$ is even, and $v = \nicefrac{(\ell - 1)}{2} + 1$ if $\ell$ is odd, 
 \item $\FCC = \{s_1^{\zeta_1}, \{4,8\}^{\zeta_1}\}$ if $m=2$ and  $s_2^{\zeta_1}$ is removed as part of a $2$-collapse,    
 \item $\FCC = \{s_2^{\zeta_1}, \{4,8\}^{\zeta_1}\}$ if $m=2$ and  $s_1^{\zeta_1}$ is removed as part of a $2$-collapse.
 \end{compactitem}
 
 We now execute the following $1$-collapses (1-3 highlighted in green, and 4 highlighted in blue as illustrated in~\Cref{{fig:gadgetpaths}}).

  \begin{enumerate}
 \item Since  $5^{\zeta_1},$  $7^{\zeta_1},$   ${c_k}^{\zeta_1},$ and  ${d_k}^{\zeta_1}$   are free for all $k \in [1,\ell]$, for all  \mbox{$\maingadgetx \subset \kc$}, we execute the following collapses  for all $k \in [1,\ell]$, for all  \mbox{$\maingadgetx \subset \kc$}:
 \[(5^{\zeta_1}, \{5,6\}^{\zeta_1}), \quad  (7^{\zeta_1},\{7,6\}^{\zeta_1}),\quad (c_k^{\zeta_1},\{c_k,6\}^{\zeta_1}), \quad \text{ and } \quad (d_k^{\zeta_1},\{d_k,6\}^{\zeta_1}).\]
 \item  Since  the vertices $4^{\zeta_1}$ are free for all  \mbox{$\maingadgetx \subset \kc$}, for all  \mbox{$\maingadgetx \subset \kc$}:
 \begin{itemize}
 \item  if $m = 1$, we execute the collapse $(4^{\zeta_1}, \{4,v\}^{\zeta_1})$, where $v = 6$, if $\ell$ is even, and $v = \nicefrac{(\ell - 1)}{2} + 1$ if $\ell$ is odd,
 \item if $m=2$, we execute the collapse $(4^{\zeta_1}, \{4,8\}^{\zeta_1})$.
 \end{itemize}
 \item Now, ${a_k}^{\zeta_1}$  become free for all $k \in [1,\ell]$, for all  \mbox{$\maingadgetx \subset \kc$}. So, we execute the collapses  $({a_k}^{\zeta_2},\{a_k,y_k\}^{\zeta_2})$  for all $k \in [1,\ell]$, for all  \mbox{$\maingadgetx \subset \kc$}.
 \item Now, ${8}^{\zeta_1}$  become free for all  \mbox{$\gadget_{2,\ell}^{\zeta_1} \subset \kc$}. 
 So,  for all  \mbox{$\gadget_{2,\ell}^{\zeta_1} \subset \kc$}:
 \begin{itemize} 
 \item If $s_2^{{\zeta_1}}$ was removed as part of a $2$-collapse, we execute the collapse $({8}^{\zeta_1}, s_1^{{\zeta_1}})$,
\item  else if $s_1^{{\zeta_1}}$ was removed as part of a $2$-collapse, we execute the collapse $({8}^{\zeta_1}, s_2^{{\zeta_1}})$.
 \end{itemize} 
 Note that because of the identifications, there may exist several  \mbox{$\altgadgetx \subset \kc$} with  points $ {y_k}^{\zeta_2} \in \altgadgetx$ that are identical to $8^{\zeta_1}$. 
 So, the above collapses $({8}^{\zeta_1}, s_r^{{\zeta_1}}), r\in [1,2]$ may appear as  $ ({y_k}^{\zeta_2},   {\{y_k, z_k\}}^{\zeta_2} )$ in other dunce hats \mbox{$\altgadgetx \subset \kc$}.
 \end{enumerate}
 
   The complex obtained after collapsing all the dangling edges is denoted by $K^{5}(C)$. 
 
 So far, we have, $K^{1}(C) \searrow K^{2}(C) \searrow K^{3}(C)$ and $K^{4}(C) \searrow K^{5}(C)$.

 
 
 \paragraph*{Step 6: Collapsing the cycle-filling disks}
 The  $1$-complex $H$ formed by the union of stems of $\gadget_{m,\ell}^{\zeta}$,  for all $\zeta \in [1, |\Xi|]$ described in~\Cref{sec:kc} is clearly a subcomplex of $K^{5}(C)$.
 Let $\mathscr{D} = \kc \setminus \kprime$ be the set described in \Cref{alg:buildkc} obtained while building $\kc$ from  $\kprime$.
It is, in fact, easy to check that $K^{5}(C) = H \sqcup \mathscr{D}$. 
 Next, we show that $H$ is a connected graph.
 
  \begin{lemma}  $H$ is connected. \label{lem:hconnap}
 \end{lemma}
 \begin{proof}
 First note that  for every $\zeta \in [1, |\Xi|]$, the stem of $\gadget_{m,\ell}^{\zeta}$ is connected.
 In particular, the stem of $\gadget_{m,\ell}^{\zeta}$ connects $1^{\zeta}$ and $3^{\zeta}$ to $z_k^{\zeta}$ all $k \in \ell $.
 
 Suppose $G_i$ and $G_j$ are two gates in $C$ such that  $G_i$  is the predecessor of $G_j$.
 Then, in every dunce hat associated to $G_i$, there exists a $t$-edge that is connected to an $s$-edge to every dunce hat associated to $G_j$.
 That is, for all $p, q \in [1,n]$ there exists a $z_k^{(i,p)}$ that  is identified to either $1^{(j,q)}$ or $3^{(j,q)}$.
 Thus, the stems of $G_i$ are connected to the stems of $G_j$. 
 Now, since $C$ itself is a connected directed acyclic graph, it follows that the complex $H$ which is the the union of stems of $\gadget_{m,\ell}^{\zeta}$,  for all $\zeta \in [1, |\Xi|]$ is also connected.
 \end{proof}

 Now, as in~\Cref{alg:buildkc}, let $\mathbf{M}$ be the matrix whose columns represent a basis $\BCC$ of the cycle space of $H$.
 The cycles $z_i$ of $\BCC$ are represented by columns $M^i$. Let $n_i$ denote the number of edges in $z_i$. 
 Let the vertices $v_i^j \in z_i, j \in [1,n_i]$ and the edges $e_i^j \in z_i, j \in [1,n_i]$ be indexed so that $e_i^1$ represents the lowest entry (that is the pivot) for column  $M^i$, and $v_i^{j}$ and $v_i^{j+1}$ form the endpoints of $e_i^j$.
 Simplices $\sigma_i^j$ are indexed so that the vertices incident on  $\sigma_i^j$ are $v_i^{j}$ and $v_i^{j+1}$ and $v_i$. Please refer to~\Cref{fig:cyclefill} for an example of a cycle $z_i$ with six edges.
 The procedure to collapse all the disks in $K^{5}(C) \subset \kc$ corresponding to cycles $z_i \in \BCC$ is  described in \Cref{alg:collapsedisks}.

  \begin{figure}[hbt]
   \centering{

  \begin{tikzpicture}[scale=5,cap=round]
  \newdimen\R
  \newdimen\Rnew
  \R=0.8cm
  \coordinate (center) at (0,0);
 \draw (0:\R)
     \foreach \x in {60,120,...,360} {  -- (\x:\R) }
              -- cycle (300:\R) node[below] {$v^1_i$}
              -- cycle (240:\R) node[below] {$v^2_i$}
              -- cycle (180:\R) node[left] {$v^3_i$}
              -- cycle (120:\R) node[above] {$v^4_i$}
              -- cycle (60:\R) node[above] {$v^5_i$}
              -- cycle (0:\R) node[right] {$v^6_i$};
  \draw { (0:\R) -- (60:\R) -- (center) -- (0:\R) } [fill=verylightgray];
  \draw { (60:\R) -- (120:\R) -- (center) -- (60:\R) } [fill=lightergray];
    \draw { (120:\R) -- (180:\R) -- (center) -- (120:\R) } [fill=verylightgray];
  \draw { (180:\R) -- (240:\R) -- (center) -- (180:\R) } [fill=lightergray];
    \draw { (240:\R) -- (300:\R) -- (center) -- (240:\R) } [fill=verylightgray];
  \draw { (0:\R) -- (300:\R) -- (center) -- (0:\R) }  [fill=lightergray];
    \node[below] at($(300:\R)!1/2!(240:\R)$) {$e^1_i$};
     \node[below left] at($(240:\R)!1/2!(180:\R)$) {$e^2_i$};
      \node[above left] at($(180:\R)!1/2!(120:\R)$) {$e^3_i$};
            \node[above ] at($(120:\R)!1/2!(60:\R)$) {$e^4_i$};
      \node[above right] at($(60:\R)!1/2!(0:\R)$) {$e^5_i$};
      \node[below right] at($(0:\R)!1/2!(300:\R)$) {$e^6_i$};
     
 \Rnew=0.0cm      
     \node at($(300:\R)!1/3!(240:\R)!1/3!(0:\Rnew)$) {$\sigma^1_i$};
          \node at($(240:\R)!1/3!(180:\R)!1/3!(0:\Rnew)$) {$\sigma^2_i$};
               \node at($(180:\R)!1/3!(120:\R)!1/3!(0:\Rnew)$) {$\sigma^3_i$};
                       \node at($(120:\R)!1/3!(60:\R)!1/3!(0:\Rnew)$) {$\sigma^4_i$};
                             \node at($(60:\R)!1/3!(0:\R)!1/3!(0:\Rnew)$) {$\sigma^5_i$};
       \node at($(0:\R)!1/3!(300:\R)!1/3!(0:\Rnew)$) {$\sigma^6_i$};

          \draw[thick,->] ($(300:\R)!1/2!(240:\R)$)--($(300:\R)!1/2!(240:\R)!1/3!(0:\Rnew)$);
            \draw[thick,->] ($(240:\R)!3/5!(0:\Rnew)$)--($(240:\R)!3/5!(0:\Rnew)!1/3!(180:\R)$);
         \draw[thick,->] ($(180:\R)!3/5!(0:\Rnew)$)--($(180:\R)!3/5!(0:\Rnew)!1/3!(120:\R)$);
           \draw[thick,->] ($(120:\R)!3/5!(0:\Rnew)$)--($(120:\R)!3/5!(0:\Rnew)!1/3!(60:\R)$);
            \draw[thick,->] ($(60:\R)!3/5!(0:\Rnew)$)--($(60:\R)!3/5!(0:\Rnew)!1/3!(0:\R)$);
       \draw[thick,->] ($(0:\R)!3/5!(0:\Rnew)$)--($(0:\R)!3/5!(0:\Rnew)!1/3!(300:\R)$);
       
         \draw[thick,->] (0:\Rnew)--($(0:\Rnew)!1/2!(300:\R)$); 
       
\R=0.03cm
 \node[above right] at (0:\R) {$v_i$};

\end{tikzpicture}
}
 \caption{ The above figure shows a triangulated disk that fills the cycle $z_i$. Here, $e^1_i$ is the pivot edge of $z_i$. The gradient field starts with a gradient pair that includes the pivot edge.
   \label{fig:cyclefill}}
  \end{figure}
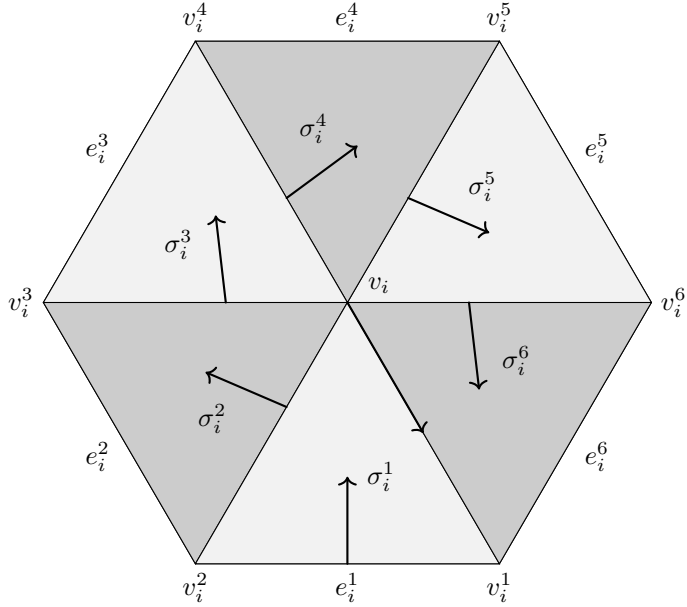

\begin{algorithm}
\caption{Procedure for collapsing cycle-filling disks}\label{alg:collapsedisks}
\begin{algorithmic}[1]

 \For{$i \gets 1,  |\BCC|$}
 	\State{Execute the collapse $(e_i^1,\sigma_i^1).$}  \label{lst:line:maincollapse}
	\For{$j \gets 2, n_i$}
		\State{Execute the collapse $(\{v_i,v_i^j\},\sigma_i^j).$}
	\EndFor
	\State{Execute the collapse $(v_i,\{v_i,v_i^1\})$.}
 \EndFor
 
 \State{\textsc{Return } $T$.}
\end{algorithmic}
\end{algorithm} 

 Note that in \Cref{alg:collapsedisks}, it is possible to execute the collapse $(e_i^1,\sigma_i^1)$ for each $i$ because the matrix $M$ of basis $\BCC$ is upper-triangular.
 This guarantees that after collapsing all the disks  corresponding to cycles $z_k, k \in [1,i-1]$,  $e_i^1$ is  free.
 Denote the complex obtained at the end of \Cref{alg:collapsedisks} as $T$.
 
 \paragraph*{Step 7: Collapsing the tree}

 Now observe that \Cref{alg:collapsedisks} removes all simplices in $\mathscr{D}$ from $K^{5}(C)$.
So, in particular, $T \subset H \subset K^{5}(C)$.
Moreover,  the pivot edges $e_i^1$ from cycles $z_i$ are also removed as part of $2$-collapses in Line~\ref{lst:line:maincollapse}~of~\Cref{alg:collapsedisks}.
In other words, $T = H \setminus  \bigcup_{i=1}^{|\BCC|}\{ e_i^1\}$, where $\BCC$ forms a basis for cycle space of $H$.
 
 \begin{claim} $T$ is a tree.
 \end{claim}
\begin{proof} By \Cref{lem:hconnap}, $H$ is connected. Removal of each edge $e_i^{1}$ from $H$, decreases the $\beta_1$ of $H$ by $1$, whereas $\beta_0$ of $H$ is unaffected.
Hence, $T$ is connected.
Moreover, since we destroy all $|\BCC|$ cycles of  $H$, $T$ has no cycles, proving the claim.  
\end{proof}

Next, we greedily collapse the tree $T$ to a vertex $v_{0} \in K(C)$, which can be done in time linear in the size of $T$. 
Finally, we make $v_{0} $ critical.
Let $\VCC$ be the collection of gradient pairs arising out of all the collapses from Steps 1-7. 
Also, note that $K^{1}(C)$ is obtained from $K(C)$ by deletion of $m$ critical $2$-simplices. 
Then, $K^{1}(C) \searrow K^{3}(C)$. Then, $K^{4}(C)$  is obtained from $K^{3}(C)$ by deleting $m$ critical  $1$-simplices. Then, $K^{5}(C)$ is obtained from $K^{4}(C)$ by executing some $1$-collapses. Finally, $K^{5}(C) \searrow T \searrow v_{0} $.
So, using~\Cref{lem:easy}, we conclude that given a circuit $C$ with a satisfying assignment $A$ of Hamming weight $m$, we can obtain a vector field  $\VCC$  on $\kc$ with $m$ critical $2$-simplices, $m$ critical $1$-simplices and a single critical vertex.
Now, for a circuit $C$ if the assignment $A$  is, in fact, optimal, that is, assuming $m = \OPTA  $, then $\minrmm$ for complex $\kc$ has a solution of size $(2m+1)-1$ giving us the following proposition.

\begin{proposition} \label{prop:appendoptrel}
 $ \OPTB  \leq 2\cdot \OPTA   $.
\end{proposition}

We highlight the entire collapsing sequence in~\Cref{fig:gadgetpaths}~$(a)$~and~$(b)$. First we perform the $2$-collapses as described in Steps 2-3. Then, the $1$-collapses for highlighted edges (in green) are executed. This is followed by $1$-collapses for highlighted edges (in blue), whenever these edges are available. These edges may not be  available
  if they are involved in $2$-collapses in other dunce hats, or if they are made critical. After executing the above collapses, for edges in green and blue, we first execute the $2$-collapses to erase all the cycle-filling disks, which leaves behind a tree supported by the edges in pink. The tree is then collapsed to a point.




\section{Hardness results for $\minrmorse$} \label{sec:mainminrmorse}

For maps $K$ and $\ICC$ described in \Cref{sec:kcmain}, we can establish the following relations.

\begin{proposition} \label{prop:optrel}
 $ \OPTB  \leq 2\cdot \OPTA   $.
\end{proposition}
\begin{proof}
For proof, please refer to \Cref{prop:appendoptrel} in  \Cref{sec:relateopt}.
\end{proof}

\begin{proposition} \label{prop:mineq}
$\ALGBB  \geq 2\cdot \ALGAA   $
\end{proposition} 
\begin{proof}
For proof, please refer to \Cref{prop:mineqappend} in  \Cref{sec:algalg}.
\end{proof}

\begin{proposition} \label{prop:equality}
$\OPTB = 2 \OPTA$.
\end{proposition}
\begin{proof}
For proof, please refer to  \Cref{prop:equalityappend} in \Cref{sec:algalg}.
\end{proof}

\begin{proposition}\label{prop:errorbound} 
\[ \ALGA - \OPTA   \leq \frac{1}{2} \ALGB - \OPTB  ) \]
\end{proposition}
\begin{proof}
Combining  \Cref{prop:mineq,prop:optrel} proves the claim.
\end{proof}

We will use the following straightforward fact about L-reductions. 
\begin{theorem}[Williamson, Shmoys~\cite{WS10}]\label{thm:newfactor} If there is an L-reduction with parameters $\mu$ and $\nu$ from a minimization problem $A$ to a minimization problem $B$, and there is a $(1+ \delta )$-approximation algorithm for $B$,
then there is a $(1+ \mu\nu\delta)$-approximation algorithm for $A$.
\end{theorem}

Next, we shall use the following result by Alekhnovich et al.~\cite{ABMP}.

\begin{theorem}[Theorem 3,\cite{ABMP}] \label{thm:abmp}
Unless $\NP\subseteq\QP$, there is no polynomial time algorithm which can approximate $\mincircuitsat$ within a factor of $2^{\log^{(1-\epsilon)}n}$, for any $\epsilon >0$. 
 \end{theorem}
 
 \begin{theorem} \label{thm:key}
   $\minrmm$  cannot be approximated within a factor of $2^{\log^{(1-\epsilon)}n}$, for any $\epsilon >0$ unless $\NP\subseteq\QP$.
\end{theorem}
\begin{proof} From \Cref{prop:mineq} and \Cref{prop:errorbound}, we conclude that the reduction from $\mincircuitsat$ to $\minrmm$ is a strict reduction with parameters $\mu=2$ and $\nu=\frac{1}{2}$. By Theorem~\ref{thm:newfactor}, if there exists a $(1+ \delta)$-approximation algorithm for $\minrmm$, then there exists a  $(1+ \mu\nu\delta)$-algorithm for $\mincircuitsat$. Using~\Cref{thm:abmp}, the claim follows.
\end{proof}

 Denote the standard parameterizations for $\mincircuitsat$ and $\minrmm$ by $\mincircuitsat'$ and $\minrmm'$ respectively.
 Using the map $K: C \to \kc$ that transforms instances of $\mincircuitsat$ to instances of $\minrmm$, we define a new map $\hat{K}: (C,k) \to (\kc,k')$ that transforms instances of $\mincircuitsat'$ to  instances of $\minrmm'$, where we set $k' = 2k$.

\begin{proposition} \label{prop:mainreduce} The map $\hat{K}$  from  $\mincircuitsat'$ to $\minrmm'$ is
\begin{enumerate}
 \item an FPT reduction,
 \item a $(\delta,\delta')$-gap preserving reduction for every function $\delta$, where $\delta'(k) = \delta(\left\lfloor \frac{k}{2} \right\rfloor)$.
 \end{enumerate}
\end{proposition}
\begin{proof} 
\begin{enumerate}
\item First note that, using~\Cref{prop:equality}, for any value of parameter $k$, 
\[\OPTA \leq k \Leftrightarrow \OPTB \leq 2k.\]
Then, the conclusion follows immediately from observing that  complex $\kc$ can be constructed in time polynomial in the size of $C$.
\item Suppose an instance $(C,k)$ is a $\delta$-gap instance of $\mincircuitsat'$. That is, either $ \OPTA \leq k$ or $\OPTA \geq k\delta(k)$. So, we have two cases to check:

Uusing~\Cref{prop:equality},
\begin{itemize}
\item  $\OPTA \leq k \Rightarrow \OPTB \leq 2k =  k'.$ 
\item If $\OPTA \geq k\delta(k)  \Rightarrow \OPTB \geq 2k\delta(k) = k' \delta( \frac{k'}{2} ) = k' \delta'(k').$ \qedhere
\end{itemize}
\end{enumerate}
\end{proof}




\begin{theorem} \label{thm:interthm}
\begin{enumerate} 
\item $\minrmm$ is \WP-hard.
\item $\minrmm$ has no fixed-parameter tractable approximation algorithm with any approximation ratio function $\rho$, unless $\fpt = \WP$.
\end{enumerate} 
\end{theorem}
\begin{proof}

The first statement follows immediately from~\Cref{prop:mainreduce}~and~\Cref{thm:wpcircuit}.

Eickmeyer et al.~\cite{Eickmeyer} provides a standard template to carry over FPT inapproximability results using gap preserving FPT reductions. 
Accordingly to prove the second statement, we closely follow the line of argument from \cite[Corollary 12]{Eickmeyer}. In this case, the strong FPT inapproximability result for $\mincircuitsat$ from \Cref{thm:bigguy} is carried over to
$\minrmm$. 
We first reduce $\mincircuitsat'$ to the approximation variant of $\minrmm'$.
Assume there exists an FPT cost approximation algorithm for $\minrmm$ with approximation ratio $\rho$, where $\rho$ is any computable function. 

Given an input $(C,k)$ for $\mincircuitsat'$, we first use the construction described in the proof of \cite[Theorem 6]{Eickmeyer}. Using this construction, we obtain a circuit $C$ of size $|C|=f(k)\cdot|C|^{O(1)}$ for some computable function $f$
in FPT time (with parameter k), such that
\begin{compactitem}
\item $(C',\alpha(k))$ is a $\delta$-gap instance for some $\alpha: \bbn \rightarrow \bbn $ and $\delta: \bbn \rightarrow \bbr_{>1}$,
\item and $\rho(2\alpha(k))<\delta(\alpha(k))$.
\end{compactitem}
Note that satisfying the second condition becomes possible since we have no restriction on the function $\delta$.

Using the FPT gap-preserving reduction described in \Cref{prop:mainreduce} from $\mincircuitsat$ to $\minrmm$ on the $\delta$-gap instance $(C,\alpha(k))$,
we get a $\delta'$-gap instance $(\kc,2\alpha(k))$ of $\minrmm'$ with $\delta'(2\alpha(k))=\delta(\alpha(k))$. 
We run $\mathbb{A}$ on $(\kc,\rho(2\alpha(k))\cdot 2 \alpha(k))$.
 
If $\OPTB \le 2 \alpha(k)$, then 
\[\rho(2\alpha(k)) \cdot 2\alpha(k)  \ge  \rho(\OPTB) \cdot \OPTB\]
and $\mathbb{A}$ accepts. 
If, on the other hand,  $\OPTB\ge \delta'(2\alpha(k))2\alpha(k) $ then
\[\rho(2\alpha(k)) \cdot 2\alpha(k) < \delta(\alpha(k)) \cdot 2\alpha(k) = \delta'(2\alpha(k)) \cdot 2\alpha(k)  \leq \OPTB, \]
and $\mathbb{A}$ rejects. 

Hence, using such an algorithm $\bba$ we could devise an FPT cost approximable algorithm for $\mincircuitsat$ some computable function $\rho$,
which in turn would imply $\WP = \fpt$ using~\Cref{thm:bigguy}.
\end{proof}

\section{Hardness results for $\minmorse$} \label{sec:mainminmorse}

Denoting the standard parameterizations for  $\minmm$ by  $\minmm'$, we now consider the map $\tilde{K}: (K,p) \mapsto (K,p+1)$ that transforms instances of $\minrmm'$ (simplicial complexes) to instances of $\minmm'$ (identical simplicial complexes).

\begin{proposition} \label{prop:dumbreduce} The map $\tilde{K}$  from  $\minrmm'$ to $\minmm'$ is
\begin{enumerate}
 \item an FPT reduction,
 \item a $(\delta,\delta')$-gap preserving reduction for every function $\delta$, where $\delta'(p) =\frac{\left(p-1\right)\delta\left(p-1\right)+1}{p}$.
 \end{enumerate}
\end{proposition}
\begin{proof} 
\begin{enumerate}
\item By definition, $\OPTD = \OPTC + 1$.
So, for any value of $p$, 
\[\OPTC \leq p \Leftrightarrow \OPTD \leq p+1.\]
So, the conclusion follows immediately.
\item Suppose an instance $(K,p)$ is a $\delta$-gap instance of $\minrmm'$. That is, either $ \OPTC \leq p$ or $\OPTC \geq p\delta(p)$. So, we have two cases to check:
\begin{itemize}
\item If $\OPTC \leq p$, then 
\[\OPTD \leq p+1 =  p'.\]
\item If $\OPTC \geq p\delta(p)$, then
\[\OPTD \geq p\delta(p) + 1= p' \delta(p').
\qedhere
\] 
\end{itemize}
\end{enumerate}
\end{proof}

Combining \Cref{thm:interthm} and \Cref{prop:dumbreduce}, we obtain the following result:

\begin{theorem} \label{thm:theoremone}
$\minmm$ is \WP-hard.
Furthermore, it has no fixed-parameter tractable approximation algorithm within any approximation ratio function $\rho$, unless $\fpt = \WP$.
\end{theorem}

\begin{definition}[Amplified complex] Given a pointed simplicial complex $K$ with $n$ simplices, the amplified complex $\madk$ is defined as the wedge sum of $n$ copies of~$K$. 
\end{definition}

\begin{lemma} \label{lem:ampopt}
For any $2$-complex K,  $\OPTDD=n \cdot \OPTC+1$.
\end{lemma}
\begin{proof} It is easy to check that the optimal vector field on $\madk$ is obtained by repeating the optimal vector field on $K$ on each of the $n$ copies of $K$ in $\madk$, while making the distinguished vertex of $\madk$
the unique critical vertex in $\madk$.
\end{proof}
 
 \begin{lemma} \label{lem:ampapx}
 Using a vector field $\madv$ on $\madk$ with $m + 1$ critical simplices, one can compute a vector field $\VCC$ on $K$ with at most $\left\lfloor  \frac{m}{n} \right\rfloor   + 1$ critical simplices in polynomial time.  
 \end{lemma}
 \begin{proof} Using \Cref{lem:uniqueCriticalVertex}, we can assume without loss of generality that  $\madv$ has the distinguished vertex as its unique critical simplex.
Restricting $\madv$ to each of the $n$ copies of $K$, the claim follows.
 \end{proof}

 \begin{proposition} \label{prop:tedious}
  For a fixed $\epsilon >0$, let  $\rho =f(n)$, where $f(n) = o(n)$. Then, for any $\delta  \in (0,1)$  and $\varrho =f(n) -\delta $, if there exists a $\varrho$-factor approximation algorithm for $\minmm$, then there exists a  $\rho$-factor approximation algorithm for $\minrmm$.
 \end{proposition}
 \begin{proof} 
 For a complex $K$, the optimal value of $\minrmm$ on $K$ is denoted by $\OPTC$.
 Suppose that there exists  a $\varrho$-factor approximation algorithm $\bba$ for $\minmm$. 
 If we apply $\bba$ on $\madk$, then using~\Cref{lem:ampopt}, we obtain a vector field with at most $\varrho\left(n \cdot \OPTC+1 \right )$ critical simplices.
 Then, using~\Cref{lem:ampapx}, we can  compute a vector field $\VCC$  on $K$ with at most $m(\VCC)$ critical simplices, where
\begin{align*} 
m(\VCC)  & \leq  \left\lfloor \frac{\varrho\cdot n\cdot\OPTC+\varrho-1}{n}\right\rfloor +1 && \textnormal{}\\
 		& \leq  \left\lfloor \varrho\cdot\OPTC +  \frac{\varrho-1}{n}\right\rfloor +1 && \textnormal{}\\
		& \leq   \left\lfloor  \varrho\cdot\OPTC \right\rfloor + \left\lfloor \frac{\varrho-1}{n}  \right\rfloor +  2 && \textnormal{using $\left\lfloor x+y\right\rfloor \leq\left\lfloor x\right\rfloor +\left\lfloor y\right\rfloor +1$}\\
		&  = \left\lfloor  \varrho\cdot\OPTC \right\rfloor +   2 && \textnormal{using $\left\lfloor \frac{\varrho-1}{n}  \right\rfloor = 0$ for large $n$,} \\
		\intertext{which gives us}
m(\VCC) -1  & \leq \varrho\cdot\OPTC + 1		&& \textnormal{} \\
m(\VCC) -1  & \leq \rho\cdot\OPTC - \delta\cdot\OPTC + 1		&& \textnormal{} \\
		   & \leq \rho\cdot\OPTC && \textnormal{assuming $\OPTC > \frac{1}{\delta}$.}	   
 \end{align*}
 The above analysis shows that one can obtain a $ \rho$-factor approximation algorithm for $\minrmm$ assuming a $\varrho$ factor approximation algorithm for $\minmm$. 		   
Note that $n^{\frac{1}{\delta}}$   is bounded by a polynomial in $n$ given the fact that $ \frac{1}{\delta}$ is a constant.
 So, we can assume without loss of generality that $\OPTC > \frac{1}{\delta}$ based on  the observation  by Joswig and Pfetsch~\cite{JP06} that if $\OPTC \leq c$, for some constant $c$, then one can find the optimum in $O(n^{c})$ time.
 \end{proof}

Combining \Cref{thm:key,prop:tedious}, we conclude  that for a fixed $\epsilon >0$, $\minmm$  cannot be approximated within a factor of $2^{\log^{(1-\epsilon)}n} - \delta$, for any $\delta >0$, unless $\NP\subseteq\QP$.
But, in order to get rid of the  $\delta$-term in the inapproximability bound for $\minmm$, we can do slightly better by allowing $\epsilon$ to vary. 
To make this precise, suppose there exists  an $\iota >0$ such that $\minmm$ can be approximated within a factor of $2^{\log^{(1-\iota)}n} $, and let $\delta \in (0,1)$.
Then, using \Cref{prop:tedious}, this  would give a $2^{\log^{(1-\iota)}n}  + \delta$ approximation algorithm for $\minrmm$. 
However, one can always find an $\epsilon>0$ such that $2^{\log^{(1-\iota)}n}  + \delta  =  O(2^{\log^{(1-\epsilon)}n}) $. 
Then, for sufficiently large $n$,  $2^{\log^{(1-\iota)}n}  + \delta  < 2^{\log^{(1-\epsilon)}n} $. 

Hence, the assumption of a $2^{\log^{(1-\iota)}n} $-factor approximation algorithm for $\minmm$ contradicts \Cref{thm:key}.
We can thus make the following claim.

 \begin{theorem} \label{thm:keytwo} For any $\epsilon >0$, $\minmm$  cannot be approximated within a factor of $2^{\log^{(1-\epsilon)}n} $, unless $\NP\subseteq\QP$.
  \end{theorem}

\section{An approximation algorithm for \minmorse} \label{sec:approxalgo}
In this section, we assume without loss of generality that the input complex $K$ is connected.
The algorithm can be described as follows. Given a $2$-complex $K$, let $n$ be the number of 2-simplices.
Assume without loss of generality that $\log n$ is an integer that divides $n$.
Partition the set of $2$-simplices of $K$ arbitrarily into $\log n$ parts each of size $\frac{n}{\log n}$.
Writing $S$ for the partition, we note that the power set $P(S)$ of the parts has $n$ elements.
The $2$-simplices that belong to a part $s \in S$ is denoted by $K_s^{(2)}$.
Each element of $P(S)$ gives us a subset $\hat{S}$ of $S$. To each $\hat{S}$ we can associate a binary incidence vector $\mathbf{j} (\hat{S})$ of length $\log n$ in the natural way. 
 Let $\khat$ be a complex induced  by the $2$-simplices belonging to the parts that belong to some   $\hat{S} \subset S$.
 In this case, we may also write $\khat$ as $\khat = K(\mathbf{j}(\hat{S}))$ to emphasize the data from which $\khat$ can be constructed.
Compute such a complex $\khat$ for each subset $\hat{S}$, and let $\hat{S}_{\max}$ be  the subset  of largest cardinality whose induced complex $\khat_{\max}$ is erasable.  
In particular, $\khat_{\max} \searrow L$ where $L$ is a $1$-complex.
Make all the $2$-simplices in $K \setminus \khat_{\max}$ critical. The gradient on $\khat_{\max}$ is comprised of the erasing gradient of  $\khat_{\max}$, namely $ \VCC^2 $, combined with the 
optimal gradient for $L$, namely $ \VCC^1$.
In what follows, we will show that this simple algorithm provides a $O(\frac{n}{\log n})$-factor approximation for  \minmorse on $2$-complexes.  

\begin{lemma} \label{lem:crittwo}
Let $\khatmax = K(\mathbf{j}(\hat{S}))$ for some $\hat{S}$. Let $w_{\bfj}$ be the Hamming weight of $\mathbf{j}(\hat{S})$, and let $\gamma = \log n - w_{\bfj}$. Then, every Morse matching on $K$ has at least $\gamma$ critical $2$-simplices.
\end{lemma}
\begin{proof}
Suppose that there exists a gradient vector field $\VCC^2$ with $\mu$ critical $2$-simplices where $ \mu< \gamma $. Let $\Psi$ denote the critical $2$-simplices of $\VCC^2$. 
Define $\khat_{new}^{\brat}$ as follows:
\[\khat_{new}^{\brat} = \bigcup\limits_{\substack{s \in S, \\ \Psi \cap K_{s}^{\brat} = \emptyset }}  K_{s}^{\brat} . \]
As before, let $\khat_{new}$ be the complex induced by simplices in $\khat_{new}^{\brat} $.
Then, $\khat_{new} \subset K \setminus \Psi$. However, $K \setminus \Psi$ is erasable via gradient $\VCC^2$. So, by \Cref{lem:subErasable}, $\khat_{new}$ is erasable. But this contradicts the maximality of $\khatmax$, proving the claim.
\end{proof}

We denote the critical $k$-simplices of $ \VCC^2 \cup \VCC^1 $ by $c_k$.

\begin{lemma} \label{lem:algcrit}
The gradient vector field $ \VCC^2 \cup \VCC^1 $ over $K$ has at most $\beta_1 - \beta_2 + 1 + 2\gamma \cdot \floornlogn$ critical simplices.
\end{lemma}
\begin{proof} From~\Cref{lem:crittwo}, we have 
\begin{equation} \label{eq:c2}
 c_2 = \gamma \cdot\left(\frac{n}{\log n}\right).
\end{equation}
By~\cite[Lemma~4.2]{JP06}, $\kone$ is connected, and one can compute a gradient vector field $\VCC^1$ on $\kone$ with a single critical vertex in linear time  using depth first search starting from an arbitrary vertex in $\kone$ (see, e.g., \cite{RBN17}).
We have by~\cite[Theorem 1.7]{Fo98},
\[c_0 - c_1 + c_2 = \beta_0 - \beta_1 + \beta_2.\]
Since $\beta_0 = c_0 = 1$, we have, 
\begin{equation} \label{eq:morseineq}
c_2-\beta_2 = c_1 - \beta_1.
\end{equation}
Thus, combining~\Cref{eq:c2} and \Cref{eq:morseineq}, we have 
\[c_1 = \beta_1 - \beta_2 +  \gamma \cdot \floornlogn.\]
The claim follows.
\end{proof}

\begin{theorem}The exists a $O(\frac{n}{\log n})$-factor approximation algorithm for  \minmorse on $2$-complexes.
\end{theorem}
\begin{proof} 
To begin with, we know from Tancer~\cite[Proposition 5]{MR3439259} that a $2$-complex is erasable if and only if greedily collapsing triangles yields a $1$-dimensional complex. That is, erasability of  a complex can checked in polynomial time. 
Since we check the erasability of  $O(n)$ complexes each of size $O(n)$, the algorithm terminates in polynomial time. 

Now, let $\VCC_{\min}$ be an optimal gradient vector field. By~\Cref{lem:crittwo}, $\VCC_{\min}$ has at least $\gamma$ critical $2$-simplices. 
By weak Morse inequalities~\cite[Theorem 1.7]{Fo98}, the number of critical $1$-simplices of $\VCC_{\min}$  is  at least $\beta_1$, and the number of critical $0$-simplices of $\VCC_{\min}$  is $\beta_0 = 1$.
Thus, an optimal gradient vector field has at least $O( \gamma + \beta_1)$ critical simplices.
Combining this  observation with \Cref{lem:algcrit}, it follows that  the algorithm described in this section provides an $O(\frac{n}{\log n})$-factor approximation for  \minmorse on $2$-complexes.
\end{proof}

\section{Morse matchings for Costa--Farber complexes}
\label{sec:costafarber}

The strong hardness results established in \Cref{sec:mainminmorse} belie what is
observed in computer experiments for both structured as well as random instances~\cite{RBN17,HMMNWJD10}.  In particular, the structured instances  generated by Lutz~\cite{Lutzlib,JLT14}
and by the RedHom and CHomP groups~\cite{HMMNWJD10}  and the random instances that come from Linial--Meshulam model~\cite{Mesh09} and the Costa--Farber model (referred to as type-2 random complexes in~\cite{RBN17})
turn out to be \enquote{easy} for Morse matching~\cite{RBN17}. We use the terms \enquote{easy} and \enquote{hard} in an informal sense.
Here, by \emph{easy} instances, we mean those instances for which simple heuristics give near-optimal matchings, and by \emph{hard} instances we mean instances for which known heuristics produce 
matchings that are far from optimal. Unfortunately, the approximation bounds in~\cite{RBN17}, and in \Cref{sec:approxalgo} of this paper do not explain the superior performance of simple heuristics in obtaining near-optimal matchings. 
Below, we provide some justification for this phenomenon from the lens of random complexes. We describe two types of gradients on the Costa--Farber complexes, namely,
\begin{enumerate}[(a)]
\item the apparent pairs gradient and
\item the random face gradient.
\end{enumerate}
Finally, we summarize the behavior of these gradients on Linial--Meshulam complexes.

\subsection{The apparent pairs gradient.}

We start with the definition of the \emph{apparent pairs gradient}~\cite{ripser}, which originates from work by Kahle on random complexes \cite{kahle2} and is used as a powerful optimization tool in software for computing persistent homology of Rips filtrations, like Ripser~\cite{ripser}
and Eirene~\cite{eirene}. 

Let $K$ be a $d$-dimensional simplicial complex, and let $V$ denote
the set of vertices in $K$. Suppose that the vertices in $V$ are
equipped with a total order $<$. For two simplices $\sigma,\tau \in K$, we write $\sigma  \prec_{K} \tau$
if $\sigma$ comes before $\tau$ in the lexicographic ordering.
%
\begin{absolutelynopagebreak}
Following~\cite{kahle2,ripser}, we call a pair of simplices $(\sigma,\tau)$ of $K$ an \emph{apparent pair} of $K$ (with respect to the lexicographic order on simplices) if both 
\begin{itemize}
\item $\sigma$ is the lexicographically highest facet of $\tau$, and 
\item $\tau$ is the lexicographically lowest cofacet of $\sigma$.
\end{itemize}
\end{absolutelynopagebreak}

As observed in~\cite{kahle2}, the collection of all the apparent pairs in  $\cfm$  forms a discrete gradient  on $\cfm$; see also~\cite[Lemma 3.5]{ripser}. We denote this gradient by $\VCC_1$.

Kahle \cite{kahle2} introduced the lexicographic apparent pairs gradient to construct matchings with provably few critical simplices on Vietoris-Rips complexes built on a random set of points in space.
An earlier variation of this construction \cite{kahle1} was used by the same author to study random clique complexes.
The  Morse numbers from this matching provides upper bounds on the Betti numbers of respective degrees.
In what follows, we observe that the apparent pairs gradient also provides very efficient matchings for an even larger class of random combinatorial complexes, namely, the multiparameteric complexes of Costa and Farber~\cite{cf1,cf2,cf3,cf4}.
Our analysis closely follows Kahle's work on random clique complexes~\cite[Section 7]{kahle2}.

The Costa--Farber complex $\cfm$  on a vertex set $V$ of size $n$ and a probability vector {$\mathbf{p}=\left\{ p_{1},\dots,p_{n-1}\right\} $} can be described as follows.
 First, add all the vertices in $V$ to the complex  $\cfm$. Next, include every
possible edge independently with probability $p_{1}$. So far, this
is the same as the Erd\H{o}s--R\'enyi graph $G(n,p_{1})$. 
That is, the $1$-skeleton $\mathsf{X}_{1} = G(n,p_{1})$.
Subsequently,
for every $3$-clique in $G(n,p_{1})$, include a $2$-simplex independently
with probability $p_{2}$ to obtain the $2$-skeleton $\mathsf{X}_{2}$.
More generally, consider an $r$-simplex $\sigma$ defined on an $r+1$-element
vertex set $V_{r}\subset V$. If all the simplices of the set $\partial\sigma$
are present in $\mathsf{X}_{r-1}$, then include $\sigma$ in $\mathsf{X}_{r}$
with probability $p_{r}$.  Do this for every for every $r+1$-element subset of $V$ to obtain the $r$-skeleton $\mathsf{X}_{r}$.
Following this process for every $r\in[n-1]$
gives  the complex $\mathsf{X}_{n-1}={\mathsf{X}(n,\textbf{p})}$.

Note that the Costa--Farber model generalizes both the $d$-dimensional Linial--Meshulam model $\mathsf{Y}_{d}(n,p)$ as well as the random clique complex model $\mathsf{X}(n,p)$. For instance, when $p_1=p$  and $p_i = 1$ for $i \in [2,n-1]$, we obtain the random clique complex model with parameter $p$.
When $p_i = 1$ for $i\in[d-1]$, $p_d = p$, and $p_i = 0$ for $i \in [d+1,n-1]$, we recover the $d$-dimensional Linial--Meshulam model with parameter $p$.

Let $\sigma=\left\{ v_{0},v_{1}\dots v_{r}\right\} $ be an $r$-dimensional simplex of $\cfm$, where the vertices are ordered such that $v_{0} < v_{1} < \dots <  v_{r}$.
A necessary condition for $\sigma$ to be a critical simplex in $\VCC_1$ is  that $\sigma$ has no lexicographically lower cofacet $\tau$
in $\mathsf{X}(n,\textbf{p})$. That is, if $\tau=\sigma\bigcup\left\{ v'\right\} $
and $ v' \prec_{K} v_i$ for all $v_i \in \sigma$, then $\tau\not\in\mathsf{X}(n,\textbf{p})$ (as otherwise $(\sigma,\tau)$ would be an apparent pair) 
Note that this condition is necessary but not sufficient for $\sigma$ to be a critical simplex, as it may still be matched to a lower dimensional simplex.

Using independence, the probability that $\sigma$ is not matched to a higher-dimensional simplex (and is possibly critical) is given by
\[\prod\limits_{i=1}^{r}p_{i}^{\binom{r+1}{i+1}}\left(1-\prod\limits_{\ell=1}^{r+1}p_{\ell}^{\binom{r+1}{\ell}}\right)^{v_{0} - 1} . \]
Let $m_{r}$ denote the total number of critical $r$-simplices. 
Assume without loss of generality that $V=\{1,\dots,n\}$.
Since
there are $\binom{n-j}{r}$ possible choices for $\sigma$ with $v_{0}=j$,
and since  $v_{0} \in [1,n-r]$, we obtain the following
expression for $\mathbb{E}(m_{r})$:
\begin{align*}
\mathbb{E}(m_{r}) &\leq  \sum\limits _{j=1}^{n-r}\binom{n-j}{r}\prod\limits _{i=1}^{r}p_{i}^{\binom{r+1}{i+1}}\left(1-\prod\limits _{\ell=1}^{r+1}p_{\ell}^{\binom{r+1}{\ell}}\right)^{j-1}\\
 & \leq\binom{n}{r}\prod\limits _{i=1}^{r}p_{i}^{\binom{r+1}{i+1}}\sum\limits_{j=1}^{n-r}\left(1-\prod\limits _{\ell=1}^{r+1}p_{\ell}^{\binom{r+1}{\ell}}\right)^{j-1}\\
 & \leq\binom{n}{r}\prod\limits _{i=1}^{r}p_{i}^{\binom{r+1}{i+1}}\sum\limits _{j=1}^{\infty}\left(1-\prod\limits _{\ell=1}^{r+1}p_{\ell}^{\binom{r+1}{\ell}}\right)^{j-1}\\
 & = \left (\binom{n}{r}\prod\limits _{i=1}^{r}p_{i}^{\binom{r+1}{i+1}} \right )  \cdot \left ( {\prod\limits _{\ell=1}^{r+1}p_{\ell}^{-\binom{r+1}{\ell}}} \right ) .
\end{align*}

Let $c_{r}$ denote the total number of $r$-simplices in $\mathsf{X}(n,\textbf{p})$.
Then, the expected number of $r$-simplices in $\mathsf{X}(n,\textbf{p})$ is given by
\[\mathbb{E}(c_{r})=\binom{n}{r+1}\prod\limits _{i=1}^{r}p_{i}^{\binom{r+1}{i+1}}.\]
Therefore, 
\begin{align*}
\frac{\mathbb{E}(m_{r})}{\mathbb{E}(c_{r})} & 
\leq \frac{ \left (\binom{n}{r}\prod\limits _{i=1}^{r}p_{i}^{\binom{r+1}{i+1}} \right )  \cdot \left ( {\prod\limits _{\ell=1}^{r+1}p_{\ell}^{-\binom{r+1}{\ell}}} \right )}{\binom{n}{r+1}\prod\limits _{i=1}^{r}p_{i}^{\binom{r+1}{i+1}}}
 =\frac{(r+1)}{(n-r)\prod\limits _{\ell=1}^{r+1}p_{\ell}^{\binom{r+1}{\ell}}} .
\end{align*}
When $r$ is a fixed constant, 
\begin{equation}\label{eqn:kahlerone}
\frac{\mathbb{E}(m_{r})}{\mathbb{E}(c_{r})}=O\left(\frac1n{\prod\limits _{\ell=1}^{r+1}p_{\ell}^{-\binom{r+1}{\ell}}}\right).
\end{equation}
Assuming the denominator
$n\prod\limits _{\ell=1}^{r+1}p_{\ell}^{\binom{r+1}{\ell}}\to\infty$,
we obtain $\frac{\mathbb{E}(m_{r})}{\mathbb{E}(c_{r})}=o(1)$.

\begin{remark}
While the asymptotics of the ratio  $\frac{\mathbb{E}(m_{r})}{\mathbb{E}(c_{r})} $ is informative, it would still be interesting to also say something about the ratio   $\mathbb{E} \big ( \frac{m_{r}}{c_{r}} \big )$. 
In particular, our results do not say anything about the variance of $m_r$ vis-a-vis the variance of $c_r$. Towards this end, the authors of \cite{TNR} show that for the apparent pairs gradient (referred to as the lexicographic gradient in~\cite{TNR}) on random clique complexes, the variance of the number of critical  $r$-simplices is of the same order as the  variance of the total number of $r$-simplices for random clique complexes, for fixed $r$. 
Furthermore, the authors of  \cite{TNR} prove a multivariate normal approximation theorem for the random vector  $(m_{2},m_{3},\dots,m_{r})$ for a fixed $r$. We recommend the results of \cite{TNR}  to the interested reader.
\end{remark}

\subsection{The random face gradient.}

Kahle~\cite[Section 7]{kahle2} describes an alternative method for designing gradients on 
random clique complexes with parameter $p_{1}$ for which the following
holds true.

\begin{equation}\label{eqn:kahler}
\frac{\mathbb{E}(m_{r})}{\mathbb{E}(c_{r})}=\frac{\binom{r+2}{2}\binom{n}{r+2}}{\binom{n}{r+1}}p_{1}^{r}.
\end{equation}

 We extend Kahle's strategy to the scenario where some of the $(r-1)$-simplices may be matched to $(r-2)$-dimensional facets and may not be available to be matched to $r$-dimensional cofacets. We call such simplices \emph{inadmissible simplices}. On the other hand, $(r-1)$-simplices that are not matched to their facets are called \emph{admissible simplices}. An admissible simplex along with a cofacet forms an \emph{admissible pair}, whereas an inadmissible simplex along with a cofacet forms an \emph{inadmissible pair}.  Randomly match
every $r$-simplex to one of its admissible facets. This strategy doesn't give
you a discrete gradient on the nose as there will be $(r-1)$-simplices
that are matched to more than one cofacets, and there might also
be some cycles. These events are termed as bad events. It suffices
to make one pair of simplices critical per bad event. Once the corresponding
simplices associated to all bad events are made critical, one indeed
obtains a discrete gradient $\VCC_2$. Bounding the expected number of bad
events $\mathcal{B}_{r}$ therefore gives a bound on the expected number
of critical simplices. It is then shown that the total number of bad
events for dimension $r$ is given by \[\mathbb{E}(\mathcal{B}_{r})  \leq \binom{r+2}{2}\binom{n}{r+2}p_{1}^{\binom{r+2}{2}-1}.\]
This is because each bad event contains at least one pair of $r$-simplices meeting in an admissible $(r-1)$-simplex. The total number of vertices involved are, therefore, $r+2$. 
So there are $\binom{n}{r+2}$ choices of $(r+2)$-vertex sets, and for every choice of an  $(r+2)$-vertex set, there are at most $\binom{r+2}{2}$ admissible pairs.
Finally, for an admissible pair to be a bad pair, all but one edge must be present among  the  $(r+2)$ vertices.
The expected number of simplices of dimension $r$ is given by \[\mathbb{E}(c_{r})=\binom{n}{r+1}p_{1}^{\frac{r(r+1)}{2}}.\]
Dividing the two, we obtain
\begin{equation}\label{eqn:fieldtwo}
\frac{\mathbb{E}(\mathcal{B}_{r})}{\mathbb{E}(c_{r})} \leq \frac{\binom{r+2}{2}\binom{n}{r+2}}{\binom{n}{r+1}}p_{1}^{r}.
\end{equation}
Note that if $r$ is a fixed constant, and   $np_{1}^{r}\to 0$, then
 \[\frac{\mathbb{E}(\mathcal{B}_{r})}{\mathbb{E}(c_{r})} = o(1).\]

Kahle's method~\cite{kahle2}  for constructing gradients in the sparse regime
easily extends to the Costa--Farber model, and \Cref{eqn:fieldtwo} generalizes as follows:
\[
\frac{\mathbb{E}(\mathcal{B}_{r})}{\mathbb{E}(c_{r})} \leq \frac{\binom{r+2}{2}\binom{n}{r+2}}{\binom{n}{r+1}}\prod\limits_{j=1}^{r}p_{j}^{\binom{r}{j}}.
\]
Let $r$ be a fixed constant, and $m_r$ denote the critical simplices of $\VCC_2$. Then, we obtain
\begin{equation}\label{eqn:kahlertwo}
\frac{\mathbb{E}(\mathcal{B}_{r})}{\mathbb{E}(c_{r})} =  \frac{\mathbb{E}(m_{r})}{\mathbb{E}(c_{r})} = O\left(n\prod_{j=1}^{r}p_{j}^{\binom{r}{j}}\right).
\end{equation}
When the parameters $\boldp$ in $\cfm$ are  such that
\[n\prod\limits _{j=1}^{r}p_{j}^{\binom{r}{j}}\to 0, \quad\text{we obtain}\quad \frac{\mathbb{E}(m_{r})}{\mathbb{E}(c_{r})}=o(1).\]
 
 
 \subsection{Linial--Meshulam complexes}
 
 From \Cref{eqn:kahlerone,eqn:kahlertwo}, in both cases, we obtain very good discrete gradients for typical instances. In particular, we obtain the following theorem.
 
 \begin{theorem} Let $r$ be a fixed dimension. Then, for the regimes of Costa--Farber complexes $\cfm$ that satisfy 
 \[
 n\prod\limits _{\ell=1}^{r+1}p_{\ell}^{\binom{r+1}{\ell}}\to\infty
 \quad\text{or}\quad
 n\prod\limits _{j=1}^{r}p_{j}^{\binom{r}{j}}\to 0 
 \]
 there exist respective discrete gradients that satisfy $\frac{\mathbb{E}(m_{r})}{\mathbb{E}(c_{r})}=o(1)$.
  \end{theorem}
  
 Specializing the above analysis to Linial--Meshulam complexes we obtain the following corollary. 
  \begin{corollary} \label{cor:lm} For the regimes  of Linial-Meshulam complexes $\mathsf{Y}_{d}(n,p)$ that  satisfy 
 \[
(n p \to\infty \text{ and } r+1= d)
 \quad\text{or}\quad
 (n p \to 0 \text{ and } r=d)
 \]
 there exist respective discrete gradients that satisfy $\frac{\mathbb{E}(m_{r})}{\mathbb{E}(c_{r})}=o(1)$.
  \end{corollary}
 
In other words, the above corollary says that if ($p=o(\frac{1}{n})$  and $r+1= d$), or if  ($p=\omega(\frac{1}{n})$  and $r=d$) for Linial--Meshulam complexes $\mathsf{Y}_{d}(n,p)$, then  there exist respective discrete gradients that satisfy $\frac{\mathbb{E}(m_{r})}{\mathbb{E}(c_{r})}=o(1)$.

To further refine our analysis, we now define a gradient $\UCC$ on the entire Linial--Meshulam complex $\lmm$ as follows:
\begin{compactitem} \item when   $n p \to\infty$,  let $\UCC$  be the apparent pairs gradient on $\lmm$.
 \item when $n p \to 0$,  $\UCC$  is comprised of 
 \begin{compactitem}
  \item $\text{the  apparent  pairs gradient for matching } k-1\text{-simplices to } k \text{-simplices for } k\in[d-1]$
  \item $ \text{the random face gradient for matching the admissible } (d-1) \text{-simplices to } d \text{-simplices.} $ 
  \end{compactitem}
\end{compactitem} 
Also, let $\mathsf{opt}$ be the optimal discrete gradient on $\lmm$. Then, the following holds true.
  \begin{theorem} \label{thm:lmmain} For the regimes  of Linial-Meshulam complexes $\mathsf{Y}_{d}(n,p)$ that  satisfy 
 \[
 n p \to\infty
 \quad\text{or}\quad
 n p \to 0 
 \]
 the discrete gradient $\UCC$  satisfies  $\frac{\mathbb{E}(|\UCC|)}{\mathbb{E}(|\mathsf{opt}|)} \to 1$ as $n\to \infty$.
  \end{theorem}

For proof, we refer the reader to \Cref{sec:lmc}.

We would like to contrast the above observations with a known result from literature. We  start with a definition.
If a $d$-dimensional simplicial complex collapses to a $d-1$-dimensional complex, then we say that it is \emph{$d$-collapsible}.
The following result concerning the $d$-collapsibility threshold was established by Aronshtam and Linial\cite{la1,la2}.
See also \cite[Theorem~23.3.17]{handbook}.
\begin{theorem}[Aronshtam, Linial\cite{la1,la2}] \label{thm:lmside}
There exists a dimension dependent constant $c_{d}$ for Linial--Meshulam
complexes $\mathsf{Y}_{d}(n,p)$ such that
\begin{itemize}
\item If $p\geq\frac{c}{n}$ where $c>c_{d}$ then with high probability
$\mathsf{Y}_{d}(n,p)$ is not $d$-collapsible, 
\item and if $p\leq\frac{c}{n}$ where $c<c_{d}$ then $\mathsf{Y}_{d}(n,p)$
is $d$-collapsible with probability bounded away from zero.
\end{itemize}
\end{theorem}

Therefore, from \Cref{thm:lmmain,thm:lmside}, we conclude that sufficiently away from the $d$-collapsibility threshold, we expect to have very good gradients. It is natural to ask what happens at the threshold? In relation to what is known for hard satisfiability instances~\cite{sat1,sat2}, are complexes of dimension larger than $2$ sampled at the collapsibility thresholds of $\mathsf{Y}_{d}(n,p)$ and more generally $\cfm$ hard? The experiments in \cite{RBN17} do not address this question.
Secondly, for $2$-complexes is it possible to define a simple random  model built out of gluing dunce hats  geared specifically towards generating hard instances for $\minmorse$ for a wide range of parameter values?
We are optimistic about affirmative answers to both questions, but leave this topic for future investigation.

\section{Conclusion and Discussion} \label{sec:extra} 
 In this paper, we establish several hardness results for $\minmorse$. In particular, we show that  for complexes of all dimensions, $\minmorse$ with standard parameterization is $\WP$-hard and has no FPT approximation algorithm for any approximation factor.
 We also establish novel (in)approximability bounds for  $\minmorse$ on $2$-complexes.
While we believe that this paper provides a nearly complete picture of complexity of Morse matchings,  we  conclude the paper with two remarks.

\paragraph*{Strengthening of hardness results} We conjecture that for complexes of dimension $d>2$, $\minmorse$ does not admit an $f(n)$-approximation algorithm for any $f  = o(n)$. 
In particular, a result of this nature would show that while the problem is hard for complexes of all dimensions, it is, in fact, slightly harder for higher dimensional complexes when compared to $2$-dimensional complexes, from an inapproximability standpoint.

\paragraph*{Hardness of other related combinatorial problems}
In \cite{BRS19}, the complexity of the following problem (\paraexptwo) was studied: Given a $2$-dimensional simplicial complex $K$ and a natural number $p$, does there exist a sequence of expansions and collapses that take $K$ to a $1$-complex such that this sequence has at most $p$ expansions? A more natural variant (\paraexpert) would be to study the complexity of determining sequences of expansions and collapses (with at most $p$ expansions) that take $K$ to a point. From what we understand, the only obstruction in \cite{BRS19} towards considering the complexity of determining whether $K$ is simple homotopy equivalent to a point with bounded number of expansions is that the gadgets used in  \cite{BRS19} have $1$-cycles. We believe that an immediate application of the cycle filling method introduced in this paper would be towards establishing  $\WP$-completeness for \paraexpert.

\paragraph*{Acknowledgements}
The authors would like to thank Vidit Nanda, Tadas Tem\v{c}inas and Gesine Reinert for  helpful discussions regarding Section 8 of this paper.
The authors  would also like to thank Benjamin Burton and Herbert Edelsbrunner for providing valuable feedback during the second author's thesis defense. 


\bibliography{morsetwocomplex}

\appendix


\section{Reducing $\mincircuitsat$ to $\minrmm$} \label{sec:appendreduceit}

\subsection{Structural properties of the reduction}\label{sec:spofr}

Note that \Cref{lem:erasefree,lem:remainerase}  appear as Lemmas~4.1~and~4.4~in~\cite{BR19}, but with slightly different notation. For the sake of completeness, we restate the lemma with the notation introduced in this paper. 



\begin{lemma} \label{lem:erasefree} For a circuit $C = (V(C),E(C))$, let $F$ is a discrete Morse function on $\kc$ with gradient $\VCC$:
\begin{enumerate}[(i)]
\item If $s_{i}^{\zeta} \in \gadget_{m,\ell}^{\zeta}$ is eventually free in $\kc$, then $\gadget_{m,\ell}^{\zeta}$ is erasable in  $\kc$. 
\item  Suppose that $\gadget_{m,\ell}^{\zeta}$ is erasable in $\kc$ through a gradient $\VCC$,
\begin{itemize}
\item If $m=1$, then  $\left( s_1^{\zeta} , \Gamma_1^{\zeta}  \right )$  is a gradient pair in $\VCC$, and for any simplex $\sigma^{\zeta} \in \gadget_{1,\ell}^{\zeta}$ such that $\sigma^{\zeta} \notin \{s_1^{\zeta}, \Gamma_1^{\zeta} \}$ we have  $F(s_1^{\zeta}) > F(\sigma^{\zeta})$. 
\item If $m=2$, then  $\left( s_1^{\zeta} , \Gamma_1^{\zeta}  \right ) \in \VCC$ or $\left( s_2^{\zeta} , \Gamma_2^{\zeta}  \right ) \in \VCC$, and for any simplex $\sigma^{\zeta} \in \gadget_{m,\ell}^{\zeta}$ such that $\sigma^{\zeta} \notin \{s_i^{\zeta}, \Gamma_i^{\zeta} \}$ for $i=\{1,2\}$, then we have, $\max(F(s_1^{\zeta}),F(s_2^{\zeta})) > F(\sigma^{\zeta})$.
\end{itemize}
\end{enumerate}        
\end{lemma}          
\begin{proof} 
Suppose $s_i^{\zeta}$ is eventually free in $\kc$. Then there exists a subcomplex $\LCC$ of $\kc$ such that $\kc \searrow \LCC$ and $s_i^{\zeta}$ is free in $\LCC$.
Note that, by construction of $\gadget_{m,\ell}$, this implies that $\gadget_{m,\ell}^{\zeta}$ is a subcomplex of $\LCC$.
Now using the gradient specified in \Cref{fig:gadgetpaths} all the 2-simplices of $\gadget_{m,\ell}^{\zeta}$ can be collapsed, making $\gadget_{m,\ell}^{\zeta}$ erasable in $\kc$. This proves the first statement of the lemma. 
The last two statements of the lemma immediately follows from observing that the $s$-edges are the only free edges in complex $\gadget_{m,\ell}^{\zeta}$, the simplices $\{\Gamma_i^{\zeta}\}$ are the unique cofaces incident on edges $\{s_i^{\zeta}\}$ respectively, and $\gadget_{m,\ell}^{\zeta}$ is erasable in $\kc$ through the gradient $\VSS$ of $F$.   
\end{proof}  

\begin{lemma} \label{lem:remainerase} For any input gate $G_i$, the subcomplex $\gadget^{(i,1)} \setminus \{\Gamma_1^{(i,1)} \}$ is erasable in $\kc$.
\end{lemma}
\begin{proof} 
Under the reindexing scheme described in~\Cref{rem:reindex}, let $\zeta_1$ be such that $\gadget^{(i,1)} = \gadget^{\zeta_1} $.
Consider the discrete gradient specified in \Cref{fig:gadgetpaths} $(a)$ as a gradient $\VCC^{(i,1)}$ on $\gadget^{(i,1)} \subseteq \kc$.
First note that $\gadget^{(i,1)} \setminus \{\Gamma^{(i,1)}\}$ is erasable in $\gadget^{(i,1)}$ through the gradient $\VCC^{(i,1)} \setminus \{(s_1^{(i,1)},\Gamma^{(i,1)})\}$.
Moreover, all 1-simplices of $\gadget^{(i,1)}$ that are paired in $\VCC^{(i,1)}$ with a $2$-simplex do not appear in $\gadget^{\zeta_2}$ for any edge $\zeta_1 \neq \zeta_2$.
It follows that $\gadget^{(i,1)} \setminus \{\Gamma_1^{(i,1)} \}$ is erasable in  $\kc$.
\end{proof}


  \subsection{Reducing $\mincircuitsat$ to $\minrmm$: Backward direction} \label{sec:algalg}
  
   We intend to establish an L-reduction from $\mincircuitsat$ to $\minrmm$.
To this end,  in~\Cref{sec:kcmain}~and~\Cref{sec:kc} we described the map $K: C \mapsto \kc$ that transforms instances of
$\mincircuitsat$ (monotone circuits) to instances of $\minrmm$ (simplicial complexes). 
In this section, we seek to construct a map $\ICC$ that transforms solutions of $\minrmm$ (discrete gradients $\VCC$ on $\kc$) to solutions of $\mincircuitsat$ (satisfying input assignments $\ICC(C,\VCC)$ of circuit $C$).
Recall that  $\ALGB$ denotes the objective value of some solution $\VCC$ on $\kc$ for $\minrmm$, whereas $\ALGA $ denotes the objective value of a solution $\ICC(C,\VCC) $ on $C$ for $\mincircuitsat$.

Suppose that we are  given a circuit $C = (V(C),E(C))$ with   $ n= |V(C)|$ number of nodes.
  Also, for a vector field $\tildev$  on $\kc$, we denote the critical simplices of dimension $2,1$ and $0$ by $m_2(\tildev), m_1(\tildev)$ and $m_0(\tildev)$ respectively.
 Then, by definition, 
 \begin{equation}\label{eq:mdefn}
 \ALGBB = m_2(\tildev) + m_1(\tildev) + m_0(\tildev) -1.
 \end{equation}
  In~\Cref{sec:relateopt}, we designed a gradient vector field $\VCC$ on $\kc$ with  $m_2(\VCC) = m, m_1(\VCC) = m$ and $m_0(\VCC) = 1$, for some $m \leq n$.
  We have from~\cite[Theorem 1.7]{Fo98},
\[m_0(\tildev) - m_1(\tildev) + m_2(\tildev) =  m_0(\VCC) - m_1(\VCC) + m_2(\VCC) .\] 
which gives $m_0(\tildev) - m_1(\tildev) + m_2(\tildev) =  1.$
 Since $m_0(\tildev) \geq 1$, this gives, for any vector field $\tildev$ on $\kc$, the following inequality 
 \begin{equation}\label{eq:ineqcrit}
 m_2(\tildev) \leq m_1(\tildev). 
 \end{equation}
 In particular, from~\Cref{eq:mdefn}~and~\Cref{eq:ineqcrit}, we obtain
 \begin{equation} \label{eq:mm}
 \ALGBB \geq 2  m_2(\tildev).
 \end{equation}
 
 Now, if $m_2(\tildev) \geq n$, we set $\ICC(C,\tildev)$  to be the set of all input gates of $C$. Clearly, this gives a satisfying assignment and using~\Cref{eq:mm} also satisfies
 \[\ALGBB  \geq 2\cdot \ALGAA. \]
 So, for the remainder of this section, we assume that $m_2(\tildev) < n$. In particular, for any non-output gate $G_i$ with $n$ blocks, at most $n-1$ of them may have critical $2$-simplices.

\begin{definition}[$2$-paired edges]
Given a  vector field $\VCC$ on a $2$-complex $K$, we say that an edge $e\in K$ is  \emph{$2$-paired} in $\VCC$ if it is paired to a $2$-simplex in $\VCC$.
\end{definition}

\begin{definition}[properly satisfied gates] \label{def:propgates}
Suppose that we are given a circuit $C$, and a vector field $\tildev$ on the associated complex $\kc$. Then,
\begin{enumerate}
\item an ordinary gate $G_q$  is said to be \emph{properly satisfied} if there exists a $\jin$ such that
\begin{itemize}
\item for an or-gate $G_q$  at least of the two edges $s_{\iota_1}^{\qjpair}, \, s_{\iota_2}^{\qjpair}$ is $2$-paired (in $\tildev$),  or 
\item for an and-gate  $G_q$ both the  edges $s_{\iota_1}^{\qjpair}, \, s_{\iota_2}^{\qjpair}$ are $2$-paired (in $\tildev$),
\item in both cases, the $j$-th block has no critical $2$-simplices;
\end{itemize}
\item an input gate $G_i$ is said to be \emph{properly satisfied} if the dunce hat associated to it contains at least one critical $2$-simplex,
\item the output gate $G_o$ is said to be \emph{properly satisfied} if $G_o$ is an or-gate and at least one of the two inputs gates of $G_o$ is properly satisfied, or if $G_o$ is an and-gate and both input gates of  $G_o$ are properly satisfied.
\end{enumerate}
\end{definition}

  \begin{lemma} \label{lem:propsat}
  Suppose that $G_k$ is a non-output gate that is properly satisfied. Then,
 \begin{enumerate}
 \item if $G_k$ is an or-gate, then at least one of the two gates that serve as  inputs to $G_k$ is also properly satisfied.
 \item if $G_k$ is an and-gate, then both gates that serve as  inputs to $G_k$ are also properly satisfied.
 \end{enumerate}
 \end{lemma}
 \begin{proof} 
 Assume without loss of generality that $G_k$ is an and-gate, and the two inputs that go into $G_k$, namely $G_{\ell}$ and $G_j$ are  both input gates.
 Since $G_{k}$ is properly satisfied, there exists $p \in [1,n]$ such that $s_{\iota_1}^{(k,p)}, \, s_{\iota_2}^{(k,p)}$ are $2$-paired and the $p$-th block of $G_{k}$ has no critical $2$-simplices.
 Now suppose that either $G_{\ell}$ or $G_j$  is not properly satisfied. For the sake of argument, suppose that $G_{\ell}$ is not properly satisfied.
That is, $\gadget^{(\ell,1)}$ has no critical $2$-simplices and $s_{f}^{(\ell,1)}$ is $2$-paired. Note that $s_{\iota_1}^{(k,p)}$ is identified to a $t$-edge in $\gadget^{(\ell,1)}$.
Using~\Cref{lem:erasefree}, we obtain $\tildef(s_{f}^{(\ell,1)}) > \tildef(s_{\iota_1}^{(k,p)})$. Since $^{3}s_1^{(k,p)}$ occurs as a $t$-edge in $^{1}\gadget^{(k,p)}$, using~\Cref{lem:erasefree}, we obtain 
$\tildef( s_{\iota_1}^{(k,p)}) > \tildef( ^{3}s_1^{(k,p)})$. Combining the two inequalities we obtain
 \begin{equation} \label{eq:first}
 \tildef(s_{f}^{(\ell,1)}) >\tildef( ^{3}s_1^{(k,p)})
 \end{equation}
 Moreover, $s_{f}^{(\ell,1)}$ is identified to a $t$-edge in $^{3}\gadget^{(k,p)}$, and by assumption  $^{3}\gadget^{(k,p)}$ has no critical $2$-simplices and, hence  $^{3}s_1^{(k,p)}$ is $2$-paired. Therefore, once again, using~\Cref{lem:erasefree},
 we obtain 
  \begin{equation} \label{eq:second}
    \tildef( ^{3}s_1^{(k,p)}) > \tildef(s_{f}^{(\ell,1)}) 
 \end{equation}
  Since \Cref{eq:first}~and~\Cref{eq:second} combine to give a contradiction, we conclude that  $G_{\ell}$ is properly satisfied.
All combinations of $G_k$ as  an \{and-gate, or-gate\}, and $G_{\ell}$ and $G_j$ as \{or-gates, input gates, and-gates\} give similar contradictions, proving the claim.
 \end{proof}
 
 \begin{lemma} \label{lem:propsatout}
 Given a Morse function  $\tildef$ on $\kc$ with vector field $\tildev$, the output gate $G_o$ is properly satisfied.
 \end{lemma}
 \begin{proof} 
 Assume without loss of generality that $G_o$ is an or-gate, and  the two inputs to $G_o$, namely $G_{\ell}$ and $G_j$ are non-input and-gates.
 Let $k \in [1,n]$ be such that the $k$-th copy of $G_o$ has no critical $2$-simplices. Such a $k$ exists because by assumption we have less than $n$ critical simplices.
 Now, suppose that neither $G_{\ell}$ nor $G_j$  is properly satisfied.
 
 Since, $G_{\ell}$ is not properly satisfied there exists a $p \in [1,n]$ such that either $s_{\iota_1}^{\lppair}$ or $s_{\iota_2}^{\lppair}$ is not $2$-paired and $p$-th block has no critical $2$-simplices (because by assumption we have less than $n$ critical simplces). Assume without loss of generality that $s_{\iota_1}^{\lppair}$ is not $2$-paired. Then, $s_{f_1}^{\lppair}$ is $2$-paired. 
 Using~\Cref{lem:erasefree}, we obtain $\tildef(s_{f_1}^{\lppair}) > \tildef(^{3}s_{1}^{\lppair})$. Now, $s_{\iota_1}^{(o,k)}$  is identified to a $t$-edge in $^{3}\gadget^{\lppair}$. So, using ~\Cref{lem:erasefree},
 we obtain  $ \tildef(^{3}s_{1}^{\lppair}) > \tildef(s_{\iota_1}^{(o,k)}).$ Combining the two inequalities, we obtain,
 \begin{equation} \label{eq:onelarge}
 \tildef(s_{f_1}^{\lppair}) > \tildef(s_{\iota_1}^{(o,k)}). 
 \end{equation}
 Similarly, there exists a $q \in [1,n]$ such that  either $s_{\iota_1}^{\jqpair}$ or $s_{\iota_2}^{\jqpair}$ is not $2$-paired. Assume without loss of generality that $s_{\iota_1}^{\jqpair}$ is not $2$-paired. Hence, we can show that 
 \begin{equation} \label{eq:twolarge}
 \tildef(s_{f_1}^{\jqpair}) > \tildef(s_{\iota_2}^{(o,k)}). 
 \end{equation}
 Combining~\Cref{eq:onelarge}~and~\Cref{eq:twolarge}, we obtain:
  \begin{equation} \label{eq:threelarge}
 \max (\tildef(s_{f_1}^{\lppair}), \tildef(s_{f_1}^{\jqpair})) >  \max (\tildef(s_{\iota_1}^{(o,k)}), \tildef(s_{\iota_2}^{(o,k)})). 
 \end{equation}
 But, $s_{f_1}^{\lppair}$ and $s_{f_1}^{\jqpair}$ appear as $t$-edges in $\gadget^{(o,k)}$. So, once again, using~\Cref{lem:erasefree}, 
\[ \max (\tildef(s_{\iota_1}^{(o,k)}), \tildef(s_{\iota_2}^{(o,k)})) >  \tildef(s_{f_1}^{\lppair})\text{ and } \max (\tildef(s_{\iota_1}^{(o,k)}), \tildef(s_{\iota_2}^{(o,k)})) >  \tildef(s_{f_1}^{\jqpair}), \]
 which combine to give:
 \begin{equation} \label{eq:fourlarge}
\max (\tildef(s_{\iota_1}^{(o,k)}), \tildef(s_{\iota_2}^{(o,k)})) > \max(\tildef(s_{f_1}^{\lppair}), \tildef(s_{f_1}^{\jqpair}))
 \end{equation} 
  Since \Cref{eq:threelarge}~and~\Cref{eq:fourlarge} combine to give a contradiction, we conclude that  $G_o$ is properly satisfied, and at least one of the two gates $G_{\ell}$ and $G_j$ are also  properly satisfied.
All combinations of $G_o$ as  an \{and-gate, or-gate\}, and $G_{\ell}$ and $G_j$ as \{or-gates, input gates, and-gates\} give similar contradictions, proving the claim.
Moreover, if $G_o$ is a properly satisfied and-gate, then both $G_{\ell}$ and $G_j$ are also be properly satisfied.
 \end{proof}
 
 Now, we construct the map $\ICC(C,\tildev) $ as follows: For every input gate $G_{\ell}$ whose associated dunce hat is properly satisfied, we set $\ICC(C,\tildev) (G_{\ell}) = 1$. 
That is, our assignment $\ICC(C,\tildev) (\cdot)$  ensures that an input gate is satisfied if and only if it is properly satisfied.

 \begin{claim}  \label{cl:allsat}
 With input assignment  $\ICC(C,\tildev) (\cdot)$, the circuit $C$ is satisfied. 
 \end{claim}
 \begin{proof}
 We prove the following claim inductively: Every gate of $C$ that is properly satisfied is also satisfied.
 
   To begin with, let $\prec_{C}$ be some total order on $V(C)$ consistent with the partial order imposed by $C$.
   Assume that the gates in  $C$ are indexed from $1$ to $|C|$ so that 
\[\text{ for all }G_i, G_j \in C,\quad i < j \Leftrightarrow G_i \prec_{C} G_j.\]

Let $\mathcal{P}$  denotes the set of properly satisfied gates. Let  $i_1, i_2, \dots i_{|\PCC|}$ denote the indices of the properly satisfied gates, where $i_k > i_{k-1}$ for all $k$.
 By repeated application of~\Cref{lem:propsat}, it follows that $G_{i_1}$ is an input gate.
 Then, from our construction of $\ICC(C,\tildev) (\cdot)$, we can conclude that $G_{i_1}$ is also satisfied, giving us the base case.
  
 Now, we make the inductive hypothesis that the gates $G_{i_1} \dots G_{i_{k-1}}$ are  satisfied.  
 Suppose that $G_{i_k}$ is an or-gate. Then, by~\Cref{lem:propsat}, one of the inputs to $G_{i_k}$, say $G_{i_j}$ is properly satisfied. 
 As a consequence of our indexing we have $j \in [1,k-1]$, and owing to the inductive hypothesis,  $G_{i_j}$ is  satisfied. 
 But, since $G_{i_k}$ is an or-gate, this implies that $G_{i_k}$ is also satisfied. 
 Suppose that $G_{i_k}$ be an and-gate. Then, by~\Cref{lem:propsat}, both the inputs to $G_{i_k}$, say $G_{i_j}, G_{i_p}$ are properly satisfied. 
 As a consequence of our indexing we have $j,p \in [1,k-1]$, and owing to the inductive hypothesis,  $G_{i_j}, G_{i_p}$ are  satisfied. 
 But, since $G_{i_k}$ is an and-gate, this implies that $G_{i_k}$ is also satisfied, completing the induction.
  
 Finally, using~\Cref{lem:propsatout}, the output gate is properly satisfied, and by the argument above it is also satisified.
 \end{proof}


 An immediate consequence of \Cref{cl:allsat}  is the following:
\begin{equation}\label{eq:mbigin}
m_2(\tildev)\geq \ALGAA
\end{equation}

\begin{proposition} \label{prop:mineqappend}
$\ALGBB  \geq 2\cdot \ALGAA   $
\end{proposition} 
 \begin{proof}
 This follows immmediately by combining~\Cref{eq:mm}~and~\Cref{eq:mbigin}.
 \end{proof}

Now, if the gradient vector field $\tildev$ is, in fact optimal for $\kc$, then $\tildev$ has a single critical $0$-simplex. That is, $m_0(\tildev) = 1$
Recall that  in~\Cref{sec:relateopt}, we designed a gradient vector field $\VCC$ on $\kc$ with  $m_2(\VCC) = m, m_1(\VCC) = m$ and $m_0(\VCC) = 1$, for some $m \leq n$.
 From~\cite[Theorem 1.7]{Fo98}, we have
\[m_0(\tildev) - m_1(\tildev) + m_2(\tildev) =  m_0(\VCC) - m_1(\VCC) + m_2(\VCC) .\]  
which gives us
\begin{equation}\label{eq:equal}
- m_1(\tildev) + m_2(\tildev)  = 0 
\end{equation}

From~\Cref{eq:equal}, we conclude that $\OPTB = 2 m_2(\tildev)$.

By~\Cref{eq:mbigin}, we have
\[m_2(\tildev)\geq \ALGAA.\]

Since by definition \[\ALGAA \geq \OPTA,\] we have the following proposition

\begin{proposition} \label{prop:opteq}
$\OPTB \geq 2 \OPTA$.
\end{proposition}

Combining \Cref{prop:appendoptrel} and \Cref{prop:opteq}, we obtain the following proposition.

\begin{proposition} \label{prop:equalityappend}
$\OPTB = 2 \OPTA$.
\end{proposition}

\section{Morse matchings for Linial--Meshulam complexes} \label{sec:lmc}

For Linial--Meshulam complexes $\lmm$, 
\begin{itemize} \item when   $n p \to\infty$,  let $\UCC$  be the apparent pairs gradient on $\lmm$.
 \item when $n p \to 0$,  $\UCC$  is comprised of 
 \begin{itemize}
  \item $\text{the  apparent  pairs gradient for matching } k-1\text{-simplices to } k \text{-simplices for } k\in[d-1]$
  \item $\text{the random face gradient for matching the remaining } (d-1) \text{-simplices to } d \text{-simplices.} $
  \end{itemize}
\end{itemize} 
The apparent pairs gradient, and the random face gradient are  described in \Cref{sec:costafarber}.

Let $V$ be the vertex set of $\mathsf{Y}_{d}(n,p)$, and $v'$ be the lexicographically lowest vertex of $V$.
For each $r\in[0,d]$, let $c_{r}$ denote the total number of $r$-dimensional
simplices in $\mathsf{Y}_{d}(n,p)$. Let $m_{r}$ denote the total
number of critical $r$-simplices of $\UCC$ and $\overline{m_{r}}$
be the total number of regular simplices of $\UCC$. Also, let $n_{r}$
and $\overline{n_{r}}$ denote the total number of critical $r$-simplices
and regular $r$-simplices respectively of the optimal discrete gradient
on $\mathsf{Y}_{d}(n,p)$. 

\begin{lemma} \label{lem:lower}
All the $k$-simplices of $\mathsf{Y}_{d}(n,p)$ for $k\in[0,d-2]$ are matched by $\UCC$. In particular, $\overline{n_{k}}=\overline{m_{k}}=c_{k}$ for $k\in[d-2]$, and $\overline{n_{0}}=\overline{m_{0}}= V -1$.
\end{lemma}
\begin{proof}
Let $\sigma$ be a $k$-simplex, where $k\in[d-2]$. If $v'\in\sigma$,
then $(\sigma\setminus\left\{ v\right\} ,\sigma)\in\UCC$, whereas
if $v'\not\in\sigma$, then $(\sigma,\sigma\cup\left\{ v'\right\} )\in\UCC$. 
That is for $k\in[0,d-2]$, $n_{k}=m_{k}=0$, and $\overline{n_{k}}=\overline{m_{k}}=c_{k}$.
Moreover, any vertex $v \neq v'$ is matched to the edge $\{v,v'\}$.
\end{proof}
Note that
\[
\frac{\mathbb{E}(|\UCC|)}{\mathbb{E}(|\mathsf{opt}|)}=\frac{\mathbb{E}(2|\UCC|)}{\mathbb{E}(2|\mathsf{opt}|)}=\frac{\mathbb{E}(\sum_{k=0}^{d}\overline{m_{k}})}{\mathbb{E}(\sum_{k=0}^{d}\overline{n_{k}})}
\]

\begin{theorem*}  For the regimes  of Linial-Meshulam complexes $\mathsf{Y}_{d}(n,p)$ that  satisfy 
 \[
 n p \to\infty
 \quad\text{or}\quad
 n p \to 0 
 \]
 the discrete gradient $\UCC$  satisfies  $\frac{\mathbb{E}(|\UCC|)}{\mathbb{E}(|\mathsf{opt}|)} \to 1$ as $n\to \infty$.
  \end{theorem*}

\begin{proof} We consider the following two cases:

\begin{description}

\item[Case 1]  $n p\to \infty$

\end{description}

By definition,
\begin{equation}\label{eqn:no1}
\overline{m_{d}} \geq c_{d-1}-m_{d-1}-c_{d-2}
\end{equation}
Also since the complex is $d$-dimensional,  we get
\begin{equation}\label{eqn:no2}
\overline{n_{d}}\leq c_{d-1}.
\end{equation}
Using \Cref{eqn:no1,eqn:no2,lem:lower}, we obtain
\begin{align*}
\mathbb{E}(\sum_{k=0}^{d}\overline{n_{k}}) & \leq\mathbb{E}(\sum_{k=0}^{d-1}c_{k}+c_{d-1}) = \sum_{k=0}^{d-1}c_{k}+c_{d-1}.
\end{align*}
\begin{align*}
\mathbb{E}(\sum_{k=0}^{d}\overline{m_{k}}) & \geq\mathbb{E}(\sum_{k=0}^{d-2}c_{k}+c_{d-1}-m_{d-1}+c_{d-1}-m_{d-1}-c_{d-2})\\
 & =\sum_{k=0}^{d-1}c_{k}+c_{d-1}-c_{d-2}-2\mathbb{E}(m_{d-1}).
\end{align*}

Therefore, 
\begin{align*}
\frac{\mathbb{E}(|\UCC|)}{\mathbb{E}(|\mathsf{opt}|)} & =\frac{\mathbb{E}(\sum_{k=0}^{d}\overline{m_{k}})}{\mathbb{E}(\sum_{k=0}^{d}\overline{n_{k}})}\\
 & \geq\frac{\sum_{k=0}^{d-1}c_{k}+c_{d-1}-c_{d-2}-2\mathbb{E}(m_{d-1})}{\sum_{k=0}^{d-1}c_{k}+c_{d-1}}\\
 & =1+\frac{-c_{d-2}-2\mathbb{E}(m_{d-1})}{\sum_{k=0}^{d-1}c_{k}+c_{d-1}}.
\end{align*}
Using \Cref{cor:lm}, and  the fact that in $\lmm$, $\frac{c_j}{c_k}\to 0$ for $j<k$ and $j,k\in[0,d-1]$, we conclude that 
\[\frac{\mathbb{E}(|\UCC|)}{\mathbb{E}(|\mathsf{opt}|)} \to 1.\]

\begin{description}

\item[Case 2] $n p\to 0$

\end{description}

Since a regular $(d-1)$-simplex is paired to either a $d$-simplex or a $(d-2)$-simplex, we obtain
\begin{equation} \label{eqn:no3}
\overline{n_{d-1}}\leq c_{d}+c_{d-2},
\end{equation}
\begin{equation} \label{eqn:no4}
\overline{m_{d-1}}\geq c_{d}-m_{d}+c_{d-2}-c_{d-3}.
\end{equation}
Therefore, using \Cref{eqn:no3,eqn:no4,lem:lower}, 
\begin{align*}
\frac{\mathbb{E}(|\UCC|)}{\mathbb{E}(|\mathsf{opt}|)} & =\frac{\mathbb{E}(\sum_{k=0}^{d}\overline{m_{k}})}{\mathbb{E}(\sum_{k=0}^{d}\overline{n_{k}})}\\
 & \geq\frac{\sum_{k=0}^{d-2}c_{k}+ (\mathbb{E}(c_d)-\mathbb{E}(m_{d})+c_{d-2}-c_{d-3})+ (\mathbb{E}(c_d)-\mathbb{E}(m_{d}))}{\sum_{k=0}^{d-2}c_{k}+\mathbb{E}(c_d)+c_{d-2}+\mathbb{E}(c_d)}\\
 & =1+\frac{-c_{d-3}-2\mathbb{E}(m_{d})}{\sum_{k=0}^{d-2}c_{k}+c_{d-2}+2\mathbb{E}(c_d)}.
\end{align*}
Using \Cref{cor:lm}, and  the fact that in $\lmm$, $\frac{c_j}{c_k}\to 0$ for $j<k$ and $j,k\in[0,d-1]$, we conclude that 
\[\frac{\mathbb{E}(|\UCC|)}{\mathbb{E}(|\mathsf{opt}|)} \to 1.\]

\end{proof}

\end{document}